\PassOptionsToPackage{prologue,dvipsnames}{xcolor}  

\documentclass[format=acmsmall, review=false]{acmart}

\usepackage{acm-ec-25}
\usepackage{booktabs} 
\usepackage[ruled]{algorithm2e} 

\SetAlFnt{\small}
\SetAlCapFnt{\small}
\SetAlCapNameFnt{\small}
\SetAlCapHSkip{0pt}
\IncMargin{-\parindent}

\usepackage{cleveref}

\usepackage[bibliography=common]{apxproof}

\renewcommand{\appendixprelim}[1]{%
}

\newtheoremrep{theorem}{Theorem}[section]
\newtheoremrep{proposition}[theorem]{Proposition}
\newtheoremrep{lemma}[theorem]{Lemma}
\newtheoremrep{claim}[theorem]{Claim}
\newtheoremrep{observation}[theorem]{Observation}
\newtheoremrep{open}[theorem]{Open Question}

\theoremstyle{definition}
\newtheoremrep{remark}[theorem]{Remark}

\ifdefined\DRAFT
    \newcommand{\er}[1]{\textcolor{blue}{#1}}
    \newcommand{\erel}[1]{\er{(Erel says: #1)}}
    \newcommand{\alex}[1]{\textcolor{red}{(Alex says: #1)}}
    \newcommand{\eden}[1]{\textcolor{BrickRed}{(Eden says: #1)}}
    \newcommand{\edeng}[1]{\textcolor{OliveGreen}{(Eden says: #1)}}
    
    \newcommand{\rmark}[1]{\color{BrickRed}#1~\color{black}}
    \newcommand{\biaoshuai}[1]{\textcolor{purple}{(Biaoshuai says: #1)}}
\else
    \newcommand{\er}[1]{#1}
    \newcommand{\erel}[1]{}
    \newcommand{\alex}[1]{}
    \newcommand{\eden}[1]{}
    \newcommand{\edeng}[1]{}
    \newcommand{\rmark}[1]{#1}
    \newcommand{\biaoshuai}[1]{}
\fi

\newcommand{\valT}[2]{v_{#1,#2}}
\newcommand{\repT}[2]{r_{#1,#2}}

\newcommand{\prefs}{\mathbf{P}}
\newcommand{\prefsExcI}{\prefs_{-i}}
\newcommand{\domain}{\mathcal{D}}
\newcommand{\domains}{\mathbf{\domain}}
\newcommand{\domainsExcI}{\domains_{-i}}
\newcommand{\notK}{\widebar{K}}

\usepackage{wrapfig}

\DeclareMathOperator*{\argmax}{arg\,max}

\newcommand{\ceil}[1]{\left\lceil #1 \right\rceil}
\newcommand{\floor}[1]{\left\lfloor #1 \right\rfloor}

\newcommand{\propallocations}{\mathcal{X}^{\mathrm{PROP}}}
\newcommand{\efallocations}{\mathcal{X}^{\mathrm{EF1}}}
\newcommand{\yefallocations}{\mathcal{Y}^{\mathrm{EF1}}}

\usepackage{multirow}
\usepackage{array}
\usepackage{tabulary}
\newcolumntype{K}[1]{>{\centering\arraybackslash}m{#1}}

\usepackage{tikz}
\usetikzlibrary{decorations.pathreplacing,calligraphy}
\usetikzlibrary{patterns}

\ifdefined\COMSOC
    \newcommand{\hideComsoc}[1]{}
    \newcommand{\onlyComsoc}[1]{#1}
    \newcommand{\extendedVer} {\href{https://arxiv.org/abs/2502.18805}{extended version }}
    
\else
    \newcommand{\hideComsoc}[1]{#1}
    \newcommand{\onlyComsoc}[1]{}
    \newcommand{\extendedVer}{}
    
\fi

\usepackage{hhline}

\newcommand{\Hquad}{\hspace{0.5em}}

\newcommand{\bv}{\mathbf{v}}

\usepackage{enumitem}

\setcitestyle{authoryear}

\title[Degree of Risk-Avoiding Truthfulness]{It's Not All Black and White: Degree of Truthfulness for Risk-Avoiding Agents}

\author{Eden Hartman}
\affiliation{%
	\institution{Bar-Ilan University}
	\country{Israel}
}

\author{Erel Segal-Halevi}
\affiliation{%
	\institution{Ariel University}
	\country{Israel}
}

\author{Biaoshuai Tao}
\affiliation{%
	\institution{Shanghai Jiao Tong University}
	\country{China}
}

\begin{abstract}
The classic notion of \emph{truthfulness} requires that no agent has a profitable manipulation — an untruthful report that, for \emph{some} combination of reports of the other agents, increases her utility. 
This strong notion implicitly assumes that the manipulating agent either knows what all other agents are going to report, or is willing to take the risk and act as-if she knows their reports.

Without knowledge of the others' reports, most manipulations are \emph{risky} -- they might decrease the manipulator's utility for some other combinations of reports by the other agents.
Accordingly, a recent paper (Bu, Song and Tao, ``On the existence of truthful fair cake cutting mechanisms'', Artificial Intelligence 319 (2023), 103904) suggests a relaxed notion, which we refer to as \emph{risk-avoiding truthfulness (RAT)}, which requires only that no agent can gain from a \emph{safe} manipulation — one that is sometimes beneficial and never harmful.

Truthfulness and RAT are two extremes: the former considers manipulators with complete knowledge of others, whereas the latter considers manipulators with no knowledge at all. 
In reality, agents often know about some — but not all — of the other agents.
This paper introduces the \emph{RAT-degree} of a mechanism, 
defined as the smallest number of agents whose reports, if known, may allow another agent to safely manipulate, or $n$ if there is no such number. 
This notion interpolates between classic truthfulness (degree $n$) and RAT (degree at least $1$): a mechanism with a higher RAT-degree is harder to manipulate safely. 

To illustrate the generality and applicability of this concept, we analyze the RAT-degree of prominent mechanisms across various social choice settings, including auctions, indivisible goods allocations, cake-cutting, voting, and two-sided matching. 
\end{abstract}

\begin{document}

\maketitle

\setcounter{tocdepth}{1} 
\tableofcontents


\section{Introduction}

The Holy Grail of mechanism design is the \emph{truthful mechanism} --- a mechanism in which the (weakly) dominant strategy of each agent is truthfully reporting her type.
But in most settings, there is provably no truthful mechanism that satisfies other desirable properties such as budget-balance, efficiency or fairness. 
%
Practical mechanisms are thus \emph{manipulable} in the sense that some agent~$a_i$ has a \emph{profitable} manipulation -- for \emph{some} combination of reports by the other agents, agent~$a_i$ can induce the mechanism to yield an outcome (strictly) better for her by reporting non-truthfully.

%
%

This notion of manipulability implicitly assumes that the manipulating agent either knows the reports made by all other agents, or is willing to take the risk and act \emph{as-if} she knows their reports. 
Without knowledge of the others' reports, most manipulations are \emph{risky} -- they might decrease the manipulator's utility for some other combinations of reports by the other agents.
In practice, many agents are \emph{risk-avoiding} and will not manipulate in such cases. 
This highlights a gap between the standard definition and the nature of such agents. 
%
%

To illustrate, consider a simple example in the context of voting. Under the Plurality rule, agents vote for their favorite candidate, and the candidate with the most votes wins. If an agent knows that her preferred candidate has no chance of winning, she may find it beneficial to vote for her second-choice candidate to prevent an even less preferred candidate from winning. However, if the agent lacks precise knowledge of the other votes and decides to vote for her second choice, it may backfire -- she might inadvertently cause an outcome worse than if she had voted truthfully. For instance, if the other agents vote in a way that makes the agent the tie-breaker.

Indeed, various papers on cake-cutting
(e.g. \cite{brams2006better,BU2023Rat}),
voting (e.g. \cite{slinko2008nondictatorial,slinko2014ever,hazon2010complexity})
 stable matching (e.g. \citet{regret2018Fernandez, chen2024regret})
 and coalition formation \citep{waxman2021manipulation}
studies truthfulness among such agents.
%
In particular, \citet{BU2023Rat} introduced a weaker notion of truthfulness, suitable for risk-avoiding agents, for cake-cutting. 
Their definition can be adapted to any problem as follows.


Let us first define a \emph{safe manipulation} as a non-truthful report that 
may never harm the agent's utility. 
Based on that, a mechanism is \emph{safely manipulable} if some agent~$a_i$ has a manipulation that is both profitable and safe; otherwise, the mechanism is Risk-Avoiding Truthful (RAT).

Standard truthfulness and RAT can be seen as two extremes with respect to safe-and-profitable manipulations: the former considers manipulators with complete knowledge of
others, whereas the latter considers manipulators with no knowledge at all. 
In reality, agents often know about some — but not all — of the other agents.

This paper introduces the \emph{RAT-Degree} — a new measurement that quantifies how robust a mechanism is to such safe-and-profitable manipulations.
The RAT-Degree of a mechanism is an integer $d \in \{0, \ldots, n\}$, which represents --- roughly --- 
the smallest number of agents whose reports, if known, may allow another agent to safely manipulate; or $n$ if there is no such number.   (See \Cref{sec:RAT-degree} for formal definition).

This measure allows us to position mechanisms along a spectrum. A higher degree implies that an agent has to work harder in order to collect the information required for a successful manipulation; therefore it is less likely that mechanism will be manipulated.
%
%
On one end of the spectrum are truthful mechanisms -- where no agent can safely manipulate even with complete knowledge of all the other agents. The RAT-degree of such mechanisms is $n$.
While on the other end are mechanisms that are safely manipulable -- no knowledge about other agents is required to safely manipulate. The RAT-degree of such mechanisms is $0$.

Importantly, the RAT-degree is determined by the worst-case scenario for the mechanism designer, which corresponds to the best-case scenario for the manipulating agents.
The way we measure the amount of knowledge is based on a general objective applicable to all social choice settings.

\paragraph{Contributions.}
Our main contribution is the definition of the RAT-degree.

To illustrate the generality and usefulness of this concept, we selected several different social choice domains, and analyzed the RAT-degree of some prominent mechanisms in each domain.
As our goal is mainly to illustrate the new concepts, we did not attempt to analyze all mechanisms and all special cases of each mechanism, but rather focused on some cases that allowed for a more simple analysis. 
To prove an upper bound on the RAT-degree, we need to show a profitable manipulation. However, in contrast to usual proofs of manipulability, we also have to analyze more carefully, how much knowledge on other agents is sufficient in order to guarantee the safety of the manipulation.
This analysis gives more insight on the kind of manipulations possible in each mechanism, and on potential ways to avoid them.
For clarity, we detail the results for each social choice setting in its corresponding section.


\paragraph{Organization.} \Cref{sec:preliminaries} introduces the model and required definitions. \Cref{sec:RAT-degree} presents the definition of RAT-degree. 
\Cref{sec:single-item-auction} explores auctions for a single good. \Cref{sec:indivisible-good-aloc} examines indivisible goods allocations. 
\Cref{sec:cake-cutting} focuses on cake cutting.  
\Cref{sec:single-winner-voting} addresses single-winner ranked voting.
\Cref{sec:matching} considers two-sided matching.
\Cref{sec:discussion} concludes with some future work directions.

To enhance readability, we include only intuitive proof sketches in the main body of the paper. Complete and rigorous versions of all proofs are provided in the appendix.

\subsection{Related Work}\label{sec:related}

There is a vast body of work on truthfulness relaxations and alternative measurements of manipulability.

\subsubsection{Truthfulness Relaxations}

The large number of impossibility results have lead to extensive research into relaxations of truthfulness. 
A relaxed truthfulness notion usually focuses on a certain subset of all possible manipulations, which are considered more ``likely''. It requires that none of the manipulations from this subset is profitable. Different relaxations consider different subsets of ``likely'' manipulations.
Our approach can be combined with most of these ideas by adapting the degree concept to the relevant manipulation style.


\paragraph{RAT}
Closest to our work is the recent paper by \cite{BU2023Rat}, which introduces the definition on which this paper is build upon: \emph{risk-avoiding truthfulness (RAT)}. 
The definition assumes that agents avoid risk - they will manipulate only when it is sometimes beneficial but never harmful. 
We first note that their original term for this concept is risk-\emph{averse} truthfulness. However, since the definition assumes that agents completely avoid any element of risk, we adopt this new name, aiming to more accurately reflect this assumption.

We extend their work in two key directions.
First, we generalize the definition from cake-cutting to any social choice problem. 
Second, we move beyond a binary classification of whether a mechanism is RAT, to a quantitative measure of its robustness to manipulation by such agents. Our new definition provides deeper insight of the mechanism' robustness.
In \Cref{sec:cake-cutting}, we also analyze the mechanisms proposed in their paper alongside additional mechanisms for cake-cutting.

Their paper (in Section~5) provides an extensive comparison between RAT of related relaxation of truthfulness. For completeness, we include some of these comparisons here as well.

\rmark{\paragraph{Dominating Manipulations with Partial Information.}
\citet{conitzer2011dominating} studied a closely related model in the context of voting, where the manipulator has only partial information about the preferences of the other (non-manipulating) voters and seeks a dominating manipulation, defined analogously to the profitable and safe manipulations we consider in this paper. 

There are two main differences.
First, their work focuses on the \emph{computational complexity} of finding a dominating manipulation, while our work focuses on existence proofs.
Second, their paper considers different information structures about the other agents --- we assume that the manipulator knows \emph{exactly} the preferences of \emph{some subset} of voters, whereas they assume \emph{partial knowledge} about \emph{all} other voters.
These two approaches can be combined: one could study the computational complexity under our information structure, or conversely, prove existence results under theirs.

As part of their analysis, they considered the special case in which the manipulator has no information about the other voters and showed that, for several voting rules, no dominating manipulation exists --- in our terms, that the RAT-degree is positive. We extend some of these results in \Cref{sec:single-winner-voting}.

There has been a line of research following this paper --- for example, \citet{Reijngoud2012} and \citet{endriss2016strategic}.
To the best of our knowledge, however, these studies focus specifically on voting settings, and the information structures they consider are tailored to that domain.
In contrast, our measure of the number of known agents can be applied to any social choice problem, as illustrated later in the paper.

\paragraph{Veto Power.}
\citet{Moulin1982vetoPower} introduced the \emph{veto power} for voting: a coalition of voters of size $t$ is said to have veto power $v(t)$ if it can prevent any set of $v(t)$ candidates from winning, regardless of the preferences of the remaining voters.

We conjectured that this notion could provide an upper bound on the RAT-degree: our hypothesis was that knowing that one of the candidate has no chance of winning is sufficient to enable a safe and profitable manipulation for a manipulator for \emph{some} candidate.
This indeed holds for Plurality and Veto rules (if the eliminated candidate is the manipulator’s favorite under Plurality, and the least favorite under Veto; see \Cref{sec:single-winner-voting}).
Therefore, we expected that the minimal size of a coalition with positive veto power would serve as an upper bound on the RAT-degree.
However, this claim appears to fail in general: eliminating a single candidate may not suffice for some manipulation to be both profitable and safe. A counterexample and several partial connections are presented in \Cref{sec:voting-RAT-vs-veto-power}.
}


\paragraph{Maximin Strategy-Proofness}
Another related relaxation of truthfulness, introduced by \citet{brams2006better}, assumes the opposite extreme to the standard definition regarding the agents' willingness to manipulate. In the standard definition, an agent is assumed to manipulate if for \emph{some} combination of the other agents' reports, it is beneficial. In contrast, this relaxation assumes that an agent will manipulate only if the manipulation is \emph{always} beneficial — i.e., for \emph{any} combination of the other agents' reports. 

This definition is not only weaker than the standard one but also weaker than RAT, as it assumes that agents manipulate in only a subset of the cases where RAT predicts manipulation. 
We believe that, in this sense, RAT provides a more realistic and balanced assumption.
However, we note that a similar approach can be applied to this definition as well —- the degree to which a mechanism is robust to always-profitable manipulations (rather than to safe manipulations). The degree of Maximin Strategy-Proofness would probably be higher, as the manipulation is stronger.


\paragraph{NOM} 
\citet{troyan2020obvious} introduce the notion of \emph{not-obvious manipulability (NOM)}, which focuses on \emph{obvious manipulation} — informally, a manipulation that benefits the agent either in the worst case or in the best case. It presumes that agents are boundedly rational, and only consider extreme situations --- best or worst cases. A mechanism is NOM if there are no obvious manipulations. 
This notion is a relaxation of truthfulness since the existence of an obvious manipulation imposes a stronger requirement than merely having a profitable one.

However, an obvious manipulation is not necessarily a \emph{safe} manipulation, nor is a safe manipulation necessarily an obvious one (see Appendix D in \cite{BU2023Rat} for more details). Therefore, NOM and RAT are independent notions.

Our approach can be applied to NOM as well.


\paragraph{Regret-Free Truth-telling (RFTT)} \citet{regret2018Fernandez} proposes another relaxation of truthfulness, called \emph{regret-free truth-telling}. This concept also considers the agent’s knowledge about other agents but does so in a different way.
Specifically, a mechanism satisfies RFTT if no agent ever regrets telling the truth after observing the outcome -- meaning that she could not have safely manipulated given that only the reports that are consistent with the observed outcome are possible.
Note that the agent does not observe the actual reports of others but only the final outcome (the agents are assumed to know how the mechanism works).


An ex-post-profitable manipulation is always profitable, but the opposite is not true, as it is possible that the profiles in which a manipulation could be profitable are not consistent with the outcome. This means that RFTT does not imply RAT.
A safe manipulation is always ex-post-safe, but the opposite is not true, as it is possible that the manipulation is harmful for some profiles that are inconsistent with the outcome. This means that RAT does not imply RFTT.

\erel{Maybe say that we can quantify this by allowing the agent to play many times with different reports.}

\paragraph{Different Types of Risk}
\citet{slinko2008nondictatorial,slinko2014ever} and \citet{hazon2010complexity} study ``safe manipulations'' in voting. Their concept of safety is similar to ours, but allows a simultaneous manipulation by a coalition of voters. They focus on a different type of risk for the manipulators: the risk that too many or too few of them will perform the exact safe manipulation.

\citet{waxman2021manipulation} study safe manipulations \footnote{
They use ``safe manipulation'' for a manipulation that is both safe and profitable.
} on coalitional games in social networks. They examine the possible manipulations of agents with three different levels of knowledge on the social networks.
\erel{Interesting and closely-related paper; should read in depth. Maybe also send our paper to the authors.}

\subsubsection{Alternative Measurements}
Various measures can be used to quantify the manipulability of a mechanism. Below, we compare some existing approaches with our RAT-degree measure.

\paragraph{Computational Complexity}
Even when an agent knows the reports of all other agents, it might be hard to compute a profitable manipulation. 
Mechanisms can be ranked according to the run-time complexity of this computation: mechanisms in which a profitable manipulation can be computed in polynomial time are arguably more manipulable than mechanisms in which this problem is NP-hard. 
This approach was pioneered by \citet{bartholdi1989computational,bartholdi1991single} for voting, and applied in diverse social choice settings, e.g. coalitional games in social networks \citep{waxman2021manipulation}. See \citet{faliszewski2010ai,veselova2016computational} for surveys.

Nowadays, with the advent of efficient SAT and CSP solvers, NP-hardness does not seem a very good defense against manipulation. Moreover, some empirical studies show that voting rules that are hard to manipulate in theory, may be easy to manipulate in practice  \citep{walsh2011hard}.
We claim that the main difficulty in manipulation is not the computational hardness, but rather the informational hardness --- learning the other agents' preferences.

\paragraph{Queries and Bits}
Instead of counting the number of agents, we could count the number of \emph{bits} that an agent needs to know in order to have a safe manipulation
(this is similar in spirit to the concepts of communication complexity - e.g., \cite{nisan2002communication, grigorieva2006communication, Communication2019Branzei,Babichenko2019communication} and compilation complexity - e.g., \citep{chevaleyre2009compiling,xia2010compilation,karia2021compilation}).
But this definition may be incomparable between different domains, as the bit length of the input is different.
We could also measure the number of basic queries, rather than the number of agents. But then we have to define what ``basic query'' means, which may require a different definition for each setting. The number of agents is an objective measure that is relevant for all settings.

\paragraph{Probability of Having Profitable Manipulation.} 
Truthfulness requires that not even a \emph{single} preference-profile allows a profitable manipulation.
Mechanisms can be ranked according to the probability that such a profile exists -- e.g., \citep{barrot2017manipulation,lackner2018approval,lackner2023free}.
The probability of truthful mechanism is zero, and the lower the probability is, the more resistant the mechanism is to profitable manipulations.

A disadvantage of this approach is that it requires knowledge of the distribution over preference profiles; our RAT-degree does not require it.
However, if we do have such information, it can be combined with our known-agents approach.

\paragraph{Incentive Ratio} Another common measure is the \emph{incentive ratio} of a mechanism \citet{chen2011profitable}, which describes the extent to which an agent can increase her utility by manipulating. For each agent, it considers the maximum ratio between the utility obtained by manipulating and the utility received by being truthful, given the same profile of the others. The incentive ratio is then defined as the maximum of these values across all agents.

This measure is widely used - e.g., \cite{chen2022incentive,li2024bounding,cheng2022tight,cheng2019improved,bei2025incentive,tao2024fair}.
However, it is meaningful only when utility values have a clear numerical interpretation, such as in monetary settings.
In contrast, our RAT-degree measure applies to any social choice setting, regardless of how utilities are represented.

\paragraph{Degree of Manipulability}
\citet{aleskerov1999degree} define four different indices of "manipulability" of ranked voting rules, based on the fraction of profiles in which a profitable manipulation exists; the fraction of manipulations that are profitable; and the average and maximum gain per manipulation.
As the indices are very hard to compute, they estimate them on 26 voting rules using computer experiments. These indices are useful under an assumption that all profiles are equally probable (``impartial culture''), or assuming a certain mapping between ranks of candidates and their values for the agent.

\citet{andersson2014budget,andersson2014least} measure the degree of manipulability by counting the number of agents who have a profitable manipulation. They define a rule as \emph{minimally manipulable} with respect to a set of acceptable mechanisms if for each preference profile, the number of agents with a profitable manipulation is minimal across all acceptable mechanisms.

\eden{TODO: add a comparison to RAT-degree}
\erel{This can also be combined with our approach.}

\erel{Maybe cite \citet{peters2022robust}} 
\section{Preliminaries}\label{sec:preliminaries}

We consider a generic social choice setting, with a set of $n$ \emph{agents} $N = \{a_1, \ldots, a_n\}$, and a set of potential \emph{outcomes} $X$.
%
%
Each agent, $a_i \in N$, has preferences over the set of outcomes $X$, that can be described in one of two ways: (a) a linear ordering of the outcomes, or (b) a utility function from $X$ to $\mathbb{R}$.
The set of all possible preferences for agent~$a_i$ is denoted by $\domain_i$, and is referred to as the agent's \emph{domain}. 
We denote the agent's \emph{true} preferences by $T_i \in \domain_i$.
Unless otherwise stated, when agent~$a_i$ weakly prefers the outcome $x_1$ over $x_2$, it is denoted by $x_1 \succeq_i x_2$; and when she strictly prefers $x_1$ over $x_2$, it is denoted by $x_1 \succ_i x_2$.



A \emph{mechanism} or \emph{rule}  is a function $f: \domain_1\times\cdots \times \domain_n \to X$, which takes as input a list of reported preferences $P_1,\ldots,P_n$ (which may differ from the true preferences), and returns the chosen outcome.
In this paper, we focus on deterministic and single-valued mechanisms. 

For any agent $a_i \in N$, we denote by $(P_i, \prefsExcI)$ the preference profile in which agent~$a_i$ reports $P_i \in \mathcal{D}_i$, and the other agents report $\prefsExcI \in \mathcal{D}_{-i}$ (where $\mathcal{D}_{-i}:= \displaystyle \times_{j \in N\setminus\{i\}} \mathcal{D}_j$).


\paragraph{Truthfulness.}
A \emph{manipulation} for a mechanism $f$ and agent~$a_i \in N$ is an untruthful report $P_i \in \domain_i \setminus \{T_i\}$.
A manipulation is \emph{profitable} if there exists a set of preferences of the other agents for which it increases the manipulator's utility:
\begin{align}
\label{eq:manipulation}
    &\exists \prefsExcI
    \in \domainsExcI 
    :~ f(P_i,\prefsExcI) \succ_i f(T_i,\prefsExcI)
\end{align}
A mechanism $f$ is called \emph{manipulable} if some agent $a_i$ has a profitable manipulation; otherwise $f$ is called \emph{truthful}.

\paragraph{RAT}
A manipulation is \emph{safe} if it never harms the manipulator's utility -- it is weakly preferred over telling the truth for any possible preferences of the other agents:
\begin{align}
\label{eq:safe-manipulation}
    &\forall \prefsExcI
    \in \domainsExcI 
    :~ f(P_i,\prefsExcI) \succeq_i f(T_i,\prefsExcI)
\end{align}
%
%
%
%
%
A mechanism $f$ is called \emph{safely-manipulable} if some agent~$a_i$ has a manipulations that is profitable and safe; otherwise $f$ is called \emph{risk-avoiding truthful (RAT)}.

\section{The RAT-Degree}\label{sec:RAT-degree}

\newcommand{\prefsOf}[1]{\mathbf{P}_{#1}}
\newcommand{\domainsOf}[1]{\mathbf{D}_{#1}}

Let $a_i \in N$, $k \in \{0,\ldots, n-1\}$, $K \subseteq N \setminus \{a_i\}$ with $|K| = k$ and $\notK := N \setminus (\{a_i\} \cup K)$.
We denote by $(P_i, \prefsOf{K}, \prefsOf{\notK})$ the preference profile in which the preferences of agent $a_i$ are $P_i$, the preferences of the agents in $K$ are $\prefsOf{K}$, and the preferences of the agents in $\notK$ are $\prefsOf{\notK}$.

\begin{definition}
\label{def:given-K}
Given an agent $a_i$, a subset $K\subseteq N\setminus \{a_i\}$ and preferences for them $\prefsOf{K}
            \in \domain_{K}$:

A manipulation $P_i$ is called \emph{profitable for $a_i$ given $K$ and $\prefsOf{K}$} if
    \begin{align}
    \label{eq:k-manipulation}
            \quad \exists \prefsOf{\notK}
            \in \domain_{\notK}
            :~ f(P_i,\prefsOf{K}, \prefsOf{\notK}) \succ_i f(T_i,\prefsOf{K}, \prefsOf{\notK})
    \end{align}

A manipulation $P_i$ is called \emph{safe for $a_i$ given $K$  and $\prefsOf{K}$} if
    \begin{align}
    \label{eq:k-safe-manipulation}
            \forall \prefsOf{\notK}
            \in \domain_{\notK}
            :~ f(P_i,\prefsOf{K}, \prefsOf{\notK}) \succeq_i f(T_i,\prefsOf{K}, \prefsOf{\notK})
    \end{align}
\end{definition}

In words: The agents in $K$ are those  whose preferences are \emph{Known} to $a_i$; the agents in $\notK$ are those whose preferences are unknown to $a_i$.
Given that the preferences of the known agents are $\prefsOf{K}$, \Cref{eq:k-manipulation} says that there exist a preference profile of the unknown agents that makes the manipulation profitable for agent~$a_i$; while \Cref{eq:k-safe-manipulation} says that the manipulation is safe -- it is weakly preferred over telling the truth for any preference profile of the unknown-agents.

The previous two definitions are special cases of \Cref{def:given-K}: \Cref{eq:manipulation} --- which defines a profitable manipulation --- is equivalent to $P_i$ being profitable given $\emptyset$; and \Cref{eq:safe-manipulation} --- which defines a safe manipulation --- is equivalent to $P_i$ being safe given $\emptyset$.

%
%

%

\begin{definition}
A mechanism $f$ is called \emph{$k$-known-agents safely-manipulable} if for some agent $a_i$, some subset $K\subseteq N\setminus \{a_i\}$ with $|K|=k$ and some preferences for them $\prefsOf{K}$,
there exists a manipulation $P_i$ that is both profitable and safe for $a_i$ given $K$ and $\prefsOf{K}$.
\end{definition}


\begin{propositionrep}
\label{prop:monotonicity}
    Let $k \in \{0, \ldots, n-2\}$.
    If a mechanism is $k$-known-agents safely-manipulable, then it is also $(k+1)$-known-agents safely-manipulable.
\end{propositionrep}

\begin{proof}
    By definition, some agent~$a_i$ has a profitable-and-safe-manipulation-given-$k$-known-agents. 
    That is, there exists a subset $K \subseteq N \setminus \{a_i\}$ with $|K| = k$ and some preference profile for them $\prefsOf{K} \in \domainsOf{K}$, such that \eqref{eq:k-manipulation} and \eqref{eq:k-safe-manipulation} hold.
    Let $a_j \in \notK$.
    Consider the preferences $P_j$ that $a_j$ has in some profile satisfying \eqref{eq:k-manipulation} (profitable). 
    Define  $K^+ := K \cup \{a_j\}$ and construct a preference profile where the preferences of the agents in $K$ remain $\prefsOf{K}$, and $a_j$'s preferences are set to $P_j$.
    Since \eqref{eq:k-manipulation} holds for $P_j$, the same manipulation remains profitable given the new set of known-agents.
    Moreover, \eqref{eq:k-safe-manipulation} continues to hold, as the set of unknown agents has only shrunk.
    Thus, the mechanism is also $(k+1)$-known-agents safely manipulable.
\end{proof}

\Cref{prop:monotonicity}  (proved in the appendix) justifies the following definition:
\begin{definition}
    The \emph{RAT-degree} of a mechanism $f$ is the minimum $k$ for which the mechanism is $k$-known-agent safely-manipulable, or $n$ if there is no such $k$.
\end{definition}
Intuitively, a mechanism with a higher RAT-degree is harder to manipulate, as a risk-avoiding agent would need to collect more information in order to find a safe manipulation.

\begin{observation}
(a) A mechanism is truthful~ if-and-only-if its RAT-degree is $n$.

(b) A mechanism is RAT if-and-only-if its RAT-degree is at least $1$.
\end{observation}

\Cref{fig:hierarchy-RAT-Degree} illustrates the relation between classes of different RAT-degree. 



\begin{figure}[h]
  \centering
\begin{tikzpicture}[node distance = 1.25cm, text width=2.5cm, text height=0.25cm, text depth=0.1cm,                        rectangle, draw, text centered, ] 
\draw[black,line width=1pt] (0,0)--(11.5,0);
\draw[decorate,decoration={brace,mirror,amplitude=10pt},thin] (0,0) -- (3,0);
\draw[decorate,decoration={brace,mirror,amplitude=10pt},thin] (3,0) -- (4.5,0);
\draw[decorate,decoration={brace,mirror,amplitude=10pt},thin] (4.5,0) -- (6,0);
\draw[decorate,decoration={brace,mirror,amplitude=10pt},thin] (8.5,0) -- (10,0);
\draw[decorate,decoration={brace,mirror,amplitude=10pt},thin] (10,0) -- (11.5,0);
\draw[thin,draw=gray!70] (0,.3) -- (3,.3); 
\draw[thin, draw=gray!70] (0,.1) -- (0,0);
\draw[thin, dashed, draw=gray!70] (3,.3) -- (3,0);
\draw[thin, draw=gray!70] (0,0.3) -- (0,0.2);
\draw (1.35,.2) node[above]{\mbox{Safely-Manipulable}};
\draw[thin,draw=gray!70] (0,.8) -- (4.5,.8); 
\draw[thin,dashed, draw=gray!70] (4.5,.8) -- (4.5,0);
\draw (2.1,.7) node[above]{\mbox{$1$-KA Safely-Manipulable}};
\draw[thin, dashed, draw=gray!70] (0,2.6) -- (0,0.4);
\draw[thin,draw=gray!70] (0,1.3) -- (6,1.3); 
\draw[thin,dashed, draw=gray!70] (6,1.3) -- (6,0);
\draw (2.9,1.2) node[above]{\mbox{$2$-KA Safely-Manipulable}};
\draw[thin,draw=gray!70] (0,2.1) -- (8.5,2.1); 
\draw[thin, dashed, draw=gray!70] (8.5,2.1) -- (8.5,0);
\draw (5.2,2) node[above]{\mbox{$(n-2)$-KA Safely-Manipulable}};
%
%
\draw[thin,draw=gray!70] (0,2.6) -- (10,2.6); 
\draw[thin, dashed, draw=gray!70] (10,2.6) -- (10,0);
\draw (6,2.5) node[above]{\mbox{$(n-1)$-KA Safely-Manipulable}};
\filldraw[draw=black,color=lightgray!18] (-1,-.95) rectangle (11.5,-0.4);
\draw (0,-.4) node[below]{RAT-Degree:};
\draw (1.5,-.4) node[below]{$0$};
\draw (3.75,-.4) node[below]{$1$};
\draw (5.25,-.4) node[below]{$2$};
\draw (7.25,-.4) node[below]{$\cdots$};
\draw (9.25,-.4) node[below]{$n-1$};
\draw (10.75,-.4) node[below]{$n$};
\draw[thin, dashed,draw=gray!70] (10,-1.3) -- (10,0);
\draw[thin, draw=gray!70] (10,-1.3) -- (10,-1.2);
\draw[thin, draw=gray!70] (11.5,-1.3) -- (11.5,-1.2);
\draw[thin, dashed, draw=gray!70] (11.5,-1.1) -- (11.5,0);
\draw[thin, dashed, draw=gray!70] (11.5,-1.2) -- (11.5,-1.65);
\draw[thin, draw=gray!70] (11.5,-1.8) -- (11.5,-1.7);
\draw[thin, draw=gray!70] (10,-1.3) -- (11.5,-1.3);
\draw (10.75,-1.3) node[below]{Truthful};
\draw[thin, dashed, draw=gray!70] (3,-1.8) -- (3,0);
\draw[thin, draw=gray!70] (3,-1.8) -- (11.5,-1.8);
\draw (6.25,-1.8) node[below]{\mbox{Risk-Avoiding Truthful (RAT)}};
\end{tikzpicture}
  \caption{\centering Hierarchy of the Manipulability and Truthfulness Classes with respect to the RAT-Degree.
The horizontal axis represents the RAT-Degree, from $0$ (safely-manipulable) 
to $n$ (truthful).
Labels above the axis correspond to Manipulability Classes, while labels below the axis correspond to Truthfulness Classes. KA~stands for Known-Agents.}
    \label{fig:hierarchy-RAT-Degree}
\end{figure}

\subsection{An Intuitive Point of View}\label{sec:intuitive}
Consider \Cref{tab:safe-manip-i}. We adopt the point of view of a particular agent $a_i$. The rows $T_i, P^1_i, \ldots, P^3_i$ correspond to the possible reports of agent $a_i$, where $T_i$ is the truthful report and the rest are potential manipulations.
The columns $\prefsExcI^1, \prefsExcI^2, \ldots$ represent possible strategy profiles of the remaining $n-1$ agents. The values $x_1, x_2 \ldots$ indicate the utility of agent $a_i$ under each of these profiles when she reports truthfully.

When the risk-avoiding agent has no information ($0$-known-agents), a profitable-and-safe manipulation is a row in the table that represents an alternative report $P_i \neq T_i$, that (strictly) dominates $T_i$.
That is, in each column, the outcome of the manipulation is at least as good as the truthful outcome, and in at least one column it is strictly better. In this example, $P^1_i$ satisfies this property. 

When the risk-avoiding agent has more information ($k$-known-agents, when $k >0$), it is equivalent to considering a strict subset of the columns.
For instance, suppose the agent can infer—based on the information she has over some $k$ other agents—that only the profiles $\prefsExcI^3, \prefsExcI^4, \prefsExcI^5$ are possible. Then $P^2_i$ is safe and profitable given this set. Notice that the utilities of agent $a_i$ in profiles not among $\prefsExcI^3, \prefsExcI^4, \prefsExcI^5$ are irrelevant, since $a_i$ considers them impossible.

Lastly, when the risk-avoiding agent has a full information ($(n-1)$-known-agents), she knows the exact strategy profile of the other agents, so it is equivalent to consider only one column in which the manipulation is profitable. $P^3_i$ in the table illustrates this type of manipulation. 

\renewcommand{\arraystretch}{1.3}
\begin{table}[h]
    \centering
    \begin{tabular}{c|c|c|c|c|c|c|c|c}
             & $\prefsExcI^1$ & $\prefsExcI^2$ & $\prefsExcI^3$ & $\prefsExcI^4$ & $\prefsExcI^5$& $\prefsExcI^6$& $\prefsExcI^7$& $\ldots$ \\
             \hline
             $T_i$ & $x_1$ & $x_2$ & $x_3$ & $x_4$ & $x_5$ & $x_6$ & $x_7$ & $\ldots$\\
             \hline
             $P^1_i \neq T_i$ & \cellcolor{black!10} $\geq x_1$ & \cellcolor{black!10} $\geq x_2$ & \cellcolor{black!10} $\geq x_3$ & \cellcolor{black!35} $>x_4$ & \cellcolor{black!10} $\geq x_5$ & \cellcolor{black!10} $\geq x_6$ & \cellcolor{black!10} $\geq x_7$ & $\ldots$\\
             \hline
             $P^2_i \neq T_i$ &  &  & \cellcolor{black!10} $\geq x_3$ & \cellcolor{black!35} $>x_4$ & \cellcolor{black!10} $\geq x_5$ &  & & \\
             \hline
             $P^3_i \neq T_i$ &  &  &  & \cellcolor{black!35} $>x_4$ &  &  & & \\
             \hline
        \end{tabular}
    \caption{A Safe-And-Profitable Manipulation from an Agent Perspective.
    Dark-gray cells mean the value must be strictly higher, light-gray means it must be at least as high, and white means any value is allowed.}
    \label{tab:safe-manip-i}
\end{table}

\renewcommand{\arraystretch}{1.2}

\section{Auctions for a Single Good}
\label{sec:single-item-auction}
We consider a seller owning a single good, and $n$ potential buyers (the agents). 
The true preferences  $T_i$ of buyer $a_i$ are given by real values $v_i \geq 0$, representing her happiness from receiving the good. 
The reported preferences $P_i$ are the ``bids'' $b_i \geq 0 $.
A mechanism in this context has to determine the \emph{winner} --- the agent who will receive the good, and the \emph{price} --- how much the winner will pay.
We assume that agents are quasi-linear -- meaning their valuations can be interpreted in monetary units. 
Thus, the utility of the winning agent is the valuation minus the price; while the utility of the other agents is zero.

\paragraph{Results.} \Cref{tab:auction-sum} provides a summary of our results. 
The two most well-known mechanisms in this context are first-price auction and second-price auctions. 
First-price auction maximizes the seller's revenue when all buyers are truthful; but this assumption is, of course, unrealistic, as it is known to be manipulable. In fact, it is even safely-manipulable, so its RAT-degree is $0$ (\Cref{thm:first-price-auction}).
We then prove that a first-price auction with a (fixed) positive discount has RAT-degree $1$.

Second-price auction is known to be truthful, so its  RAT-degree is $n$.
However, it has some important practical disadvantages \cite{ausubel2006lovely}. In particular, when buyers are risk-averse, the expected revenue of a second-price auction is lower than that of a first-price auction \cite{nisan2007algorithmic};
even when the buyers are risk-neutral, a risk-averse \emph{seller} would prefer the revenue distribution of a first-price auction \citep{krishna2009auction}.%

This raises the question of whether it is possible to combine the advantages of both auction types.
Indeed, we prove that 
any auction that applies a weighted average between the first-price and the second-price\footnote{The average price auction is mentioned as an exercise in  \cite{krishna2009auction} (Problem 2.4).} achieves a RAT-degree of $n-1$, which is very close to being fully truthful (RAT-degree~$n$). 
This implies that a manipulator agent would need to obtain information about all $n-1$ other agents to safely manipulate -- which is a very challenging task.
The seller’s revenue from such an auction is higher compared to the second-price auction, giving this mechanism a significant advantage in this context. 
This result opens the door to exploring new mechanisms that are not truthful but come very close it. 
Such mechanisms may enable desirable properties that  are unattainable with truthful mechanisms.

\begin{table}[h]
    \centering
    \begin{tabular}{|l|c|}
    \hline 
    Mechanism & RAT-Degree  \\
    \hhline{|=|=|}
    First-Price  & $0$ \\
    \hline
    First-Price With Positive Discount & $1$ \\
    \hline
    Average Between First and Second Price  & $n-1$ \\
    \hline
    Second-Price & $n$  \\
    \hline
    \end{tabular}
    \caption{Auctions for a Single Good: Summary of Results}
    \label{tab:auction-sum}
\end{table}



\subsection{First-Price Auction}\label{sec:auction-first-price}
\begin{toappendix}
    \subsection{First-Price Auction}
\end{toappendix}
In first-price auction, the agent who bids the highest price wins the good and pays her bid; the other agents get and pay nothing. 

It in well-known that the first-price auction is not truthful.
We start by proving that the situation is even worse: first-price auction is safely-manipulable, meaning that its RAT-degree is $0$.

\begin{theoremrep}\label{thm:first-price-auction}
First-price auction is safely manipulable (RAT-degree = $0$).
\end{theoremrep}

\begin{proofsketch}
    A truthful agent always gets a utility of $0$,  whether he wins or loses.
    On the other hand, an agent who manipulates by bidding slightly less than $v_i$ gets a utility of at least $0$ when he loses and a strictly positive utility when he wins. 
    Thus, this manipulation is safe and profitable.
\end{proofsketch}

\begin{proof}
To prove the mechanism is safely manipulable, we need to show an agent and an alternative bid, such that the agent always \emph{weakly} prefers the outcome that results from reporting the alternative bid over reporting her true valuation, and \emph{strictly} prefers it in at least one scenario.
We prove that the mechanism is safely-manipulable by all agents with positive valuations.
    
Let $i \in N$ be an agent with valuation $v_i > 0$, we shall now prove that bidding $b_i < v_i$ is a safe manipulation.
    
We need to show that for any combination of bids of the other agents, bidding $b_i$ does not harm agent~$i$, and that there exists a combination where bidding $b_i$ strictly increases her utility. 
    To do so, we consider the following cases according to the maximum bid of the other agents:
    \begin{itemize}
        \item The maximum bid is smaller than $b_i$: if agent~$i$ bids her valuation $v_i$, she wins the good and pays $v_i$, results in utility $0$. 
        However, by bidding $b_i < v_i$, she still wins but pays only $b_i$, yielding a positive utility. 
        
        Thus, in this case, agent~$i$ strictly increases her utility by lying.

        \item The maximum bid is between $b_i$ and $v_i$: if agent $i$ bids her valuation $v_i$, she wins the good and pays $v_i$, resulting in utility $0$. 
        By bidding $b_i < v_i$, she loses the good but pays noting, also resulting in utility $0$.
        
        Thus, in this case, bidding $b_i$ does not harm agent~$i$.

        \item  The maximum bid is higher than $v_i$: Regardless of whether agent~$i$ bids her valuation $v_i$ or $b_i < v_i$, she does not win the good, resulting in utility $0$.
        
        Thus, in this case, bidding $b_i$ does not harm agent~$i$.
    \end{itemize}

\end{proof}

\subsection{First-Price With Positive Discount Auction}\label{sec:action-first-price-w-discount}
\begin{toappendix}
    \subsection{First-Price With Positive Discount Auction}
\end{toappendix}

In first-price auction with discount, the agent $i$ with the highest bid $b_i$ wins the item and pays $(1-t)b_i$, where $t\in(0,1)$ is a constant. As before, the other agents get and pay nothing.

We prove that, although this minor change does not make the  mechanism truthful, it increases the degree of trustfulness for risk-avoiding agents. 
However, it is still quite vulnerable to manipulation as knowing the strategy of one other agent might be sufficient to safely-manipulate it.

\begin{theoremrep}
\label{auction-with-discount}
The RAT-degree of the First-Price Auction with Discount is $1$.
\end{theoremrep}


\newcommand{\thmFpawdText}{
To prove that the RAT-degree is $1$, we need to show that the mechanism is (a) \emph{not} safely-manipulable, and (b) 
$1$-known-agents safely-manipulable.
We prove each of these in a lemma.
}

\thmFpawdText

\begin{toappendix}
\thmFpawdText
\end{toappendix}

\begin{lemmarep}
    First-Price Auction with Discount is not safely-manipulable. 
\end{lemmarep}
\begin{proofsketch}
    Here, unlike in the first-price auction, whenever the agent wins the good, she gains positive utility.  This implies that no manipulation can be safe. If she under-bids her value, she risks losing the item, which strictly decreases her utility. If she over-bids, she might end up paying more than necessary, reducing her utility without increasing her chances of winning.
\end{proofsketch}
\begin{proof}
We need to show that, for each agent $i$ and any bid $b_i\neq v_i$, at least one of the following is true: either (1) for any combination of bids of the other agents, agent $i$ weakly prefers the outcome from bidding $v_i$; or (2) there exists such a combination for which $i$ strictly prefers the outcome from bidding $v_i$. We consider two cases.

Case 1: $v_i = 0$. In this case condition (1) clearly holds, as bidding $v_i$ guarantees the agent a utility of $0$, and no potential outcome of the auction can give $i$ a positive utility.

Case 2: $v_i > 0$. In this case we prove that condition (2) holds. We consider two sub-cases:
\begin{itemize}
\item Under-bidding ($b_i<v_i$): 
    whenever $\displaystyle \max_{j\neq i}b_j \in (b_i,v_i)$, 
when agent $i$ bids truthfully she wins the good and pays $(1-t) v_i$, resulting in utility $t v_i > 0$;
but when she bids $b_i$ she does not win, resulting in utility $0$.

\item Over-bidding ($v_i<b_i$): 
    whenever $\displaystyle \max_{j\neq i} < v_i$, 
when agent $i$ bids truthfully her utility is $t v_i$ as before;
when she bids $b_i$ she wins and pays $(1-t)b_i > (1-t)v_i$, so her utility is less than $t v_i$.
\end{itemize}
In both cases lying may harm agent $i$. Thus, she has no safe manipulation. 
\end{proof}

\begin{lemmarep}
    First-Price Auction with Discount is $1$-known-agents safely-manipulable. 
\end{lemmarep}
\begin{proofsketch}
    Consider the case where the manipulator $a_i$ knows there is another agent $a_j$ bidding $b_j \in (v_i, v_i/(1-t))$.
    When biding truthfully, she loses and gets zero utility. 
    However, by bidding any value $b_i\in (b_j, v_i/(1-t))$, she either wins and gains positive utility or loses and remains with zero utility. Thus it is safe and profitable. 
\end{proofsketch}
\begin{proof}
we need to identify an agent $i$, for whom there exists another agent $j \neq i$ and a bid $b_j$, such that if $j$ bids $b_j$, then agent $i$ has a safe manipulation.

Indeed, let $i \in N$ be an agent with valuation $v_i >0$ and let $j \neq i$ be another agent who bids some value  $b_j \in (v_i, v_i/(1-t))$. 
We prove that bidding any value $b_i\in (b_j, v_i/(1-t))$ is a safe manipulation for $i$.

If $i$ truthfully bids $v_i$, she loses the good (as $b_j > v_i$), and gets a utility of $0$.

If $i$ manipulates by bidding some $b_i\in (b_j, v_i/(1-t))$, then she gets either the same or a higher utility, depending on the maximum bid among the unknown agents ($b^{\max} := \displaystyle \max_{\ell\neq i, \ell\neq i}b_{\ell}$):
    \begin{itemize}
        \item If $b^{\max} < b_i$, then $i$ wins and pays $(1-t)b_i$, resulting in a utility of $v_i - (1-t)b_i > 0$, as $(1-t)b_i < v_i$. Thus, $i$ strictly gains by lying.

        \item  If $b^{\max} > b_i$, then $i$ does not win the good, resulting in utility $0$. Thus, in this case, bidding $b_i$ does not harm $i$.
        
        \item If $b^{\max} = b_i$, then one of the above two cases happens (depending on the tie-breaking rule).
    \end{itemize}
In all cases $b_i$ is a safe manipulation, as claimed.
\end{proof}

\subsection{Average Between First and Second Price Auction}\label{sec:auction-avg-price}

\begin{toappendix}
    \subsection{Average Between First and Second Price Auction}
\end{toappendix}

In the Average-First-Second-Price (AFSP) Auction, the agent $i$ with the highest bid $b_i$ wins the item and pays $w b_i + (1-w) b^{\max}_{-i}$, where $\displaystyle b^{\max}_{-i} := \max_{j\neq i} b_j$ is the second-highest bid, and $w\in(0,1)$ is a fixed constant. That is, the price is a weighted average between the first price and the second price.

We show that this simple change makes a significantly difference -- the RAT-degree increases to $n-1$. 
This means that a manipulator agent would need to obtain information about all other agents to safely manipulate the mechanism --- a very challenging task in practice.

\begin{theoremrep}
\label{auction-average-price}
The RAT-degree of the Average-Price Auction is $(n-1)$.
\end{theoremrep}
The theorem is proved using the following two lemmas.
\begin{toappendix}
    The theorem is proved using the following two lemmas.
\end{toappendix}

\begin{lemmarep}
The Average-Price Auction is not $(n-2)$-known-agents safely-manipulable.
\end{lemmarep}

\begin{proofsketch}
Let $a_i$ be a potential manipulator, $b_i$ a potential manipulation, $K$ be a set of known agents, with $|K| = n-2$, and $\mathbf{b}_{K}$ be a vector that represents their bids.
We prove that the only unknown-agent $a_j$ can make any manipulation either not profitable or not safe.
We denote by $b^{\max}_{K}$ the maximum bid among the agents in $K$ and consider each of the six possible orderings of $v_i$, $b_i$ and $b^{\max}_{K}$.
In two cases ($v_i < b_i < b^{\max}_{K}$ or $b_i < v_i < b^{\max}_{K}$) the manipulation is not profitable;
in the other four cases, the manipulation is not safe.
\end{proofsketch}




\begin{proof}
Let $i\in N$ be an agent with true value $v_i>0$, and a manipulation $b_i \neq v_i$. We show that this manipulation is unsafe even when knowing the bids of $n-2$ of the other agents.

Let $K$ be a subset of $(n-2)$ of the remaining agents (the agents in $K$ are the ``known agents''), and let $\mathbf{b}_{K}$ be a vector that represents their bids.
Lastly, let $j$ be the only agent in $N\setminus (K\cup \{i\})$.
We need to prove that at least one of the following is true: either (1) the manipulation is not profitable --- for any possible bid of agent~$j$, agent~$i$ weakly prefers the outcome from bidding $v_i$ 
over the outcome from bidding $b_i$;
or (2) the manipulation is not safe --- there exists a bid for agent~$j$, such that agent~$i$ strictly prefers the outcome from bidding $v_i$.

Let $\displaystyle b^{\max}_{K} := \max_{\ell\in K}b_{\ell}$.
We consider each of the six possible orderings of $v_i$, $b_i$ and $b^{\max}_{K}$ (cases with equalities are contained in cases with inequalities, according to the tie-breaking rule):

\begin{itemize}
    \item 
    $v_i < b_i < b^{\max}_{K}$ or $b_i < v_i < b^{\max}_{K}$:  
    In these cases (1) holds, as for any bid of agent $j$, agent $i$ never wins. Therefore the manipulation is not profitable.
    
    \item 
    $v_i < b^{\max}_{K} < b_i$:
We show that (2) holds. Assume that $j$ bids any value $b_j \in (v_i, b_i)$.
When $i$ bids truthfully, she does not win the good so her utility is $0$.
But when $i$ bids $b_i$, she wins and pays a weighted average between $b_i$ and $\max(b_j, b_K^{\max})$. As both these numbers are strictly greater than $v_i$, the payment is larger than $v_i$ as well, resulting in a negative utility. Hence, the manipulation is not safe.
    
    
\item 
$b^{\max}_{K} < v_i < b_i$: We show that (2) holds.
Assume that $j$ bids any value $b_j < b^{\max}_{K}$.
When $i$ tells the truth, she wins and pays $w v_i + (1-w)b^{\max}_{K}$; 
but when $i$ bids $b_i$, she still wins but pays a higher price, $w b_i + (1-w)b^{\max}_{K}$, so her utility decreases.   
Hence, the manipulation is not safe.
    
\item 
$b_i < b^{\max}_{K} < v_i$: We show that (2) holds.
Assume that agent $j$ bids any value $b_j < b_i$.
When $i$ tells the truth, she wins and pays $w v_i + (1-w)b^{\max}_{K} < v_i$, resulting in a positive utility.
But when $i$ bids $b_i$, she does not win and her utility is $0$.
Hence, the manipulation is not safe.
    
\item 
$b^{\max}_{K} < b_i < v_i$: We show that (2) holds.
Assume that $j$ bids any value $b_j\in(b_i,v_i)$.
When $i$ tells the truth, she wins and pays $w v_i + (1-w)b_j < v_i$, resulting in a positive utility.
But when $i$ bids $b_i$, she does not win and her utility is $0$.
Hence, the manipulation is not safe.

\item 
$b_i < b^{\max}_{K} < v_i$ or $b^{\max}_{K} < b_i < v_i$:
We show that (2) holds. 
Assume that $j$ bids any value $b_j\in(b_i,v_i)$.
When $i$ tells the truth, she wins and pays $w v_i + (1-w)\max(b_j, b^{\max}_{K})$, which is smaller than $v_i$ as both $b_j$ and $b^{\max}_{K}$ are smaller than $v_i$. Therefore, $i$'s utility is positive.
But when $i$ bids $b_i$, she does not win and her utility is $0$.
Hence, the manipulation is not safe.
\end{itemize}
\end{proof}

\begin{lemmarep}
The Average-Price Auction is $(n-1)$-known-agents safely-manipulable.
\end{lemmarep}
\begin{proofsketch}
Consider any combination of bids of the other agents in which all the bids are strictly smaller than $v_i$. 
Let $b^{\max}_{-i}$ be the highest bid among the other agents. 
Then any alternative bid $b_i \in (b^{\max}_{-i}, v_i)$ is a safe manipulation.
\end{proofsketch}
\begin{proof}
given an agent $i \in N$, we need to show an alternative bid $b_i\neq v_i$ and a combination of $(n-1)$ bids of the other agents, such that the agent strictly prefers the outcome resulting from its untruthful bid over her true valuation.

Consider any combination of bids of the other agents in which all the bids are strictly smaller than $v_i$. 
Let $b^{\max}_{-i}$ be the highest bid among the other agents. 
We prove that any alternative bid $b_i \in (b^{\max}_{-i}, v_i)$ is a safe manipulation.

When agent $i$ bids her valuation $v_i$, she wins the good and pays $w v_i + (1-w) b^{\max}_{-i}$, yielding a (positive) utility of 
\begin{align*}
    v_i- w v_i - (1-w)b^{\max}_{-i}
    =
    (1-w) (v_i - b^{\max}_{-i}).
\end{align*}
But when $i$ bids $b_i$, as $b_i > b^{\max}_{-i}$, she still wins the good but pays $w b_i + (1-w) b^{\max}_{-i}$, which is smaller as $b_i<v_i$; therefore her utility is higher.
\end{proof}

\paragraph{Conclusion.}
By choosing a high value for the parameter $w$, the Average-Price Auction becomes similar to the first-price auction, and therefore may attain a similar revenue in practice, but with better strategic properties. 
The average-price auction is sufficiently simple to test in practice; we find it very interesting to check how it fares in comparison to the more standard auction types.

\erel{
Maybe cite \citet{braverman2016interpolating}, which --- similarly to our average auction --- interpolates between truthful and non-truthful auctions.
}

\section{Indivisible Goods Allocations}\label{sec:indivisible-good-aloc}
In this section, we consider several mechanisms to allocate $m$ indivisible goods $G = \{g_1, \ldots, g_m\}$ among the $n$ agents.
Here, the true preferences $T_i$ of agent $a_i$ are given by $m$ real values: $v_{i,\ell}  \geq 0$ for any $g_{\ell} \in G$, representing her happiness from receiving the good $g_{\ell}$. 
The reported preferences $P_i$ are real values $r_{i,\ell} \geq 0 $.
We assume the agents have \emph{additive} valuations over the goods.
Given a bundle $S\subseteq G$, let $v_i(S)=\sum_{g_\ell\in S}v_{i,\ell}$ be agent $a_i$'s utility upon receiving the bundle $S$.
A mechanism in this context gets $n$ (potentially untruthful) reports from all the agents and determines the \emph{allocation} -- a partition $(A_1,\ldots,A_n)$ of $G$, where $A_i$ is the bundle received by agent $a_i$.

\paragraph{Results.} \Cref{tab:goods-alloc-sum} summarizes our results. 
We start by considering a simple mechanism, the \emph{utilitarian} goods allocation -- which assigns each good to the agent who reports the highest value for it.
This mechanism always returns an allocation that maximizes the social welfare (i.e., sum of utilities) and is thus \emph{Pareto Optimal} --- there does not exist another allocation that all agents weakly prefer and at-least one strictly prefers.
This mechanism is safely manipulable (RAT-degree = $0$). 
We then show that the RAT-degree can be increased to $1$ by requiring normalization --- the values reported by each agent are scaled such that the set of all items has the same value for all agents.

We then consider the famous \emph{round-robin} mechanism. It is easy to compute (specifically, polynomial) and guarantees \emph{envy-freeness up to one good (EF1)} --- each agent weakly prefers their own bundle over each other agent's bundle by the removal of at most one of the goods. 
\citet{amanatidis2017truthful} prove that EF1 is incompatible with  truthfulness for $m\geq 5$, which implies that the degree is at most $n-1$. 
We prove that the situation is much worse, the RAT-degree of round-robin is at most $1$. 
The proof relies on weak preferences (i.e., preferences that allow ties).
In sharp contrast, we prove that when all agents have strict preferences, the RAT-degree is the best possible $n - 1$.
%
%

\rmark{Next, we turn to the Max-Nash-Welfare (MNW) rule, which selects an allocation that maximizes the product of agents’ valuations.
This rule always outputs an allocation that is EF1 and Pareto optimal, and it also provides a good approximation to the maximin share (MMS) guarantee (see \cite{caragiannis2019unreasonable} for more details).
We examine three variants of the MNW mechanism, which differ in how they handle tie-breaking when multiple allocations achieve the same Nash welfare.
We show that the RAT-degree of all these variants is at most~$1$.} 

Lastly, we design a new mechanism that we call \emph{volatile priority}. It satisfies EF1 and attains the best possible RAT-degree of $n-1$ for all types of preferences;
but does not run in polynomial time.
The main observation behind our approach is that, in many cases, there are multiple solutions (here, allocations) that satisfy the required fairness properties (here, EF1). 
We use a priority order over the agents, determined by their reported valuations, to select among these solutions.
As a result, a manipulating agent who lacks information about \emph{all} other agents' valuations may unintentionally lower her priority by misreporting, and thus end up worse off. This makes manipulations risky.

This raises the following open question regarding the computational complexity:
\begin{open}\label{open-alloc-1}
    Is there a polynomial time EF1 goods allocation rule with RAT-degree $n-1$?
\end{open}

\rmark{A natural candidate is The Envy-Cycle Elimination Mechanism by \citet{lipton2004approximately}.}
\begin{open}\label{open:envy-elimination}
What is the RAT-degree of The Envy-Cycle Elimination Mechanism by \citet{lipton2004approximately}?
\end{open}

\begin{table}[h]
    \centering
    \begin{tabular}{|l|l|c|l|}
    \hline 
    \multicolumn{2}{|c|}{Mechanism} &RAT-Degree  & Properties\\
    \hhline{|==|=|=|}
    \multicolumn{2}{|l|}{Utilitarian} & $0$ & Pareto Efficient\\
    \hline
    \multicolumn{2}{|l|}{\multirow{2}{*}{Normalized Utilitarian}} & \multirow{2}{*}{$1$} & Pareto Efficient \\
     \multicolumn{2}{|l|}{} &  & (w.r.t. the normalized values)\\
    \hline
    \multicolumn{2}{|l|}{Max-Nash-Welfare} & $\leq 1$ \textcolor{white}{$\leq$} & EF1, Pareto Efficient\\
    \hline
    \multirow{2}{*}{Round-Robin (for $n > m$)} & Any preferences & $1$ & \multirow{2}{*}{EF1, Polynomial-time}\\
    \cline{2-3}
     & Strict preferences & $n-1$ & \\
    \hline
    \multicolumn{2}{|l|}{Volatile Priority*} & $n-1$ & EF1\\
    \hline
    \end{tabular}
    \caption{\centering Indivisible Goods Allocations: Summary of Results. (*) marks mechanisms introduced in this paper.}
    \label{tab:goods-alloc-sum}
\end{table}

\subsection{Utilitarian Goods Allocation}\label{sec:utilitarian-alloc}
\begin{toappendix}
    \subsection{Utilitarian Goods Allocation}
\end{toappendix}
The \emph{utilitarian} rule aims to maximize the sum of agents' utilities. 
When agents have additive utilities, this goal can be achieved by assigning each good to an agent who reports the highest value for it.
This section analyzes this mechanism.

We make the practical assumption that agents’ reports are bounded from above by some maximum possible value $r_{max} >0$.

We assume that in cases where multiple agents report the same highest value for some good, the mechanism employs some tie-breaking rule to allocate the good to one of them.
However, the tie-breaking rule must operate independently for each good, meaning that the allocation for one good cannot depend on the tie-breaking outcomes of other goods.
\erel{This assumption it not obvious: it is reasonable to assume that, if an agent loses one item due to tie-breaking, her priority will increase in tie-breaking on other items.
Can it be removed easily (without complicating the proofs too much)?}
\eden{I think this assumption cannot be removed easily - it allows us to analyze each good independently of the other goods.. I'll try to think about it}
\eden{NEW: NO. I think that if we apply such a rule, it will no longer be a safe manipulation to report $r_{max}$ for all items.. I think this change will have a substantial effect, but I’m still not sure exactly how.}


\begin{theoremrep}
\label{prop:auction-knownagents}
The Utilitarian allocation rule is safely manipulable (RAT-degree = 0).
\end{theoremrep}

\begin{proofsketch}
    Manipulating by reporting the highest possible value, $r_{max}$, for all goods is both profitable and safe. It is profitable because if the maximum report among the other agents for a given good lies between the manipulator's true value and their alternative bid, $r_{max}$, the manipulator's utility strictly increases. It is safe because, in all other cases, the utility remains at least as high.
\end{proofsketch}


\begin{proof}
    To prove the mechanism is safely manipulable, we need to show one agent that has an alternative report, such that the agent always \emph{weakly} prefers the outcome that results from reporting the alternative report over reporting her true valuations, and \emph{strictly} prefers it in at least one scenario.
    
     Let $a_1 \in N$ be an agent who has a value different than $0$ and $r_{max}$ for at least one of the goods. Let $g_1$ be such good -- that is, $0 < v_{1,1} < r_{max}$.
     We prove that reporting the highest possible value, $r_{max}$, for all goods is a safe manipulation.
     Notice that this report is indeed a manipulation as it is different than the true report in at least one place.
     \erel{We can also consider the manipulation that increases only $v_{1,l}$ to $r_{max}$. Then we have fewer cases to consider: for good $l$ the manipulation is safe and profitable, and for the other goods there is no change at all.}
    
    We need to show that for any combination of reports of the other agents, this report does not harm agent~$a_i$, and that there exists a combination where it strictly increases her utility. 

Since the utilities are additive and tie-breaking is performed separately for each good, we can analyze each good independently. 
    The following proves that for all goods, agent $a_1$ always weakly prefers to \er{manipulate}. Case 4 (marked by *) proves that for the good $g_1$, there exists a combination of reports of the others, for which agent~$a_1$ strictly prefers the outcome from \er{manipulating}.
    Which proves that it is indeed a safe profitable manipulation.

    Let $g_{\ell} \in G$ be a good. 
    We consider the following cases according to the value of agent $a_1$ for the good $g_{\ell}$, $v_{1, \ell}$; and the maximum report among the other agents for this good:
    \begin{itemize}
        \item $v_{1,\ell} = r_{max}$: In this case, both the truthful and the untruthful reports are the same.

        \item $v_{1,\ell} =0$: Agent~$a_1$ does not care about this good, so regardless of the reports which determines whether or not agent~$a_1$ wins the good, her utility from it is $0$.

        Thus, in this case, agent $a_1$ is indifferent between telling the truth and manipulating.

        \item $0 < v_{1,\ell} < r_{max}$ and the maximum report of the others is strictly smaller than $v_{1,\ell}$: Agent~$a_1$ wins the good in both cases -- whether she reports her true value or reports $r_{max}$.
        
        Thus, in this case, agent $a_1$ is indifferent between telling the truth and manipulating.

        \item (*) $0 < v_{1,\ell} < r_{max}$ and the maximum report of the others is greater than $v_{i,j}$ and smaller than $r_{max}$: when agent~$a_1$ reports her true value for the good $g_{\ell}$, then she does not win it. However, by bidding $r_{max}$, she does. 

        Thus, in this case, agent~$a_1$ strictly increases her utility by lying (as her value for this good is positive).

        \item $0 < v_{1,\ell} < r_{max}$ and the maximum report of the others equals $r_{max}$: 
        when agent $a_1$ reports her true value for the good, she does not win it. 
        But by bidding $r_{max}$, she may win the good (depending on the tie-breaking rule).

        Thus, in this case, agent $a_1$ 
        \er{either strictly gain or does not lose from manipulating.}
    \end{itemize}
\end{proof}


\subsection{Normalized Utilitarian Goods Allocation}\label{sec:normalized-utilitarian-alloc}

\begin{toappendix}
    \subsection{ Normalized Utilitarian Goods Allocation}
\end{toappendix}

In the \emph{normalized} utilitarian allocation rule, the agents' reports are first normalized such that each agent's values sum to a given constant $V >0$.
Then, each good is given to the agent with the highest \emph{normalized} value.
%
We focus on the case of at least three goods.

\begin{theoremrep}
\label{thm:normalized-utilitarian-goods}
For $m \geq 3$ goods,  the RAT-degree of the Normalized Utilitarian allocation rule is $1$.
\end{theoremrep}

We prove \Cref{thm:normalized-utilitarian-goods} using several lemmas that analyze different cases. 

The first two lemmas prove that the rule is not 
safely-manipulable, so its RAT-degree is at least~$1$.
\Cref{claim:normalized-agent-who-likes-single-good} addresses agents who value only one good positively, while \Cref{claim:normalized-agent-who-likes-at-least-two} covers agents who value at least two goods positively.

\begin{toappendix}
    We prove \Cref{thm:normalized-utilitarian-goods} using several lemmas that analyze different cases. 

    The first two lemmas prove that the rule is not 
    safely-manipulable, so its RAT-degree is at least~$1$.
    \Cref{claim:normalized-agent-who-likes-single-good} addresses agents who value only one good positively, while \Cref{claim:normalized-agent-who-likes-at-least-two} covers agents who value at least two goods positively.
\end{toappendix}

\begin{lemmarep}
\label{claim:normalized-agent-who-likes-single-good}
An agent who values only one good positively cannot safely manipulate the Normalized Utilitarian allocation rule.
\end{lemmarep}

\begin{proofsketch}
Due to the normalization requirement, 
any manipulation by such an agent involves reporting a lower value for the only good she values positively, while reporting a higher values for some goods she values at 0. This reduces her chances of winning her desired good and raises her chances of winning goods she values at zero, ultimately decreasing her utility. Thus, the manipulation is neither profitable nor safe.
\end{proofsketch}

\begin{proof}
    %
    Let $a_1$ be an agent who values only one good and let $g_1$ be the only good she likes. That is, her true valuation is $v_{1,1} = V$ and $v_{1,\ell} = 0$ for any $\ell \neq 1$ (any good $g_{\ell}$ different than $g_1$).
    
    To prove that agent $a_1$ does \emph{not} have a safe manipulation, we need to show that for any report for her either (1) for any reports of the other agents, agent $a_1$ weakly prefers the outcome from telling the truth; or (2) there exists a reports of the other agents, for which agent $a_1$ strictly prefers the outcome from telling the truth.

    Let $(r_{1,1}, \ldots, r_{1,m})$ be an alternative report for the $m$ goods. We assume that the values are already normalized.
    We shall now prove that the second condition holds (lying may harm the agent). 
    
    First, as the alternative report is different, we can conclude that $r_{1,1} < V$.
    We denote the difference by $\epsilon = V - r_{1,j}$.
    Next, consider the following reports of the other agents (all agents except $a_1$): $V - \frac{1}{2}\epsilon$ for item $g_1$ and $\frac{1}{2}\epsilon$ for some other good. 
    
    When agent $a_1$ reports her true value she wins her desired good $g_1$, which gives her utility $V >0$. 
    However, when she lies, she loses good $g_1$ and her utility decreases to $0$ (winning goods different than $g_1$ does not increases her utility).
    
    That is, lying may harm the agent. 
\end{proof}


\begin{lemmarep}
\label{claim:normalized-agent-who-likes-at-least-two}
An agent who values positively at least two goods cannot safely manipulate the Normalized Utilitarian allocation rule.
\end{lemmarep}

\newcommand{\increasedInd}{\mathrm{inc}}
\newcommand{\decreasedInd}{\mathrm{dec}}

\begin{proofsketch}
    Since values are normalized, any manipulation by such an agent must involve increasing the reported value of at least one good $g_{\increasedInd}$ while decreasing the value of at least one other good $g_{\decreasedInd}$. We show that such a manipulation is not safe, by considering the case where all other agents report as follows: they assign a value of $0$ to $g_{\increasedInd}$, a value between the manipulator's true and reported value for $g_{\decreasedInd}$, and a value slightly higher than the manipulator’s report for all other goods (it is possible to construct such reports that are non-negative and normalized). 
    
    With these reports, the manipulation causes the manipulator to lose the good $g_{\decreasedInd}$ which has a positive value for her, whereas   
    she wins the good $g_{\increasedInd}$ with or without the manipulation, and does not win any other good. Hence, the manipulation strictly decreases her total value.
\end{proofsketch}

\begin{proof}
    Let $a_1$ be an agent who values at least two good, $\valT{1}{1}, \ldots, \valT{1}{m}$ her true values for the $m$ goods, and $\repT{1}{1}, \ldots, \repT{1}{m}$ a manipulation for $a_1$.
    We need to show that the manipulation is either not \emph{safe} -- there exists a combination of the other agents' reports for which agent $a_1$ strictly prefers the outcome from telling the truth; or not \emph{profitable} -- for any combination of reports of the other agents, agent $a_1$ weakly prefers the outcome from telling the truth. 
    We will show that the manipulation is not safe by providing an explicit combination of other agents' reports.

    First, notice that since the the true values and untruthful report of agent $a_1$ are different and they sum to the same constant $V$, there must be a good $g_{\increasedInd}$ whose value was increased (i.e., $ \valT{1}{\increasedInd} < \repT{1}{\increasedInd}$), and a good $g_{\decreasedInd}$ whose value was decreased (i.e., $ \valT{1}{\decreasedInd} > \repT{1}{\decreasedInd}$).
    
    Next, let $\epsilon := \min \left\{\frac{1}{m-1}\repT{1}{\increasedInd},~ \frac{1}{2}(\valT{1}{\decreasedInd} - \repT{1}{\decreasedInd})\right\}$, notice that $\epsilon > 0$.
    Also, let $c := \repT{1}{\increasedInd} - \epsilon$, notice that $c > 0$ as well \er{(here we use the condition $m\geq 3$).}

    We consider the combination of reports in which all agents except $a_1$ report the following values, denoted by $r(1), \ldots, r(m)$:
    \begin{itemize}
        \item For good $g_{\increasedInd}$ they report $r(\increasedInd) := 0$.

        \item For good $g_{\decreasedInd}$ they report 
        $r(\decreasedInd) := \repT{1}{\decreasedInd} + \epsilon$.

        \item For the rest of the goods, $g_\ell \in G \setminus \{g_{\increasedInd}, g_{\decreasedInd}\}$, they report 
        $r(\ell) := \repT{1}{\ell} + \frac{1}{m-2} c$.
        
    \end{itemize}

We prove that the above values constitute a legal report --- they are  non-negative and normalized to $V$.

First, we show that the sum of values in this report is $V$:
    \begin{align*}
        \sum_{\ell =1}^m r(\ell) &= r(\increasedInd) + r(\decreasedInd) +\sum_{g_{\ell} \in G \setminus \{g_{\increasedInd}, g_{\decreasedInd}\}} r(\ell) 
        \\
        & = 0 + (\repT{1}{\decreasedInd} + \epsilon) + \sum_{g_{\ell} \in G \setminus \{g_{\increasedInd}, g_{\decreasedInd}\}} \left(\repT{1}{\ell} + \frac{1}{m-2} c\right) 
        \\
        & = (\repT{1}{\decreasedInd} + \epsilon) + \sum_{g_{\ell} \in G \setminus \{g_{\increasedInd}, g_{\decreasedInd}\}} \repT{1}{\ell} + (m-2)\frac{1}{m-2} c
        \\
        & = (\repT{1}{\decreasedInd} + \epsilon) + \sum_{g_{\ell} \in G \setminus \{g_{\increasedInd}, g_{\decreasedInd}\}} \repT{1}{\ell} + (\repT{1}{\increasedInd} - \epsilon) = \sum_{\ell =1}^m \repT{1}{\ell} = V.
    \end{align*}
    
Second, we show that all the values are non-negative:
    \begin{itemize}
        \item Good $g_{\increasedInd}$: it is clear as $r(\increasedInd) =0$.

        \item Good $g_{\decreasedInd}$: since  $r(\decreasedInd)$ is strictly higher than the (non-negative) report of agent $a_1$ by $\epsilon >0 $, it is clearly non-negative.

        \item Rest of the goods, $g_{\ell} \in G \setminus \{g_{\increasedInd}, g_{\decreasedInd}\}$: 
        since $\epsilon = \min \{\frac{1}{m-1}\repT{1}{\increasedInd}, ~\frac{1}{2}(\valT{1}{\decreasedInd} - \repT{1}{\decreasedInd})\}$, it is clear that $\epsilon \leq \frac{1}{m-1}\repT{1}{\increasedInd}$. As $m \geq 3$ and $\repT{1}{\increasedInd} >0$, we get that $c = \repT{1}{\increasedInd} - \epsilon = \frac{m-2}{m-1} \repT{1}{\increasedInd}$ is higher than $0$.
        As $r(\ell)$ is strictly higher than the (non-negative) report of agent $a_1$ by $c >0 $, it is clearly non-negative.
    \end{itemize}

Now, we prove that, given these reports for the $n-1$ unknown agents, agent $a_1$ strictly prefers the outcome from reporting truthfully to the outcome from manipulating.

    We look at the two possible outcomes for each good -- the one from telling and truth and the other from lying, and show that the outcome of telling the truth is always either the same or better, and that for at least one of the goods that agent $a_1$ wants (specifically, $g_\decreasedInd$) it is strictly better.

    \begin{itemize}
        \item For good $g_\increasedInd$ we consider two cases.
        \begin{enumerate}
            \item If $\valT{1}{\increasedInd} = 0$: when agent $a_1$ is truthful we have a tie for this good as $r(\increasedInd) = 0$.
            When agent $a_1$ manipulates, she wins the good (as $\repT{1}{\increasedInd} > \valT{1}{\increasedInd} = 0 = r(\increasedInd)$).
            However, as $\valT{1}{\increasedInd} = 0$, in both cases, her utility from this good is $0$.

            \item If $\valT{1}{\increasedInd} > 0$:
            Whether agent $a_1$ says is truthful or not, she wins the good as $\repT{1}{\increasedInd} > \valT{1}{\increasedInd} > 0 =r(\increasedInd)$.
            Thus, for this good, the agent receives the same utility (of $\valT{1}{\increasedInd}$) when telling the truth or lying.

        \end{enumerate}

        \item For good $g_\decreasedInd$: when agent $a_1$ is truthful, she wins the good since $r(\decreasedInd) < \valT{1}{\decreasedInd}$:
        \begin{align*}
            r(\decreasedInd) &= \repT{1}{\decreasedInd} + \epsilon \\
            &= \repT{1}{\decreasedInd} + \min \{\frac{1}{m-1}\repT{1}{\increasedInd},~ \frac{1}{2}(\valT{1}{\decreasedInd} - \repT{1}{\decreasedInd})\}\\
            &\leq \repT{1}{\decreasedInd} +  \frac{1}{2}(\valT{1}{\decreasedInd} - \repT{1}{\decreasedInd}) \\
            &= \frac{1}{2}(\valT{1}{\decreasedInd} + \repT{1}{\decreasedInd})< \frac{1}{2}(\valT{1}{\decreasedInd} + \valT{1}{\decreasedInd}) = \valT{1} {\decreasedInd} && \text{(as $\repT{1}{\decreasedInd} < \valT{1}{\decreasedInd}$)}
        \end{align*}
        But when agent $a_1$ manipulates, she loses the good since $r(\decreasedInd) > \repT{1}{\decreasedInd}$ (as $r(\decreasedInd) = \repT{1}{\decreasedInd} + \epsilon$ and $\epsilon > 0$).
        
        As the real value of agent $a_1$ for this good is positive, the agent strictly prefers telling the truth for this good.

        \item Rest of the goods, $g_{\ell} \in G \setminus \{g_{\increasedInd}, g_{\decreasedInd}\}$:
        When agent $a_1$ is truthful, all the outcomes are possible -- the agent either wins or loses or that there is a tie.    
        
        However, as for this set of goods the reports of the other agents are $r(\ell) = \repT{1}{q}+\frac{1}{n-2}c > \repT{1}{\ell}$, when agent $a_1$ manipulates, she always loses the good.
        Thus, her utility from lying is either the same or smaller (since losing the good is the worst outcome).

    \end{itemize}

    Thus, the manipulation may harm the agent.
\end{proof}

The last lemma shows that the RAT-degree is at most $1$, thus completing the proof of the theorem.
\begin{lemmarep}\label{normalized-1-known}
    With $m \geq 3$ goods, Normalized Utilitarian is $1$-known-agent safely-manipulable.
\end{lemmarep}

\begin{proofsketch}
Consider a scenario where there is a known agent who reports $0$ for some good $g$ that the manipulator wants and slightly more than the manipulator’s true values for all other goods. In this case, by telling the truth, the manipulator has no chance to win any good except $g$. Therefore, reporting a 
value of $V$ for $g$ and a value of $0$ for all other goods is a safe manipulation.

The same manipulation is also profitable, since it is possible that the reports of all $n-2$ unknown agents for $g$ are larger than the manipulator’s true value and smaller than $V$. In such a case, the manipulation causes the manipulator to win $g$, which strictly increases her utility.
\end{proofsketch}

\begin{proof}
    Let $a_1$ be an agent and let $\valT{1}{1}, \ldots, \valT{1}{m}$ be her values for the $m$ goods.
    We need to show (1) an alternative report for agent~$a_1$, (2) another agent~$a_2$, and (3) a report agent ~$a_2$; such that 
    for any combination of reports of the remaining $n-2$ (unknown) agents, agent $a_1$ weakly prefers the outcome from lying, and that there exists a combination for which agent $a_1$ strictly prefers the outcome from lying.

    Let $g_{\ell^+}$ be a good that agent $a_1$ values (i.e., $\valT{1}{\ell^+} >0$).
    
    Let $a_2$ be an agent different than $a_1$.
    We consider the following report for agent~$a_2$: first, $\repT{2}{\ell^+} :=0$, and $\repT{2}{\ell} := \repT{1}{\ell} + \epsilon$ for any good $g_{\ell}$ different than $g_{\ell^+}$, where $\epsilon := \frac{1}{m-1} \valT{1}{\ell^+} $.
    Notice that $\epsilon >0$.

    We shall now prove that reporting $V$ for the good $g_{\ell^+}$ (and $0$ for the rest of the goods) is a safe manipulation for $a_1$ given that $a_2$ reports the described above.

    When agent~$a_1$ reports her true values, then she does not win the goods different than $g_{\ell^+}$ -- this is true regardless of the reports of the remaining $n-2$ agents, as agent~$a_2$ reports a higher value $\repT{1}{\ell} < \repT{1}{\ell} + \epsilon = \repT{2}{\ell}$.
    For the good $g_{\ell^+}$,  we only know that $\repT{2}{\ell^+} = 0 > \valT{1}{\ell^+} > \repT{1}{\ell^+} = V$, meaning that it depends on the reports of the $(n-2)$ remaining agents. 
    We consider the following cases, according to the maximum report for $g_{\ell^+}$ among the remaining agents:
    \begin{itemize}
        \item If the maximum is smaller than $\valT{1}{\ell^+}$: agent $a_1$ wins the good $\valT{1}{\ell^+}$ in both cases (when telling the truth or lies). 

        (the same)

        \item If the maximum is greater than $\valT{1}{\ell^+}$ but smaller than $\repT{1}{\ell^+} = V$: when agent $a_1$ tells the truth, she does not win the good. 
        However, when she lies, she does.
        
        (lying is strictly better). 

        \item If the maximum equals $\repT{1}{\ell^+} = V$: when agent $a_1$ tells the truth, she does not win the good. 
        However, when she lies, we have tie for this good.
        Although agent $a_1$ is risk-averse, having a chance to win the good is strictly better than no chance at all.
        
        (lying is strictly better). 
    \end{itemize}

\end{proof}

\subsection{Round-Robin Goods Allocation}\label{sec:round-robin}
\begin{toappendix}
    \subsection{Round-Robin Goods Allocation}
\end{toappendix}
In round-robin goods allocation, the agents are arranged according to some predetermined order $\pi$. 
There are $m$ rounds, corresponding to the number of goods. 
In each round, the agent whose turn it is (based on $\pi$) takes their most preferred item from the set of items that remain unallocated at that point.
When there are multiple items that are equally most preferred, the agent breaks the tie according to a fixed item-priority ordering.
Note that normalization is not important for Round-Robin, as it is not affected by scaling the valuations.

If there are at most $n$ items, then the rule is clearly truthful.
Thus, we assume that there are $m\geq n+1$ items.
Even the slight increment to $m=n+1$ makes a significant difference:

\begin{theorem}
    The RAT-degree of Round-Robin is: \Hquad 
    (a)~$n$ when $m\leq n$, \Hquad (b)~$1$ when $m > n$; \Hquad and (c)~$n-1$ when $m > n$ and all agents have strict preferences.
\end{theorem}

As (a) is trivial, we start by proving (b) using \Cref{RR-weak-prefs}.
The proof of the upper bound uses weak preferences; in contract, \Cref{thm:RR-strict-degree-n-1} proves that when all agents have strict preferences, the RAT-degree becomes~$n-1$, yielding part~(c).

\begin{lemmarep}\label{RR-weak-prefs}
    With $m\geq n+1$ items, the RAT-degree of Round-Robin is $1$.
\end{lemmarep}

\begin{proofsketch}
    We start by showing that Round-Robin is $1$-known-agent safely-manipulable.
    Consider the first turn of the agent who picks first, and the case where her two most preferred items are of equal value. If she knows that favorite item of the agent who picks right after is the item she planned to take first, she can safely manipulate by taking the other one instead. In some cases, this change allows her to obtain both items (since the second agent will then take a different one), and in all other cases she still ends up with the same utility as before. Hence, the manipulation is both profitable and safe.

    Next, we prove that Round-Robin is \emph{not} safely-manipulable. 
    We first show that it is never safe to manipulate in a way that results in picking an item that is strictly less preferred than the best available one as it might reduce the manipulator’s utility. Thus, the only potentially profitable manipulations are those that swap between equally valued items. 
    However, we show that if such a manipulation is profitable for some profile of the others’ preferences, we can construct a symmetric profile where the same manipulation is harmful.
\end{proofsketch}

\begin{toappendix}

To prove that the RAT-degree is exactly $1$, \Cref{claim:RR-RAT1-tie} proves that the RAT-degree is at most~$1$, while \Cref{claim-RR-degree-not-0} shows that it is at least~$1$.

\begin{lemma}\label{claim:RR-RAT1-tie}
    With $m\geq n+1$ items, Round-Robin is $1$-known-agent safely-manipulable.
\end{lemma}

\begin{proof}
    Let $\pi$ be the order according to agents' indices: $a_1,a_2,\ldots,a_n$.
    Let agent $a_1$'s valuation be such that $v_{1,1}=v_{1,2}=1$ and $v_{1,3}=\cdots = v_{1,m}=0$.
    Suppose agent $a_1$ knows that agent $a_2$ will report the valuation with $v_{2,1}=0$, $v_{2,2}=1$, and $v_{2,3}=\cdots=v_{2,m}=0.5$.
    We show that agent $a_1$ has a safe and profitable manipulation by reporting $v_{1,1}'=0.5$, $v_{1,2}'=1$, and $v_{1,3}'=\cdots=v_{1,m}'=0$.

    Firstly, we note that agent $a_1$'s utility is always $1$ when reporting her valuation truthfully, regardless of the valuations of agents $a_3,\ldots,a_n$.
    This is because agent $a_1$ will receive item $g_1$ (by the item-index tie-breaking rule) and agent $a_2$ will receive item $g_2$ in the first two turns.
    The allocation of the remaining items does not affect agent $a_1$'s utility. 

    Secondly, after misreporting, agent $a_1$ will receive item $g_2$ in the first turn, which already secures agent $a_1$ a utility of at least $1$. Therefore, agent $a_1$'s manipulation is safe.

    Lastly, if the remaining $n-2$ agents report the same valuations as agent $a_2$ does, it is easy to verify that agent $a_1$ will receive item $g_1$ in the $(n+1)$-th turn.
    In this case, agent $a_1$'s utility is $2$.
    Thus, the manipulation is profitable.
\end{proof}


\begin{lemma}\label{claim-RR-degree-not-0}
Round-robin is not safely-manipulable.
\end{lemma}

\begin{proof}
The proof is by induction on $m$. For the basis, we assume $m\leq n$.
Then every agent $a_i$ gets at most one item. With a truthful report this single item is the best remaining item according to $\bv_i$; hence, it is at least as good as any other bundle, so no manipulation is profitable.

For the induction step, we suppose the lemma holds for any number of items up to $m-1$ for some $m\geq n+1$. We aim to show it for $m$.

Consider first an agent $a_i$ for $i\geq 2$. The first $i-1$ picks are not affected at all by $a_i$. Hence, from $a_i$'s perspective, only the remaining $m-i+1$ items are relevant. As $m-i+1 < m$, the claim holds by the induction hypothesis.

It remains to prove the induction step for $a_1$ (the first picker). 
Renumber the items by descending order of $a_1$'s true values, and subject to that, by the fixed item-priority ordering.
Let $\bv_1'$ be a potential manipulation of agent $a_1$, and let $g_z$ be the item ranked first by $\bv_1'$. Denote by $A_1$ the bundle that $a_1$ gets by reporting truthfully $\bv_1$, and by $A_1'$ the bundle that $a_1$ gets by reporting  $\bv_1'$.
Note that these bundles depend on the valuations of agents $2,\ldots,n$.
Consider three cases.

\paragraph{Case 0: $g_z = g_1$.}
Under both $\bv_1'$ and $\bv_1$, $a_1$ takes $g_1$ in the first turn, so the set of remaining items in both cases is $M':=M\setminus\{g_1\}$.
As $|M'|=m-1$, we can apply the induction hypothesis to this set. By the induction hypothesis there are two options:
\begin{enumerate}[label=(\roman*), align=left]
    \item $\bv_1'$ is not safe, so $v_1(A_1\cap M') > v_1(A_1'\cap M')$ for \emph{some} valuations of agents $2,\ldots,n$ on $M'$.
    
    \item $\bv_1'$ is not profitable, so $v_1(A_1\cap M') \geq v_1(A_1'\cap M')$ for \emph{all} valuations of agents $2,\ldots,n$ on $M'$.
\end{enumerate}

As $a_1$ picks $g_1$ both by $v_1$ and by $v_1'$, we have $A_1  = (A_1\cap M') \cup \{g_1\}$ and 
$A_1'= (A_1' \cap M') \cup \{g_1\}$.
Hence, the same inequalities hold for the sets $A_1$ and $A_1'$, which concludes the proof for this case.

\paragraph{Case 1: $g_z\neq g_1$ and $v_{1,z} < v_{1,1}$.}
Note that we cannot use the same argument as in Case 0, as the set of items remaining after the first turn is different when $a_1$ reports $\bv_1$ or $\bv_1'$.

We show that $\bv_1'$ is not a safe manipulation by showing a specific set of valuations of agents $2,\ldots,n$ for which $v_1(A_1) > v_1(A_1')$.
Let $a_2$'s valuation be such that she ranks items $g_1,g_z$ as the top two (and her values for the remaining items is arbitrary). In this case, items $g_1,g_z$ are allocated in the first two turns, whether agent $a_1$ truthfully reports $\bv_1$ or misreports $\bv_1'$.
So in both the truthful and the manipulated scenarios, the set of remaining items is $M'' := M\setminus \{g_1,g_z\}$.
As $|M''|=m-2$, we can apply the induction hypothesis to this set.
By the induction hypothesis, there are two options:
\begin{enumerate}[label=(\roman*), align=left]
    \item $v_1(A_1\cap M'') > v_1(A_1'\cap M'')$ for \emph{some} valuations of agents $2,\ldots,n$  on $M''$. 
    We extend $\bv_2$ to $M$ by making $a_2$ rank $g_1,g_z$ first, and extend the valuations of agents $3,\ldots,n$ to $M$ arbitrarily.
    For these extended valuations, 	
    we again have $A_1 = (A_1\cap M'')\cup \{g_1\}$ and 
    $A_1' = (A_1'\cap M'')\cup \{g_z\}$. Hence, for these valuations,  $v_1(A_1) > v_1(A_1')$.	

    \item $v_1(A_1\cap M'')\geq v_1(A_1'\cap M'')$ for \emph{all} valuations of agents $2,\ldots,n$ on $M''$. 
    For all valuations of the other agents in which $a_2$ picks $g_1$ first, 
    we have $A_1 = (A_1\cap M'')\cup \{g_1\}$ and 
    $A_1' = (A_1'\cap M'')\cup \{g_z\}$. Hence, for these valuations,  $v_1(A_1) > v_1(A_1')$.
\end{enumerate}

In both cases, $\bv_1'$ is not safe.

\paragraph{Case 2: $g_z\neq g_1$ and $v_{1,z}=v_{1,1}$.}
Note that we cannot use the argument of Case 1:  at the first bullet $v_1(A_1\cap M'') > v_1(A_1'\cap M'')$, all we can say is that, for \emph{some} valuation profiles (those in which $a_2$ ranks $g_1,g_z$ first), the manipulation is not profitable; we cannot say that it not profitable for \emph{all} valuation profiles, nor that it is harmful for any valuation profile.

To use the induction hypothesis, we define an ``intermediate'' manipulation $\bv_1^\top$ as follows:
\begin{align*}
v^\top_{1,z} &= v_{1,z}+\epsilon;
\\
v^\top_{1,j} &= v_{1,j}  & j\neq z
\end{align*}
Denote by $A_1^\top$ the bundle that $a_1$ gets by reporting $v_1^\top$.
Under both $\bv_a^\top$ and $\bv_1'$, $a_1$ takes $g_z$ in the first turn. Similarly to Case 0, we can apply the induction hypothesis to the set of remaining items $M':=M\setminus\{g_z\}$ and to the valuation $v_1^\top$. There are two options:
\begin{enumerate}[label=(\roman*), align=left]
    \item $v_1^\top(A_1^\top\cap M') > v_1^\top(A_1'\cap M')$ for \emph{some} valuations of agents $2,\ldots,n$ on $M'$.
    
    \item $v_1^\top(A_1^\top\cap M') \geq v_1^\top(A_1'\cap M')$ for \emph{all} valuations of agents $2,\ldots,n$ on $M'$.
\end{enumerate}

As $v_1$ and $v_1^\top$ are identical on $M'$, in both cases (i) and (ii), the same inequalities hold for $v_1$.
Moreover,
As $a_1$ picks $g_z$ both by $v_1^\top$ and by $v_1'$, we have $A_1^\top = (A_1^\top\cap M') \cup \{g_z\}$ and 
$A_1'= (A_1' \cap M') \cup \{g_z\}$.
Hence, the same inequalities hold for the sets $A_1^\top$ and $A_1'$, regardless of how the valuations of agents $2,\ldots,n$ are extended to $M$. That is, at least one of the following holds:
\begin{enumerate}[label=(\roman*), align=left]
    \item $v_1(A_1^\top) > v_1(A_1')$ for \emph{some} valuations of agents $2,\ldots,n$ on $M'$;
    
    \item $v_1(A_1^\top) \geq v_1(A_1')$ for \emph{all} valuations of agents $2,\ldots,n$ on $M'$.
\end{enumerate}

We now consider each of these cases in turn.

For case (i), denote by $\bv_2,\ldots,\bv_n$ some valuations of agents $2,\ldots,n$ on $M'$ for which $v_1(A_1^\top) > v_1(A_1')$. 
We can assume w.l.o.g. that these valuations induce a strict ranking. This is because we can perturb the values for items with the same value such that items with a higher priority have slightly higher values, which does not affect the round-robin procedure at all.
Hence, we extend these valuations to valuations $\bv^+_2,\ldots,\bv^+_n$ on $M$, by ``inserting'' $g_z$ just after $g_1$.

Now, let us compare $A_1^\top$ to $A_1$.
Consider the two round-robin processes for the two profiles $(\bv_1,\bv_2^+,\ldots,\bv_n^+)$ and $(\bv_1^\top,\bv_2^+,\ldots,\bv_n^+)$.
The two processes differ at exactly two turns: the first turn where agent $a_1$ takes $g_1$ in the former process and $g_z$ in the latter process, and one turn in the middle where some agent (which may be $a_1$, or some other agent) takes $g_z$ in the former process and this same agent takes $g_1$ in the latter process.
Thus, we have either (a) $A_1=A_1^\top$ or (b) $A_1\setminus A_1^\top=\{g_1\}$ and $A_1^\top\setminus A_1=\{g_z\}$.
In both cases, $v_1(A_1)=v_1(A_1^\top)$, as $v_{1,1}=v_{1,z}$.

Combining this inequality with the inequality in (i) gives $v_1(A_1) > v_1(A_1')$, which implies that reporting $\bv_1'$ is harmful for $\bv_2^+,\ldots,\bv_n^+$. This is sufficient for concluding that $\bv_1'$ is not safe.

For case (ii) we show that, if $\bv_1'$ is a profitable manipulation, then it cannot be safe.
Assume that $\bv_1'$ is profitable, and let 
$\bv_2,\ldots,\bv_n$ be some valuations of agents $2,\ldots,n$ on $M$ for which $v_1(A_1') > v_1(A_1)$
(as before, we assume these valuations induce strict rankings).
By the inequality in (ii), this implies 
that also $v_1(A_1^\top) > v_1(A_1)$.
Consider again the two round-robin processes for the two profiles $(\bv_1,\bv_2,\ldots,\bv_n)$ and $(\bv_1^\top,\bv_2,\ldots,\bv_n)$.
The only difference between them is in the ranking of agent 1 to the items $g_1,\ldots,g_z$, 
which is $1,2,\ldots,z-1,z$ in the former case 
and $z,1,2,\ldots,z-1$ in the latter case.

Construct new valuations 
$\bv_2^\pi,\ldots,\bv_n^\pi$ that are permuted version of $\bv_2,\ldots,\bv_n$ under the permutation $\pi$ that maps $(1,2,\ldots,z-1,z,z+1,\ldots,m)$ to $(z,1,2,\ldots,z-1,z+1,\ldots,m)$.
Denote by $B_1$ the outcome of round-robin for the profile $(\bv_1,\bv_2^\pi,\ldots,\bv_n^\pi)$ and
by $B_1^\top$ the outcome of round-robin for the profile
$(\bv_1^\top,\bv_2^\pi,\ldots,\bv_n^\pi)$.
By symmetry, $B_1$ is the same as $A_1^\top$ up to permuting  $1,2,\ldots,z-1,z$ to $z,1,2,\ldots,z-1$,
and similarly 
$B_1^\top$ is the same as $A_1$ up to permuting  $1,2,\ldots,z-1,z$ to $z,1,2,\ldots,z-1$.
But since $v_1(g_1)=\cdots = v_1(g_z)$, this implies
$v_1(B_1) = v_1(A_1^\top)$ and 
$v_1(B_1^\top) = v_1(A_1)$.
Combining these equations with $v_1(A_1^\top) > v_1(A_1)$
yields
$v_1(B_1) > v_1(B_1^\top)$;
combining with the inequality in (ii) 
yields 
$v_1(B_1) > v_1(B_1')$.
Hence, reporting $\bv_1'$ is harmful for 
$\bv_2^\pi,\ldots,\bv_n^\pi$. 
This is sufficient for concluding that $\bv_1'$ is not safe.
\end{proof}

\end{toappendix}

\begin{lemmarep}\label{thm:RR-strict-degree-n-1}
     When all agents have strict preferences, the RAT-degree of Round-Robin is $(n-1)$.
\end{lemmarep}

\begin{proofsketch}
    The challenging part is to show that it is not $(n-2)$-known-agents safely-manipulable.
    The idea is to track the set of available items after the manipulator's turn — once under truthful reporting and once under misreporting — up to the turn of the unknown agent.
    We then show that there always exists a valuation for the unknown agent that causes the two bundles to become equal by the time the manipulator plays again.
    Since preferences are strict, the manipulator strictly loses in her first turn, and from that point on faces identical sets of available items.
    By an inductive argument, she can only lose from this point on.
\end{proofsketch}

\begin{toappendix}
    
To prove that the RAT-degree is exactly $(n-1)$,  \Cref{lemma:RR-strict-upper} shows that the degree is at most $(n-1)$, while \Cref{lemma:RR-strict-lower} shows that it is at least $(n-1)$.

\begin{lemma}\label{lemma:RR-strict-upper}
    When agents report strict preferences, round-robin is $(n-1)$-known-agents safely-manipulable.
\end{lemma}

\begin{proof}
    Let agent $a_1$ be the agent who picks first according to the order $\pi$, and let $v_{1,1} > v_{1,2} > \cdots > v_{1,m}$ be her valuation for the items.
    Suppose agent $a_1$ knows that every other agent $j$ (i.e., every $j \neq i$) reports a valuation in which $v_{j,2} > v_{j,3} > \cdots > v_{j,m} > v_{j,1}$.
    We show that agent $a_1$ has a safe and profitable manipulation by reporting $v_{1,2}' > v_{1,1}' > v_{1,3}'> \cdots > v_{1,m}'$.

    When agent $a_1$ reports truthfully, in the first round she takes item $g_1$. The other agents then take the items $\{g_2,\ldots, g_m\}$, so in the second round, agent $a_1$ takes $g_{n+1}$.

    However, when agent $a_1$ misreports, in the first round she takes item $g_2$. As a result, the other agents take the item $\{g_1,g_3,\ldots, g_{m+1}\}$. This means that in the second round agent $a_1$ takes $g_{1}$.

    Note that after agent $a_1$’s second turn, the set of available items is exactly the same in both cases. Therefore, the rest of the allocation proceeds identically, and the only difference between agent $a_1$’s resulting bundles is in the first two items.

    Since agent $a_1$ strictly prefers item $g_2$ over $g_{n+1}$ then the manipulation is profitable. 
    The manipulation is vacuously safe as there are no unknown agents. 

\end{proof}

\begin{lemma}\label{lemma:RR-strict-lower}
    When agents report strict preferences, round-robin is \emph{not} $(n-2)$-known-agents safely-manipulable.
\end{lemma}

\begin{proof}    
    We prove the lemma by induction on the number of items, $m$.

    For the base case, notice that when $m \leq n$, the round-robin mechanism is strategy-proof, which implies that it is not $(n-2)$-known-agents safely-manipulable.
    Next, assume the lemma holds for any instance with fewer than $m$ items, and we shall prove it for $m$ items.

    We consider the first agent in the picking order $\pi$, denoted $a_1$, as after her turn, the remaining agents face a problem with $m-1$ items, to which the induction hypothesis applies.
    
    Let $a_i \neq a_1$ be the only agent whose preferences are unknown to $a_1$, and let $K := N \setminus \{a_1, a_i\}$ be the set of known agents.
    We show that for any manipulation by $a_1$ given $K$, there exists a valuation for $a_i$ such that the manipulation is either not profitable or not safe.

    Consider $a_1$'s first turn. If $a_1$ picks the same item under both her true and manipulated valuations, we are done: her subsequent turns involve fewer than $m$ items, and by the induction hypothesis, she cannot benefit from manipulating at that point given $(n-2)$-known-agents.

    Thus, we assume that in her first turn, $a_1$ picks item $g$ under her true valuation, and item $g' \neq g$ under the manipulated one. 
    Since preferences are strict, we know that $v_{1,g} > v_{1,g'}$, so her tentative utility after this turn is strictly lower under the manipulation.

    Let $S$ and $S'$ be the sets of available items after the truthful and manipulated choices, respectively. These sets differ by exactly two items, one item in each direction: $S \setminus S' = \{g\}$ and $S' \setminus S = \{g'\}$.
    We call this the \emph{fragile asymmetry property}.

    Intuitively, if this property holds when the turn of the unknown agent arrives, she can "break the asymmetry" between the two scenarios: there exists a valuation for $a_i$ such that after his turn the two resulting sets of available items become identical.
    As a result, all subsequent turns until the next turn of $a_1$ unfold in exactly the same way in both scenarios. This means that $a_i$ only loses from the manipulation.
    We describe the specific valuation of the known agent later. 
    
    We now want to track the evolution of the sets $S$ and $S'$ up from $a_1$'s turn until the turn of the unknown agent $a_j$, to show that either the fragile asymmetry property is preserved through the turns of the known agents or that the sets become identical even sooner.
    When an agent paces two sets that satisfy the fragile asymmetry property, there are four possibilities:
    \begin{enumerate}
        \item The agent picks the same item in both cases (her most favorite available item is in $S \cap S'$). 
        
        In this case, both $S$ and $S'$ are reduced by the same item, and the property is preserved.
    
        \item The agent prefers the unique item in $S \setminus S'$ when choosing from $S$, and the unique item in $S' \setminus S$ when choosing from $S'$. 
        
        In this case, both sets become identical after the pick.
    
        \item The agent prefers the unique item in $S \setminus S'$ when choosing from $S$, but a common item when choosing from $S'$. 
        
        In this case, $S$ is reduced by its unique item, and $S'$ is reduced by one of the common items — so, the updated sets still satisfy the property (but with different unique items). \eden{I'm not sure this is clear}
    
        \item The agent prefers a common item when choosing from $S$, but the unique item in $S' \setminus S$ when choosing from $S'$. 

        As before, the updated sets still satisfy the property - $S$ is reduced by one of the common items, and $S'$ is reduced by its unique item.
    \end{enumerate}

    When the turn of the known-agent $a_i$ arrives, the two sets of available items are either identical or satisfy the  fragile asymmetry property.
    If they are identical, then any valuation of agent $a_i$ will work. \eden{is it clear enough?}
    Otherwise, let $s_1$ and $s_2$ be the unique items in $S_1$ and $S_2$, respectively. 
    There are different valuation for $a_i$ that will "break the asymmetry"---consider any valuation in which the values of $s_1$ and $s_2$ are higher than the values of all the common items in both sets. 
    For instance, $v_{i,s_1} = v_{i,s_2} = 1$ and $v_{i,j} =0$ for all other items $j \in G \setminus \{s_1, s_2\}$.
\end{proof}

\end{toappendix}

\eden{================ NEW ===========}
\subsection{Max-Nash-Welfare item allocation}\label{sec:nash-item-alloc}

The Max-Nash-Welfare rule is defined by finding an allocation $(A_1,\ldots,A_n)$ that maximizes the Nash welfare $\prod_{i=1}^nv_i(A_i)$.
\er{\citet{caragiannis2019unreasonable} do not specify how to break ties when there is more than one allocation with the same Nash welfare. Here, we assume that ties are broken by a pre-specified agent priority. Specifically,}
when the optimal Nash welfare is positive and there are multiple allocations that maximize the Nash welfare, the mechanism outputs the allocation in the following lexicographical way: among all these allocations, consider those that maximize agent $a_1$'s valuation, and subject to this, maximize agent $a_2$'s valuation, and so on.
When the optimal Nash welfare is $0$ (i.e., at least one agent's valuation is $0$ for any allocation), there are many natural variants.
Before describing them, let 
$$\mathcal{S}=\{S\subseteq\{a_1,\ldots,a_n\}\mid\text{there exists an allocation where agents in }S\text{ have positive valuations}\},$$
and
$$\mathcal{S}^{\max}=\{S\in\mathcal{S}\mid |S|\geq |S'|\mbox{ for any }S'\in\mathcal{S}\}.$$
\begin{itemize}
\item[MNW1:] Choose $S\in \mathcal{S}^{\max}$ with lexicographically optimal indices of agents (i.e., include $a_1\in S$ if there exists $S\in\mathcal{S}^{\max}$ that contains $a_1$, and, subject to this, include $a_2\in S$ if possible, and so on). Then find an allocation that maximizes $\prod_{a_i\in S}v_i(A_i)$. 
\item[MNW2:] Find the value of $\max_{S\in\mathcal{S}^{\max}}\max_{(A_1,\ldots,A_n)}\prod_{a_i\in S}v_i(A_i)$. Find a lexicographically optimal $S$ where there exists an allocation such that $\prod_{a_i\in S}v_i(A_i)$ equals this value.
Output the corresponding allocation.
\item[MNW3:] Normalize the valuations of all agents such that $\sum_{j=1}^mv_{i,j}=1$ for each agent $a_i$. Then perform MNW2.
\end{itemize}

In words, all versions output an allocation that maximizes the number of agents with positive valuations.
MNW1 performs the tie-breaking for the agent set first, and then proceeds to maximize the product
\er{(similarly to the definition in \citet{caragiannis2019unreasonable}).}
MNW2 maximizes the product first, and then performs the tie-breaking among those agent sets that maximize the product.
MNW3 is the variant of MNW2 by first normalizing the valuations
\er{(note that normalization has no effect on MNW1 because the selection of $S$ does not depend on the agents' values).}

\begin{theorem}
For any $m\geq2$ and any $n$, the RAT-degree of MNW2 is $0$ regardless of tie-breaking.

The RAT-degrees of MNW1 and MNW3 \er{with lexicographic agent-priority tie-breaking} is at most $1$.
\end{theorem}
\begin{proof}
We first show that the RAT-degree for MNW2 is $0$.
Suppose agent $a_1$'s valuation is such that $v_{1,1}=1$ and $v_{1,j}=0$ for $j\neq 1$.
It is easy to see that reporting $\bv_1'$ with $v_{1,1}'=2$ and $v_{1,j}'=0$ for $j\neq1$ is a safe and profitable manipulation.

Now, consider the other two variants.
Let agent $a_1$'s valuation be such that $v_{1,1}=2/3$, $v_{1,2}=1/3$, and $v_{1,j}=0$ for $j\notin\{1,2\}$.
Let agent $a_2$'s valuation be such that $v_{2,1}=1-\varepsilon$, $v_{2,2}=\varepsilon$, and $v_{2,j}=0$ for $j\notin\{1,2\}$, where $\varepsilon>0$ is a very small real number.
Suppose agent $a_1$ sees agent $a_2$'s valuation vector.
We show that it is safe and profitable for agent $a_1$ to report $\bv_1'$ with $v_{1,1}'=1$ and $v_{1,j}=0$ for $j\neq1$.

We first show that the manipulation is profitable under both MNW1 and MNW3.
Suppose the remaining $n-2$ agents value all items at $0$.
It is straightforward to check that, both with and without manipulation, $S = \{a_1,a_2\}$; agent $a_1$ receives $\{g_2\}$ without manipulation and $\{g_1\}$ with manipulation.

To show that the manipulation is safe, it suffices to show that, under both MNW1 and MNW3,
\begin{enumerate}
    \item agent $a_1$ receives at most one item from $\{g_1,g_2\}$ without manipulation, and
    \item agent $a_1$ surely receives $g_1$ with manipulation.
\end{enumerate}
To see (1), if agent $a_1$ receives both $g_1$ and $g_2$, then the valuation of agent $a_2$ becomes $0$.
We can then adjust the allocation by giving one of $g_1$ and $g_2$ to agent $a_2$, which strictly enlarges the set of agents with positive valuations, contradicting $S\in\mathcal{S}^{\max}$.

To see (2), consider the valuation profile where agent $a_1$ manipulates. 

In MNW1,
if agent $a_1$ does not receive $g_1$, her valuation becomes $0$.
Adjusting the allocation by transferring $g_1$ from some other agent $a_i$'s bundle (with $i>1$) to $a_1$'s bundle will either make $S$ larger or improve $S$ lexicographically.

In MNW3, suppose agent $a_1$ does not receive $g_1$.
If transferring $g_1$ from some other agent $a_i$'s bundle (with $i>1$) to $a_1$'s bundle does not increase $|S|$, this reallocation will not reduce the product (since $v_{1,1}'=1$ and $v_{i,1}\leq 1$ under normalized valuations), and it will improve $S$ lexicographically if the product is unchanged.
\end{proof}

\er{
Given the superior fairness properties of the MNW rule, the following question is interesting.
\begin{open}\label{open:allocation-mnw}
	Are there any variants of MNW with a RAT degree larger than 1?
\end{open}
}

\subsection{Volatile Priority Goods Allocation: An EF1 Mechanism with RAT-degree $n-1$}
\label{sect:indivisible-EF1-n-1}
\begin{toappendix}
    \subsection{Volatile Priority Goods Allocation: An EF1 Mechanism with RAT-degree $n-1$}
\end{toappendix}

In this section, we propose a new EF1 mechanism that has RAT-degree $n-1$.
The mechanism starts by applying a given priority function~$\Gamma$ (to be defined shortly), which takes as input a valuation profile~$(v_1,\ldots,v_n)$ and outputs two agent indices, $i^+$ and $i^-$. Agent~$a_{i^+}$ is referred to as the \emph{mechanism-favored agent}, and agent~$a_{i^-}$ as the \emph{mechanism-unfavored agent} (the reason of both names will be clear soon).
Let~$\efallocations_{i^-}$ denote the set of all EF1 allocations~$(A_1,\ldots,A_n)$ in which agent~$a_{i^-}$ is not envied by any other agent; that is, for every agent~$a_i$, we have~$v_i(A_i) \geq v_i(A_{i^-})$. 
Note that~$\efallocations_{i^-}$ is non-empty: the allocation returned by the round-robin mechanism when~$a_{i^-}$ is placed last in the picking order~$\pi$ belongs to this set.
The mechanism outputs the allocation in~$\efallocations_{i^-}$ that maximizes the valuation of the mechanism-favored agent, i.e., it selects an allocation that maximizes~$v_{i^+}(A_{i^+})$. When there
are multiple maximizers, the rule breaks the tie in an arbitrary consistent way.

We show that the priority function~$\Gamma$ can be constructed in a way that we refer to as \emph{volatile}: lacking knowledge of even a single agent's valuation makes manipulations risky. That is, for any agent $a_i$, any manipulation for $a_i$, any subset of $(n-2)$-known agents, and any valuations for them; there exists a valuation for the only remaining known-agent agent $a_j$ such that: (1) when~$a_i$ reports truthfully, she is selected as the mechanism-favored agent; (2) when~$a_i$ misreports, she becomes the mechanism-unfavored agent; and (3) this change in priority results in a strictly lower utility for~$a_i$.
A formal definition is provided in the appendix. 

\begin{toappendix}
    \begin{definition}
    A selection rule $\Gamma$ is called \emph{volatile}
    if 
    for any six indices of agents $i,j,i^+,i^-,i^{+'},i^{-'}$ with $i\neq j$, $i^+\neq i^-$, and $i^{+'}\neq i^{-'}$, any good $g_{\ell^\ast}\in G$, any set of $n-2$ valuation profiles $\{v_k\}_{k\notin \{i,j\}}$, and any two reported valuation profiles $v_i,v_i'$ of agent $a_i$ with $v_i\neq v_i'$ (i.e., $v_{i,\ell}\neq v_{i,\ell}'$ for at least one good $g_{\ell}$), there exists a valuation function $v_j$ of agent $a_j$ such that
    \begin{itemize}
    \item $v_{j,\ell^\ast}>0$, and $v_{j,\ell}=0$ for any $\ell\neq\ell^\ast$, and
    \item $\Gamma$ outputs $i^+$ and $i^-$ for the valuation profile $\{v_k\}_{k\notin\{i,j\}}\cup\{v_i\}\cup\{v_j\}$;
    \item $\Gamma$ outputs $i^{+'}$ and $i^{-'}$ for the valuation profile $\{v_k\}_{k\notin\{i,j\}}\cup\{v_i'\}\cup\{v_j\}$.
\end{itemize}
\end{definition}

In other words, a manipulation of agent $i$ from $v_i$ to $v_i'$ can affect the output of $\Gamma$ in any possible way (from any pair $i^+,i^-$ to any other pair $i^{+'},i^{-'}$), depending on the report of agent $j$.

\end{toappendix}

Any volatile  priority rule $\Gamma$ can be used by our mechanism; to conclude its description we show that such a rule exists.

\begin{proposition}
\label{prop:volatile}
    There exists a volatile priority rule $\Gamma$ for goods allocation.
\end{proposition}
\begin{proof}
    The rule $\Gamma$  first finds the maximum value among all agents and all goods: $\displaystyle v^\ast := \max_{i\in[n],\ell\in[m]}v_{i,\ell}$.
    It then views the value $v^\ast$ as a binary string that encodes the following information:
    \begin{itemize}
        \item the index $i$ of an agent $a_i$; 
        \item  \rmark{and the index $\ell$ of a good $g_\ell$} \eden{I think it was accidentally removed before submission}
        \item a non-negative integer $t$,
        \item two non-negative integers $a,b$, between $0$ and ${n \choose 2}$.
    \end{itemize}
    We append $0$'s as most significant bits to $v^\ast$ if the length of the binary string is not long enough to support the format of the encoding.
   If the encoding of $v^\ast$ is longer than the length enough for encoding the above-mentioned information, we take only the least significant bits in the amount required for the encoding.

    The mechanism-favored agent $a_{i^+}$ and the mechanism-unfavored agent $a_{i^-}$ are then decided in the following way.
    Let $s\in\{0,1\}$ be the bit at the $t$-th position of the binary encoding of the value $v_i(g_\ell)$. \eden{I think the definition of $\ell$ is missing}

    Let $p := (a\cdot s + b) \bmod {n \choose 2}$.
    Each value of $p$ corresponds to a pair of different agents $(i^+, i^-)$.
    

    To see that $\Gamma$ is volatile, suppose $v_i$ and $v_i'$ are different in the $t$-th bits of their binary encoding.
    We construct a value $v^*$ that encodes the integers
    $i,t,a,b$ where
    \begin{enumerate}
        \item the $t$-th bit of $v_{i}$ is $s$ and the $t$-th bit of $v_{i}'$ is $s'$ for $s\neq s'$;
        \item The pair $(i^+,i^-)$ corresponds to the integer $(a\cdot s + b) \bmod {n \choose 2}$.
        \item The pair $(i^{+'},i^{-'})$ corresponds to the integer $(a\cdot s' + b) \bmod {n \choose 2}$.
    \end{enumerate}
    (1) can always be achieved by some encoding rule.
    To see (2) and (3) can always be achieved, assume $s=1$ and $s'=0$ without loss of generality.
We can then take $b := $ the integer corresponding to the pair $(i^{+'},i^{-'})$, and $a := - b + $ the integer corresponding to the pair $(i^{+},i^{-})$, modulo ${n\choose 2}$.
    
    
    We then construct a valuation $v_j$ such that $v_{j,\ell^\ast}$ is the largest and is equal to $v^*$.
    In case $v^*$ is not large enough, we increase it as needed by adding most significant digits.
\end{proof}

\begin{remark}
The proof of \Cref{prop:volatile} requires that the mechanism does \emph{not} normalize the valuations, nor place any upper bound on the reported values. Suppose there were an upper bound $V$ on the value. $V$ encodes some agent $i$, bit number $t$, and integers $a,b$. It is possible that these numbers give the highest priority to agent $i$. In that case, agent $i$ could manipulate by reporting the value $V$.
\end{remark}

It is easy to see that the mechanism always outputs an EF1 allocation; we further prove that:
\begin{theoremrep}\label{thm:volatile-good-alloc-degree-n-1}
    The RAT-degree of the Volatile Priority Goods Allocation is $n-1$.
\end{theoremrep}

\begin{proofsketch}
    For every profitable manipulation by $a_i$, and for every unknown agent $a_j$, the volatility of $\Gamma$ implies that, for some possible valuation $v_j$, 
    a truthful report by $a_i$ leads to $a_i$ being the favored agent and $a_j$ being the unfavored agent, whereas the manipulation leads to 
    $a_i$ being the unfavored agent and $a_j$ being the favored agent. We use this fact  to prove that the manipulation may be harmful for $a_i$.
\end{proofsketch}

\begin{toappendix}
    Before we proceed to the proof, we first define some additional notions.
We say that $(A_1,\ldots,A_n)$ is a \emph{partial allocation} if $A_i\cap A_j=\emptyset$ for any pair of $i,j\in[n]$ and $\bigcup_{i=1}^nA_i\subseteq G$.
The definition of EF1 can be straightforwardly extended to partial allocations.
Given a possibly partial allocation $(A_1,\ldots,A_n)$, we say that agent $a_i$ \emph{strongly envies} agent $a_j$ if $v_i(A_i)<v_i(A_j\setminus\{g\})$ for any $g\in A_j$, i.e., the EF1 criterion from $a_i$ to $a_j$ fails.
Given $t \in[n]$, we say that a (possibly partial) allocation $(A_1,\ldots,A_n)$ is \emph{EF1 except for $t$} if for any pair $i,j\in[n]$ with $i\neq t$ we have $v_i(A_i)\geq v_i(A_j\setminus \{g\})$ for some $g\in G$.
In words, the allocation is EF1 except that agent $a_t$ is allowed to strongly-envy others.

We first prove some lemmas which will be used later.
\begin{lemma}\label{prop:partialtocomplete}
    Fix a valuation profile. Let $(A_1,\ldots,A_n)$ be a partial EF1 allocation. There exists a complete EF1 allocation $(A_1^+,\ldots,A_n^+)$ such that $v_i(A_i^+)\geq v_i(A_i)$ for each $i\in [n]$.
\end{lemma}
\begin{proof}
    Construct the envy-graph for the partial allocation $(A_1,\ldots,A_n)$ and then perform the \emph{envy-graph procedure} proposed by~\citet{lipton2004approximately} to obtain a complete allocation $(A_1^+,\ldots,A_n^+)$.
    The monotonic property of the procedure directly implies this proposition.
\end{proof}

\begin{lemmarep}\label{prop:maximuminexception}
    Fix a valuation profile and an arbitrary agent $a_{t}$. Let $\efallocations$ be the set of all complete EF1 allocations. Let ${\efallocations}^{-t}$ be the set of all possibly partial allocations that are EF1 except for possibly $a_t$. The allocation in $\efallocations$ that maximizes agent $a_t$'s utility is also the one in ${\efallocations}^{-t}$ that maximizes $a_t$'s utility.
\end{lemmarep}
In other words, if $a_t$ gets the maximum possible value subject to EF1, she cannot get a higher value by agreeing to give up the EF1 guarantee for himself.
This claim is trivially true for share-based fairness notions such as proportionality, but quite challenging to prove for EF1; see appendix.

\begin{proof}
    Assume $t=1$ without loss of generality.
    The allocation space $\efallocations$ is clearly a subset of ${\efallocations}^{-1}$.
    Let $\mathcal{Y}\subseteq {\efallocations}^{-1}$ be the set of all possibly partial allocations with agent $a_1$'s utility maximized, and assume $\mathcal{Y}\cap\efallocations=\emptyset$ for the sake of contradiction.
    Let $(A_1,\ldots,A_n)$ be  an allocation that minimizes $|\bigcup_{i=1}^nA_i|$ (minimizes the total number of goods allocated) among all allocations in 
    $\mathcal{Y}$.

    For each $i\neq 1$, agent $a_i$ will strongly envy some other agent if an arbitrary good is removed from $A_i$, for otherwise, the minimality is violated.
    We say that an agent $a_i$ \emph{champions} agent $a_j$ if the following holds.
    \begin{itemize}
        \item $a_i$ strongly envies $a_j$ for $i=1$;
        \item for $i\neq 1$, let agent $a_i$ removes the most valuable good from each $A_k$ (for $k\neq i$) and let $A_k^-$ be the resultant bundle; then the championed agent, agent $a_j$, is defined by the index $k$ with the maximum $v_i(A_k^-)$.
    \end{itemize}
    We then construct a \emph{champion graph} which is a directed graph with $n$ vertices where the vertices represent the $n$ agents and an edge from $a_i$ to $a_j$ represents that agent $a_i$ champions agent $a_j$.
    By our definition, each vertex in the graph has at least one outgoing edge, so the graph must contain a directed cycle $C$.

    Consider a new allocation $(A_1',\ldots,A_n')$ defined as follows. For every edge $(a_i,a_j)$ in $C$, let agent $a_i$ remove the most valuable good from $A_{j}$ and then take the bundle.
    We will show that $v_1(A_1')\geq v_1(A_1)$ and $(A_1',\ldots,A_n')\in {\efallocations}^{-1}$ (and so $(A_1',\ldots,A_n')\in \mathcal{Y}$), which will contradict the minimality of $(A_1,\ldots,A_n)$.

    It is easy to see $v_1(A_1')\geq v_1(A_1)$. If agent $a_1$ is not in the cycle $C$, then her utility is unchanged. Otherwise, she receives a bundle that she previously strongly envies, and one good is then removed from the bundle. The property of strong envy guarantees that $v_1(A_1')>v_1(A_1)$.

    To show that $(A_1',\ldots,A_n')\in {\efallocations}^{-1}$, first consider any agent $a_i$ with $i\neq 1$ that is not in $C$.
    Agent $a_i$'s bundle is unchanged, and she will not strongly envy anyone else as before (as only item-removals happen during the update).
    
    Next consider any agent $a_i$ with $i\neq 1$ that is in $C$.
    Let $a_j$ be the agent such that $(a_i,a_j)$ is an edge in $C$.
    Let $A_j^-$ be the bundle with the most valuable good (according to $v_i$) removed from $A_j$.
    By our definition, agent $a_i$ receives $A_j^-$ in the new allocation.
    We will prove that agent $a_i$, by receiving $A_j^-$, does not strongly-envy any of the original bundles $A_k$, for any $k\in [n]$.
    
    Our definition of championship ensures that this is true for any $k\neq i$, as the new bundle of $a_i$ is at least as valuable for $a_i$ than every other bundle with an item removed.
    
    It remains to show that this holds for $k=i$.
    As we have argued at the beginning, in the original allocation $(A_1,\ldots,A_n)$, removing any item from $A_i$ would make agent $a_i$ strongly envy some other agent.
    By our definition of championship, when one good $g'$ is removed from $A_i$, agent $a_i$ strongly envies $a_j$, which implies $a_i$ thinks $A_j^-$ is more valuable than $A_i\setminus\{g'\}$.
    Therefore, in the new allocation, by receiving $A_j^-$, agent $a_i$ does not strongly envy $A_i$.
    
    We have proved that agent $a_i$, by receiving the bundle $A_j^-$, does not strongly envy any of the $n$ original bundles $A_1,\ldots,A_n$.
    Since the new allocation only involves item removals, agent $a_i$ does not strongly envy anyone else in the new allocation.
    
    Hence, the new allocation is in 
    $\mathcal{Y}$,
    which contradicts the minimality of $(A_1,\ldots,A_n)$.
\end{proof}

We are now ready to prove \Cref{thm:volatile-good-alloc-degree-n-1}.

\begin{proof}[Proof of \Cref{thm:volatile-good-alloc-degree-n-1}]
Let $i,j$ be two arbitrary agents.
Fix $n-2$ arbitrary valuations $\{v_k\}_{k\notin\{i,j\}}$ for the remaining $n-2$ agents.
Consider two arbitrary valuations for agent $a_i$, $v_i$ and $v_i'$, with $v_i\neq v_i'$, where $v_i$ is $a_i$'s true valuation.
We will show that switching from $v_i$ to $v_i'$ is not a safe manipulation.

Let $\ell^\ast$ be some good that $a_i$ values positively, that is, $v_{i,\ell^\ast}>0$.
By the volatility of $\Gamma$, we can construct the valuation of agent $a_j$ such that 
\begin{itemize}
    \item $v_{j,\ell^\ast}>0$, and $v_{j,\ell}=0$ for any $\ell\neq\ell^\ast$;
    \item if agent $a_i$ truthfully reports $v_i$, then agent $a_i$ is mechanism-favored and agent $a_j$ is mechanism-unfavored; if agent $a_i$ reports $v_i'$ instead, then agent $a_j$ is mechanism-favored and agent $a_i$ is mechanism-unfavored.
\end{itemize}

Let $(A_1,\ldots,A_n)$ be the allocation output by $\Psi$ when agent $a_i$ reports $v_i$ truthfully, and $(A_1',\ldots,A_n')$ be the allocation output by $\Psi$ when agent $a_i$ reports $v_i'$.
Our objective is to show that $v_i(A_i)>v_i(A_i')$.

Let us consider $(A_1',\ldots,A_n')$ first.
We know that $g_{\ell^\ast}\in A_j'$.
To see this, notice that $A_j'$ maximizes agent $a_j$'s utility as long as $g_{\ell^\ast}\in A_j'$.
In addition, there exists a valid allocation $(A_1',\ldots,A_n')$ output by $\Psi$ with $g_{\ell^\ast}\in A_j'$: consider the round-robin mechanism with agent $a_j$ be the first and agent $a_i$ be the last under the order $\pi$.

Consider a new allocation $(A_1'',\ldots,A_n'')$ in which $g_{l^*}$ is moved from $a_j$ to $a_i$, that is,
\begin{itemize}
\item $A_i''=A_i'\cup\{g_{\ell^\ast}\}$,
\item $A_j''=A_j'\setminus\{g_{\ell^\ast}\}$,
\item $A_k''=A_k'$ for $k\notin\{i,j\}$.
\end{itemize}
Notice that $(A_1'',\ldots,A_n'')$ is EF1 except for $i$: 
\begin{itemize}
    \item agent $a_j$ will not envy any agent $a_k$ with $k\neq i$ (as each bundle $A_k''$ has a zero value for $j$), and agent $a_j$ will not envy agent $a_i$ upon removing the item $g_{\ell^\ast}$ from $A_i''$;
    \item no other agent $k$ strongly envies agent $i$: given that no one envies agent $i$ in $(A_1',\ldots,A_n')$ (as agent $i$ is mechanism-unfavored), no one strongly envies agent $i$ in $(A_1'',\ldots,A_n'')$;
    \item no agent strongly envies agent $a_j$, as $A_j''\subsetneq A_j'$;
    \item any two agents in $N\setminus\{a_i,a_j\}$ do not strongly envy each other, as their allocations are not changed.
\end{itemize}

Now, consider $(A_1,\ldots,A_n)$, which is the allocation that favors agent $a_i$ when agent $a_i$ truthfully reports $v_i$.
We can assume $A_j=\emptyset$ without loss of generality.
If not, we can reallocate goods in $A_j$ to the remaining $n-1$ agents while keeping the EF1 property among the remaining $n-1$ agents (\Cref{prop:partialtocomplete}).
Agent $a_j$ will not strongly envy anyone, as removing the good $g_{\ell^\ast}$ kills the envy.
Thus, the resultant allocation is still EF1 and no one envies the empty bundle $A_j$.
In addition, by \Cref{prop:partialtocomplete}, the utility of each agent $a_k$ with $k\neq j$ does not decrease after the reallocation of $A_j$.

Let ${\yefallocations}$ be the set of all EF1 allocations of the item-set $G$ to the agent-set $N\setminus\{a_j\}$.
Let ${\yefallocations}^{-i}$ be the set of all possibly partial allocations of the item-set $G$ to the agent-set $N\setminus\{a_j\}$ that are EF1 except for agent $i$.
The above argument shows that $(A_1,\ldots,A_{j-1},A_{j+1},\ldots,A_n)$ is an allocation in ${\yefallocations}$ that maximizes agent $a_i$'s utility.
We have also proved that $(A_1'',\ldots,A_{j-1}'',A_{j+1}'',\ldots,A_n'')\in {\yefallocations}^{-i}$.
By \Cref{prop:maximuminexception}, we have $v_i(A_i)\geq v_i(A_i'')$.
In addition, we have $A_i''=A_i'\cup\{g_{\ell^\ast}\}$, $g_{\ell^\ast}\notin A_i'$, and $v_{i,\ell^\ast}>0$ (by our assumption), which imply $v_i(A_i'')>v_i(A_i')$.
Therefore, $v_i(A_i)>v_i(A_i')$.
\end{proof}
\end{toappendix}


\begin{remark}
Our Volatile Priority algorithm does not run in polynomial time, as it requires to maximize the utility of a certain agent subject to EF1. which is 
an NP-hard problem (e.g., the proof in Appendix A.2 of~\citet{barman2019fair} can easily imply this).
We do not know if a polynomial-time EF1 algorithm with a high RAT-degree exists.
\end{remark}

In particular, an interesting direction is to to consider the round-robin algorithm but with a volatile priority order—that is, where the picking order depends on the reported valuations.

\begin{open}\label{open-alloc-2}
    What is the RAT-degree of Volatile Priority Round Robin?
\end{open}

\section{Cake Cutting}
\label{sec:cake-cutting}
In this section, we study the \emph{cake cutting} problem: the allocation of a divisible heterogeneous resource to $n$ agents.
The cake cutting problem was proposed by~\citet{Steinhaus48,Steinhaus49}, and it is a widely studied subject in mathematics, computer science, economics, and political science.

In the cake cutting problem, the resource/cake is modeled as an interval $[0,1]$, and it is to be allocated among a set of $n$ agents $N=\{a_1,\ldots,a_n\}$.
An allocation is denoted by $(A_1,\ldots,A_n)$ where $A_i\subseteq[0,1]$ is the share allocated to agent $a_i$.
We require that each $A_i$ is a union of finitely many closed non-intersecting intervals, and, for each pair of $i,j\in[n]$, $A_i$ and $A_j$ can only intersect at interval endpoints, i.e., the measure of $A_i\cap A_j$ is $0$.

The true preferences $T_i$ of agent $a_i$ are given by  a \emph{value density function} $v_i:[0,1]\to\mathbb{R}_{\geq0}$ that describes agent $a_i$'s preference over the cake.
To enable succinct encoding of the value density function, we adopt the widely considered assumption that each $v_i$ is \emph{piecewise constant}: there exist finitely many points $x_{i0},x_{i1},x_{i2},\ldots,x_{ik_i}$ with $0=x_{i0}<x_{i1}<x_{i2}<\cdots<x_{ik_i}=1$ such that $v_i$ is a constant on every interval $(x_{i\ell},x_{i(\ell+1)})$, $\ell=0,1,\ldots,k_i-1$.
Given a subset $S\subseteq[0,1]$ that is a union of finitely many closed non-intersecting intervals, agent $a_i$'s value for receiving $S$ is then given by
$V_i(S)=\int_Sv_i(x)dx.$

Lastly, we define an additional notion, \emph{uniform segment}, which will be used throughout this section.
Given $n$ value density functions $v_1,\ldots,v_n$ (that are piecewise constant by our assumptions), we identify the set of points of discontinuity for each $v_i$ and take the union of the $n$ sets.
Sorting these points by ascending order, we let $x_1,\ldots,x_{m-1}$ be all points of discontinuity for all the $n$ value density functions.
Let $x_0=0$ and $x_m=1$.
These points define $k$ intervals, $(x_0,x_1),(x_1,x_2),\ldots,(x_{m-1},x_m)$, such that each $v_i$ is a constant on each of these intervals.
We will call each of these intervals a \emph{uniform segment}, and we will denote $X_t=(x_{t-1},x_{t})$ for each $t=1,\ldots,m$.


\paragraph{Results.}
Our results are summarized in \Cref{tab:goods-cake-cutting-sum}.
We start by considering the \emph{utilitarian} mechanism that outputs an allocation with the maximum social welfare.
We show that the RAT-degree of this mechanism is $0$.
Similar as it is in the case of indivisible goods, we also consider the normalized variant of this mechanism, and we show that the RAT-degree is $1$.
Both mechanisms are \emph{Pareto Efficient}, they always returns an allocation that is \emph{Pareto Optimal} --- there is no other allocation that all agents weakly prefer, and at least one agent strictly prefers.

We then consider several fair mechanisms that have been considered by~\citet{BU2023Rat}.
Their paper studies whether or not these mechanisms are risk-avoiding truthful (in our language, whether the RAT-degree is positive).
Here we provide a more fine-grained view.

Each mechanism is guaranteed to return an allocation satisfying at least one of the following properties:
\emph{Envy free (EF)} --- each agent weakly prefers their own bundle to the bundle of any other agent, \emph{Proportional} --- each agent receives at least $1/n$ of the total value of the cake, according to their own valuation; and \emph{Connected pieces} --- each agent receives a single connected piece.

We start by considering \emph{equal division mechanisms} --- where we evenly allocate each uniform segment $X_t$ to all agents. The output allocation is always envy-free and proportional. However, we show that the RAT-degree relies heavily on the order in which we allocate the pieces of each segment. Specifically, with fixed order (e.g., let agent $a_1$ get the left-most interval and agent $a_n$ get the right-most interval), \citet{BU2023Rat} already proved that it is not even RAT, which means that its RAT-degree is $0$. 
To avoid this type of manipulation, a different tie-breaking rule was considered by~\citet{BU2023Rat} (See Mechanism 3 in their paper).
We prove that its RAT-degree is $n-1$.
~\citet{tao2022existence} shows that \emph{no} EF mechanism can also be truthful, so $n-1$ is the best possible RAT-degree.
However, this mechanism requires quite many cuts on the cake, and has a very poor performance in terms of pareto efficiency --- it gives each agent a value of exactly $1/n$ and nothing more, which intuitively is the most inefficient allocation among the proportional ones.

%
    


%

We then focus on proportional mechanisms with connected
pieces. \citet{BU2023Rat} prove that the Dubins-Spanier’s Moving-Knife \cite{dubins1961cut} is not RAT, which in our terms means that its RAT-degree is $0$. There was hope for those proven by \citet{BU2023Rat} to have a positive RAT-degree:
Bu, Song, and Tao's moving knife \cite{BU2023Rat} (two variants - Mechanisms 4 and 5 in their paper); and Ortega and Segal-Halevi’s moving knife \cite{ortega2022obvious}.
We show that they all have a very low RAT-degree of~$1$.
These results invoke the following question:
\begin{open}\label{open-cake-1}
    Is there a proportional connected cake-cutting rule with RAT-degree at least $2$?
\end{open}

A natural candidate to consider is the classic Even–Paz algorithm~\cite{even1984note}; we conjecture that its RAT-degree is~$\lceil n/2 \rceil$, but the question remains open.
\begin{open}\label{open-cake-2}
    What is the RAT-degree of the Even-Paz algorithm?
\end{open} 


%

Lastly, we adapt our \emph{volatile priority} approach from \Cref{sect:indivisible-EF1-n-1}, and propose a new mechanism that always returns a proportional and Pareto-efficient allocation, with the best possible RAT-degree of $n-1$.
Unlike its counterpart, this mechanism runs in polynomial time.

Given this result, it is natural to ask if the fairness guarantee can be strengthened to envy-freeness.
A compelling candidate is the mechanism that always outputs allocations with maximum \emph{Nash welfare} --- the product of agents utilities.
It is well-known that such an allocation is EF and Pareto-efficient.
However, computing its RAT-degree turns out to be very challenging for us.
We conjecture the answer is $n-1$.
\begin{open}\label{open-cake-3}
    What is the RAT-degree of the maximum Nash welfare mechanism?
\end{open}


\begin{table}[h]
    \centering
    \begin{tabular}{|l|c|l|}
    \hline 
    Mechanism & RAT-Degree  & Properties\\
    \hhline{|=|=|=|}
    Utilitarian & $0$ & Pareto Efficient\\
    \hline 
    Equal Division With Fixed Order & \cellcolor{black!10}$0$ & Proportional, EF\\
    \hline
    Dubins-Spanier’s Moving-Knife & \cellcolor{black!10}  &\multirow{2}{*}{Proportional, Connected}\\ 
    \cite{dubins1961cut} & \multirow{-2}{*}{\cellcolor{black!10}$0$} &\\
    \hline
    \multirow{2}{*}{Normalized Utilitarian} & \multirow{2}{*}{$1$} & Pareto Efficient\\
     &  & (w.r.t. the normalized values)\\
    \hline
    Bu, Song, and Tao’s Moving Knife  
     & \multirow{2}{*}{$1$} & \multirow{2}{*}{Proportional, Connected} 
    \\\cite{BU2023Rat} (Mechanism 4 and 5) 
     & &  \\
    \hline
    Ortega and Segal-Halevi’s Moving Knife 
     & \multirow{2}{*}{$1$} & \multirow{2}{*}{Proportional, Connected}\\
      \cite{ortega2022obvious} 
     &  & \\
    \hline
    Bu, Song, and Tao's Equal Division
    & \multirow{2}{*}{$n-1$} & \multirow{2}{*}{Proportional, EF}\\
    \cite{BU2023Rat} (Mechanism 3)
    &  & \\
    \hline
    Volatile Priority*
    & $n-1$ & Proportional, Pareto Efficient   \\
    \hline
    \end{tabular}
    \caption{\centering Cake Cutting: Summary of Results.  (*) indicates that the mechanism was proposed in this paper. Gray cells indicate indicate results proven by \citet{BU2023Rat}.}
    \label{tab:goods-cake-cutting-sum}
\end{table}


\subsection{Utilitarian Cake Cutting}
\label{sect:cake-msw}

\begin{toappendix}
    \subsection{Utilitarian Cake Cutting}
\end{toappendix}

It is easy to find an allocation that maximizes the social welfare: for each uniform segment $X_t$, allocate it to an agent $a_i$ with the maximum $v_i(X_t)$.
When multiple agents have equally largest value of $v_i(X_t)$ on the segment $X_t$, we need to specify a tie-breaking rule.
However, as we will see later, the choice of the tie-breaking rule does not affect the RAT-degree of the mechanism.


It is easy to see that, whatever the tie-breaking rule is, the maximum social welfare mechanism is safely manipulable.
It is safe for an agent to report higher values on every uniform segment.
For example, doubling the values on all uniform segments is clearly a safe manipulation.

\begin{observation}
\label{obs:utilitarian-cake-cutting}
Utilitarian Cake-Cutting with any tie-breaking rule has RAT-degree $0$.
\end{observation}

We next consider the following variant of the maximum social welfare mechanism:
first rescale each $v_i$ such that $V_i([0,1])=\int_0^1v_i(x)dx=1$, and then output the allocation with the maximum social welfare.
We will show that the RAT-degree is $1$.
The proof is similar to the one for indivisible items (\Cref{thm:normalized-utilitarian-goods}) and is given in the appendix.

\begin{theoremrep}
When there are at least three agents,
Normalized Utilitarian Cake-Cutting with any tie-breaking rule has RAT-degree $1$.
\end{theoremrep}
\begin{proof}
    We assume without loss of generality that the value density function reported by each agent is normalized (as, otherwise, the mechanism will normalize the function for the agent).
    
    We first show that the mechanism is not $0$-known-agents safely-manipulable.
    Consider an arbitrary agent $a_i$ and let $v_i$ be her true value density function.
    Consider an arbitrary misreport $v_i'$ of agent $a_i$ with $v_i'\neq v_i$.
    Since the value density functions are normalized, there must exist an interval $(a,b)$ where $v_i'$ and $v_i$ are constant and $v_i'(x)<v_i(x)$ for $x\in(a,b)$.
    Choose $\varepsilon>0$ such that $v_i(x)>v_i'(x)+\varepsilon$.
    Consider the following two value density functions (note that both are normalized):
    $$v^{(1)}(x)=\left\{\begin{array}{ll}
        v_i'(x)+\varepsilon & \mbox{if }x\in(a,b) \\
        v_i'(x)-\varepsilon\cdot\frac{b-a}{1+a-b} & \mbox{otherwise}
    \end{array}\right. \quad\mbox{and}\quad v^{(2)}(x)=\left\{\begin{array}{ll}
        v_i'(x)-\varepsilon & \mbox{if }x\in(a,b) \\
        v_i'(x)+\varepsilon\cdot\frac{b-a}{1+a-b} & \mbox{otherwise}
    \end{array}\right..$$
    Suppose the remaining $n-1$ agents' reported value density functions are either $v^{(1)}$ or $v^{(2)}$ and each of $v^{(1)}$ and $v^{(2)}$ is reported by at least one agent (here we use the assumption $n\geq 3$).
    In this case, agent $a_i$ will receive the empty set by reporting $v_i'$.
    On the other hand, when reporting $v_i$, agent $a_i$ will receive an allocation that at least contains $(a,b)$ as a subset.
    Since $v_i$ has a positive value on $(a,b)$, reporting $v_i'$ is not a safe manipulation.

    We next show that the mechanism is $1$-known-agent safely-manipulable.
    \erel{Doesn't this part follow from the analogous result on indivisible items?}\biaoshuai{I think so}
    Suppose agent $a_1$'s true value density function is
    $$v_1(x)=\left\{\begin{array}{ll}
        1.5 & \mbox{if }x\in[0,0.5] \\
        0.5 & \mbox{otherwise}
    \end{array}\right.,$$ 
    and agent $a_1$ knows that agent $a_2$ reports the uniform value density function $v_2(x)=1$ for $x\in[0,1]$.
    We will show that the following manipulation of agent $a_1$ is safe and profitable.
    $$v_1'(x)=\left\{\begin{array}{ll}
        2 & \mbox{if }x\in[0,0.5] \\
        0 & \mbox{otherwise}
    \end{array}\right.$$
    Firstly, regardless of the reports of the remaining $n-2$ agents, the final allocation received by agent $a_1$ must be a subset of $[0,0,5]$, as agent $a_2$'s value is higher on the other half $(0.5,1]$.
    Since $v_1'$ is larger than $v_1$ on $[0,0,5]$, any interval received by agent $a_1$ when reporting $v_1$ will also be received if $v_1'$ were reported.
    Thus, the manipulation is safe.

    Secondly, if the remaining $n-2$ agents' value density functions are
    $$v_3(x)=v_4(x)=\cdots=v_n(x)=\left\{\begin{array}{ll}
        1.75 & \mbox{if }x\in[0,0.5] \\
        0.25 & \mbox{otherwise}
    \end{array}\right.,$$
    it is easy to verify that agent $a_1$ receives the empty set when reporting truthfully and she receives $[0,0.5]$ by reporting $v_1'$.
    Therefore, the manipulation is profitable.
\end{proof}


\subsection{Equal Division Cake Cutting}

\begin{toappendix}
    \subsection{Equal Division Cake Cutting}
\end{toappendix}

One natural envy-free mechanism is to evenly allocate each uniform segment $X_t$ to all agents.
Specifically, each $X_t$ is partitioned into $n$ intervals of equal length, and each agent receives exactly one of them.
It is easy to see that $V_i(A_j)=\frac1nV_i([0,1])$ for any $i,j\in[n]$ under this allocation, so the allocation is envy-free and proportional.

To completely define the mechanism, we need to specify the order of evenly allocating each $X_t=(x_{t-1},x_t)$ to the $n$ agents.
A natural tie-breaking rule is to let agent $a_1$ get the left-most interval and agent $a_n$ get the right-most interval.
Specifically, agent $a_i$ receives the $i$-th interval of $X_t$, which is $[x_{t-1}+\frac{i-1}n(x_t-x_{t-1}),x_{t-1}+\frac{i}n(x_t-x_{t-1})]$.
However, it was proved in~\citet{BU2023Rat} that the equal division mechanism under this ordering rule is safely-manipulable, i.e., its RAT-degree is $0$.
In particular, agent $a_1$, knowing that she will always receive the left-most interval in each $X_t$, can safely manipulate by deleting a point of discontinuity in her value density function if her value on the left-hand side of this point is higher.

To avoid this type of manipulation, a different ordering rule was considered by~\citet{BU2023Rat} (See Mechanism 3 in their paper): at the $t$-th segment, the $n$ equal-length subintervals of $X_t$ are allocated to the $n$ agents with the left-to-right order $a_t,a_{t+1},\ldots,a_n,a_1,a_2,\ldots,a_{t-1}$.
By using this ordering rule, an agent does not know her position in the left-to-right order of $X_t$ without knowing others' value density functions.
Indeed, even if only one agent's value density function is unknown, an agent cannot know the index $t$ of any segment $X_t$.
This suggests that the mechanism has a RAT-degree of $n-1$.

\begin{theoremrep}
    Consider the mechanism that evenly partitions each uniform segment $X_t$ into $n$ equal-length subintervals and allocates these $n$ subintervals to the $n$ agents with the left-to-right order $a_t,a_{t+1},\ldots,a_n,a_1,a_2,\ldots,a_{t-1}$. It has RAT-degree $n-1$ and always outputs envy-free allocations.
\end{theoremrep}
\begin{proofsketch}
    Envy-freeness is trivial: for any $i,j\in[n]$, we have $V_i(A_j)=\frac1nV_i([0,1])$.
    The general impossibility result in~\citet{tao2022existence} shows that no mechanism with the envy-freeness guarantee can be truthful, so the RAT-degree is at most $n-1$.

 To show that the RAT-degree is exactly $n-1$, we show that, if even a single agent is not known to the manipulator, it is possible that this agent's valuation adds discontinuity points in a way that the ordering in each uniform segment is unfavorable for the manipulator.
\end{proofsketch}
\begin{proof}
    Envy-freeness is trivial: for any $i,j\in[n]$, we have $V_i(A_j)=\frac1nV_i([0,1])$.
    The general impossibility result in~\citet{tao2022existence} shows that no mechanism with the envy-freeness guarantee can be truthful, so the RAT-degree is at most $n-1$.

    To show that the RAT-degree is exactly $n-1$, consider an arbitrary agent $a_i$ with true value density function $v_i$ and an arbitrary agent $a_j$ whose report is unknown to agent $a_i$.
    Fix $n-2$ arbitrary value density functions $\{v_k\}_{k\notin\{i,j\}}$ that are known by agent $a_i$ to be the reports of the remaining $n-2$ agents.
    For any $v_i'$, we will show that agent $a_i$'s reporting $v_i'$ is either not safe or not profitable.

    Let $T$ be the set of points of discontinuity for $v_1,v_2,\ldots,v_{j-1},v_{j+1},\ldots,v_n$, and $T'$ be the set of points of discontinuity with $v_i$ replaced by $v_i'$.
    If $T\subseteq T'$ (i.e., the uniform segment partition defined by $T'$ is ``finer'' than the partition defined by $T$), the manipulation is not profitable, as agent $a_i$ will receive her proportional share $\frac1nV_i([0,1])$ in both cases.
    
    It remains to consider the case where there exists a point of discontinuity $y$ of $v_i$ such that $y\in T$ and $y\notin T'$.
    This implies that $y$ is a point of discontinuity in $v_i$, but not in $v_i'$ nor in the valuation of any other agent.
    We will show that the manipulation is not safe in this case.

    Choose a sufficiently small $\varepsilon>0$ such that $(y-\varepsilon, y+\varepsilon)$ is contained in a uniform segment defined by $T'$.
    We consider two cases, depending on whether the ``jump'' of $v_i$ in its discontinuity point $y$ is upwards or downwards.
    
    \underline{Case 1:}  $\lim_{x\to y^-}v_i(x)<\lim_{x\to y^+}v_i(x)$.
    We can construct $v_j$ such that: 1) $y-\varepsilon$ and $y+\varepsilon$ are points of discontinuity of $v_j$, and 2) the uniform segment $(y-\varepsilon, y+\varepsilon)$ under the profile $(v_1,\ldots,v_{i-1},v_i',v_{i+1},\ldots,v_n)$ is the $t$-th segment where $n$ divides $t-i$ (i.e., agent $a_i$ receives the left-most subinterval of this uniform segment).
    Notice that 2) is always achievable by inserting a suitable number of points of discontinuity for $v_j$ before $y-\varepsilon$.
    Given that $\lim_{x\to y^-}v_i(x)<\lim_{x\to y^+}v_i(x)$, agent $a_i$'s allocated subinterval on the segment $(y-\varepsilon, y+\varepsilon)$ has value strictly less than $\frac1nV_i([(y-\varepsilon, y+\varepsilon])$.

    \underline{Case 2:} $\lim_{x\to y^-}v_i(x)>\lim_{x\to y^+}v_i(x)$.
    We can construct $v_j$ such that 1) $(y-\varepsilon, y+\varepsilon)$ is a uniform segment under the profile $(v_1,\ldots,v_{i-1},v_i',v_{i+1},\ldots,v_n)$, and 2) agent $a_i$ receives the right-most subinterval on this segment.
    In this case, agent $a_i$ again receives a value of strictly less than $\frac1nV_i([(y-\varepsilon, y+\varepsilon])$ on the segment $(y-\varepsilon, y+\varepsilon)$.

    We can do this for every point $y$ of discontinuity of $v_i$ that is in $T\setminus T'$.
    By a suitable choice of $v_j$ (with a suitable number of points of discontinuity of $v_j$ inserted in between), we can make sure agent $a_i$ receives a less-than-average value on every such segment $(y-\varepsilon,y+\varepsilon)$.
    Moreover, agent $a_i$ receives exactly the average value on each of the remaining segments, because the remaining discontinuity points of $T$ are contained in $T'$.
    Therefore, the overall utility of $a_i$ by reporting $v_i'$ is strictly less than $\frac1nV_i([0,1])$.
    Given that $a_i$ receives value exactly $\frac1nV_i([0,1])$ for truthfully reporting $v_i$, reporting $v_i'$ is not safe.
\end{proof}

Although the equal division mechanism with the above-mentioned carefully designed ordering rule is envy-free and has a high RAT-degree of $n-1$, it is undesirable in at least two aspects:
\begin{enumerate}
    \item it requires quite many cuts on the cake by making $n-1$ cuts on each uniform segment; this is particularly undesirable if piecewise constant functions are used to approximate more general value density functions.

    \item it is highly inefficient: each agent $a_i$'s utility is never more than her minimum proportionality requirement $\frac1nV_i([0,1])$;
\end{enumerate}

We handle point (1) in \Cref{sect:cake-connected} and point (2) in \Cref{sect:cake-Prop+PO}.


\subsection{Connected Cake-Cutting} \label{sect:cake-connected}

Researchers have been looking at allocations with \emph{connected pieces}, i.e., allocations with only $n-1$ cuts on the cake.
A well-known mechanism in this category is \emph{the moving-knife procedure}, which always outputs proportional allocations.
This mechanism was first proposed by~\citet{dubins1961cut}. It always returns a proportional connected allocation.
Unfortunately, it was shown by~\citet{BU2023Rat} that Dubins and Spanier's moving-knife procedure is safely-manipulable for some very subtle reasons. For clarity, we provide its description here. 

\paragraph{Dubins and Spanier's moving-knife procedure}
Let $u_i=\frac1nV_i([0,1])$ be the value of agent $a_i$'s proportional share.
In the first iteration, each agent $a_i$ marks a point $x_i^{(1)}$ on $[0,1]$ such that the interval $[0,x_i^{(1)}]$ has value exactly $u_i$ to agent $a_i$.
Take $x^{(1)}=\min_{i\in[n]}x_i^{(1)}$, and the agent $a_{i_1}$ with $x_{i_1}^{(1)}=x^{(1)}$ takes the piece $[0,x^{(1)}]$ and leaves the game.
In the second iteration, let each of the remaining $n-1$ agents $a_i$ marks a point $x_i^{(2)}$ on the cake such that $[x^{(1)},x_i^{(2)}]$ has value exactly $u_i$.
Take $x^{(2)}=\min_{i\in[n]\setminus\{i_1\}}x_i^{(2)}$, and the agent $a_{i_2}$ with $x_{i_2}^{(2)}=x^{(2)}$ takes the piece $[x^{(1)},x^{(2)}]$ and leave the game.
This is done iteratively until $n-1$ agents have left the game with their allocated pieces.
Finally, the only remaining agent takes the remaining part of the cake.
It is easy to see that each of the first $n-1$ agents receives exactly her proportional share, while the last agent receives weakly more than her proportional share; hence the procedure always returns a proportional allocation.

\begin{remark}\label{remark:direct-revelation}
Although such mechanism are usually described in an iterative interactive way that resembles an extensive-form game, we will consider the \emph{direct-revelation} mechanisms in this paper, where the $n$ value density functions are reported to the mechanism at the beginning.
In the above description of Dubins and Spanier's moving-knife procedure, as well as its two variants mentioned later,
by saying ``asking an agent to mark a point'', we refer to that the mechanism computes such a point based on the reported value density function.
In particular, we do not consider the scenario where agents can adaptively choose the next marks based on the allocations in the previous iterations.
\end{remark}

\citet{BU2023Rat} proposed a variant of the moving-knife procedure that is RAT.
In addition, they showed that another variant of moving-knife procedure proposed by~\citet{ortega2022obvious} is also RAT.\footnote{It should be noticed that, when $v_i$ is allowed to take $0$ value, tie-breaking needs to be handled very properly to ensure RAT. See \citet{BU2023Rat} for more details. Here, for simplicity, we assume $v_i(x)>0$ for each $i\in[n]$ and $x\in[0,1]$.}
Below, we will first describe both mechanisms and then show that both of them have RAT-degree $1$.

\paragraph{Ortega and Segal-Halevi's moving knife procedure}
The first iteration of Ortega and Segal-Halevi's moving knife procedure is the same as it is in Dubins and Spanier's.
After that, the interval $[x^{(1)},1]$ is then allocated \emph{recursively} among the $n-1$ agents $[n]\setminus\{a_{i_1}\}$.
That is, in the second iteration, each agent $a_i$ marks a point $x_i^{(2)}$ such that the interval $[x^{(1)}, x_i^{(2)}]$ has value exactly $\frac1{n-1}V_i([x^{(1)},1])$ (instead of $\frac1nV_1([0,1])$ as it is in Dubins and Spanier's moving-knife procedure).
The remaining part is the same: the agent with the left-most mark takes the corresponding piece and leaves the game.
After the second iteration, the remaining part of the cake is again recursively allocated to the remaining $n-2$ agents.
This is continued until the entire cake is allocated.

\paragraph{Bu, Song, and Tao's moving knife procedure}
Each agent $a_i$ is asked to mark all the $n-1$ ``equal-division-points'' $x_i^{(1)},\ldots,x_i^{(n-1)}$ at the beginning such that $V_i([x_i^{(t-1)},x_i^{(t)}])=\frac1nV_i([0,1])$ for each $t=1,\ldots,n$, where we set $x_i^{(0)}=0$ and $x_i^{(n)}=1$.
The remaining part is similar to Dubins and Spanier's moving-knife procedure:
in the first iteration, agent $i_1$ with the minimum $x_{i_1}^{(1)}$ takes $[0,x_{i_1}^{(1)}]$ and leave the game; in the second iteration, agent $i_2$ with the minimum $x_{i_2}^{(2)}$ among the remaining $n-1$ agents takes $[x_{i_1}^{(1)},x_{i_2}^{(2)}]$ and leave the game; and so on.
The difference to Dubins and Spanier's moving-knife procedure is that each $x_i^{(t)}$ is computed at the beginning, instead of depending on the position of the previous cut.

\begin{theorem}
    The RAT-degree of Ortega and Segal-Halevi's moving knife procedure is $1$.
\end{theorem}
\begin{proof}
    It was proved in~\citet{BU2023Rat} that the mechanism is not $0$-known-agent safely-manipulable.
    It remains to show that it is $1$-known-agents safely-manipulable.
    Suppose agent $a_1$'s value density function is uniform, $v_1(x)=1$ for $x\in[0,1]$, and agent $a_1$ knows that agent $a_2$ will report $v_2$ such that $v_2(x)=1$ for $x\in[1-\varepsilon,1]$ and $v_2(x)=0$ for $x\in[0,1-\varepsilon)$ for some very small $\varepsilon>0$ with $\varepsilon\ll\frac1n$.
    We show that the following $v_1'$ is a safe manipulation.
    $$v_1'(x)=\left\{\begin{array}{ll}
        1 & x\in[0,\frac{n-2}n] \\
        \frac2{n\varepsilon} & x\in[1-2\varepsilon,1-\varepsilon]\\
        0 & \mbox{otherwise}
    \end{array}\right.$$
    Before we move on, note an important property of $v_1'$: for any $t\leq\frac{n-2}n$, we have $V_1([t,1])=V_1'([t,1])$.

    Let $[a,b]$ be the piece received by agent $a_1$ when she reports $v_1$ truthfully.
    If $b\leq \frac{n-2}n$, the above-mentioned property implies that she will also receive exactly $[a,b]$ for reporting $v_1'$.
    If $b>\frac{n-2}n$, then we know that agent $a_1$ is the $(n-1)$-th agent in the procedure.
    To see this, we have $V_1([a,b])\geq\frac1n([0,1])=\frac1n$ by the property of Ortega and Segal-Halevi's moving knife procedure, and we also have $V_1([b,1])=1-b<\frac2n$.
    This implies $b$ cannot be the $1/k$ cut point of $[a,1]$ for $k\geq 3$.
    On the other hand, it is obvious that agent $a_2$ takes a piece after agent $a_1$.
    Thus, by the time agent $a_1$ takes $[a,b]$, the only remaining agent is agent $a_2$.

    Since there are exactly two remaining agents in the game before agent $a_1$ takes $[a,b]$, we have $V_1([a,1])\geq\frac2nV_1([0,1])=\frac2n$.
    This implies $a\leq\frac{n-2}n$ and $b=\frac{a+1}2\leq \frac{n-1}n$.
    On the other hand, by reporting $v_1'$, agent $a_1$ can then get the piece $[a,b']$ with $b'\in[1-2\varepsilon,1-\varepsilon]$.
    We see that $b'>b$. Thus, the manipulation is safe and profitable.
\end{proof}

\begin{theorem}
    The RAT-degree of Bu, Song, and Tao's moving knife procedure is $1$.
\end{theorem}
\begin{proof}
    It was proved in~\citet{BU2023Rat} that the mechanism is not $0$-known-agent safely-manipulable.
    The proof that it is $1$-known-agents safely-manipulable is similar to the proof for Ortega and Segal-Halevi's moving knife procedure, with the same $v_1,v_1'$ and $v_2$.
    It suffices to notice that the first $n-2$ equal-division-points are the same for $v_1$ and $v_1'$, where the last equal-division-point of $v_1'$ is to the right of $v_1$'s.
    Given that agent $a_2$ will always receive a piece after agent $a_1$, the same analysis in the previous proof can show that the manipulation is safe and profitable. 
\end{proof}

\subsection{Volatile Priority Cake Cutting: Proportional, Pareto-Optimal, with RAT-degree $n-1$}

\begin{toappendix}
    \subsection{Volatile Priority Cake Cutting: Proportional, Pareto-Optimal, with RAT-degree $n-1$}
\end{toappendix}
\label{sect:cake-Prop+PO}
In this section, we provide a mechanism with RAT-degree $n-1$ that always outputs proportional and Pareto-optimal allocations.
The mechanism uses some similar ideas as the one in \Cref{sect:indivisible-EF1-n-1}.

Here, the given priority function~$\Gamma$, still takes as input a valuation profile~$(v_1,\ldots,v_n)$ but outputs an order $\pi$ of the $n$ agents.
In the appendix, we formally adapt the volatility property (defined in~\Cref{sect:indivisible-EF1-n-1}) to the cake-cutting setting and prove that such function can be constructed.
\begin{toappendix}
    
\begin{definition}
A function $\Gamma$ (from the set of valuation profiles to the set of orders on agents) is called \emph{volatile} if for any two agents $a_i\neq a_j$ and any two orders $\pi$ and $\pi'$, any set of $n-2$ value density functions $\{v_k\}_{k\notin\{i,j\}}$, any value density function $\bar{v}_j$, and any two reported valuation profiles $v_i,v_i'$ of agent $a_i$ with $v_i\neq v_i'$, there exists a valuation function $v_j$ of agent $a_j$ such that
\begin{itemize}

    \item $v_j$ is a rescaled version of $\bar{v}_j$, i.e., there exists $\alpha$ such that $v_j(x)=\alpha\bar{v}_j(x)$ for all $x\in[0,1]$;
    
    \item $\Gamma$ outputs $\pi$ for the valuation profile $\{v_k\}_{k\notin\{i,j\}}\cup\{v_i\}\cup\{v_j\}$, and $\Gamma$ outputs $\pi'$ for the valuation profile $\{v_k\}_{k\notin\{i,j\}}\cup\{v_i'\}\cup\{v_j\}$.
\end{itemize}

In other words, a manipulation of agent $i$ from $v_i$ to $v_i'$ can affect the output of $\Gamma$ in any possible way (from any order $\pi$ to any order $\pi'$), depending on the report of agent $j$.
\end{definition}

\end{toappendix}

\begin{propositionrep}
    There exists a volatile priority rule $\Gamma$ for cake cutting.
\end{propositionrep}
\begin{proof}
    
    The function does the following.
    It first finds the maximum value among all the value density functions (overall all uniform segments): $v^\ast=\max_{i\in[n],\ell\in[m]}v_i(X_\ell)$.
    It then views $v^\ast$ as a binary string that encodes the following information:
    \begin{itemize}
        \item the index $i$ of an agent $a_i$,
        \item a non-negative integer $t$,
        \item two non-negative integers $a$ and $b$ that are at most $n!-1$.
    \end{itemize}
    
    We append $0$'s as most significant bits to $v^\ast$ if the length of the binary string is not long enough to support the format of the encoding.
    If the encoding of $v^\ast$ is longer than the length enough for encoding the above-mentioned information, we take only the least significant bits in the amount required for the encoding.

    The order $\pi$ is chosen in the following way.
    Firstly, we use an integer between $0$ and $n!-1$ to encode an order.
    Then, let $s$ be the $t$-th bit that encodes agent $a_i$'s value density function.
    The order is defined to be $as+b\bmod (n!)$.

We now prove that $\Gamma$ is volatile.
Suppose $v_i$ and $v_i'$ differ at their $t$-th bits, 
so that that the $t$-th bit of $v_i$ is $s$ and the $t$-th bit of $v_i'$ is $s'\neq s$.
We construct a number $v^*$ that encodes the index $i$, the integer $t$, and two integers $a,b$ 
such that $as+b\bmod (n!)$ encodes $\pi$ and $as'+b\bmod (n!)$ encodes $\pi'$.

Then, we construct $v_j$ by rescaling $\bar{v}_j$ such that the maximum value among all density functions is attained by $v_j$, and this number is exactly $v^{\ast}$, that is, $v^\ast=v_j(X_\ell)$ for some uniform segment $X_\ell$.
If the encoded $v^*$ is not large enough to be a maximum value, we enlarge it as needed by adding most significant bits.

By definition, $\Gamma$ returns $\pi$ when $a_i$ reports $v_i$ and returns $\pi'$ when $a_i$ reports $v_i'$.

\end{proof}

Let $\propallocations$ be the set of all proportional allocations and let $\pi_i$ be the $i$-th agent in the order $\pi$.
Then, the mechanism outputs an allocation in $\propallocations$ in the following ``leximax'' way:
\begin{enumerate}
    \item the allocation maximizes agent $\pi_1$'s utility;
    \item subject to (1), the allocation maximizes agent $\pi_2$'s utility;
    \item subject to (1) and (2), the allocation maximizes agent $\pi_3$'s utility;
    \item $\cdots$
\end{enumerate}

By definition, the mechanism always outputs a proportional allocation. It is straightforward to check that it outputs a Pareto-efficient allocation. We further prove that:

\begin{theoremrep}\label{thm:volatile-cake-degree-n-1}
    The RAT-degree of the Volatile Priority Cake Cutting is $n-1$.
\end{theoremrep}

The proof is similar to the one for indivisible goods (\Cref{thm:volatile-good-alloc-degree-n-1}) and is provided in the appendix.
\begin{toappendix}
    We need the following proposition; it follows from known results on super-proportional cake-cutting \citep{dubins1961cut,woodall1986note}; for completeness we provide a proof in the appendix.

\begin{propositionrep} \label{prop:strictlymorethanproportional}
    Let $\propallocations$ be the set of all proportional allocations for the valuation profile $(v_1,\ldots,v_n)$.
    Let $(A_1,\ldots,A_n)$ be the allocation in $\propallocations$ that maximizes agent $a_i$'s utility.
    If there exists $j\in[n]\setminus\{i\}$ such that $v_i$ and $v_j$ are not identical up to scaling, then $V_i(A_i)>\frac1nV_i([0,1])$.
\end{propositionrep}
\begin{proof}
    We will explicitly construct a proportional allocation $(B_1,\ldots,B_n)$ where $V_i(B_i)>\frac1nV_i([0,1])$ if the pre-condition in the statement is satisfied.
    Notice that this will imply the proposition, as we are finding the allocation maximizing $a_i$'s utility.
    To construct such an allocation, we assume $v_i$ and $v_j$ are normalized without loss of generality (then $v_i\neq v_j$), and consider the equal division allocation where each uniform segment $X_t$ is evenly divided.
    This already guarantees that agent $a_i$ receives a value of $\frac1nV_i([0,1])$.
    Since $v_i$ and $v_j$ are normalized and $v_i\neq v_j$, there exist two uniform segments $X_{t_1}$ and $X_{t_2}$ such that $v_i(X_{t_1})>v_j(X_{t_1})$ and $v_i(X_{t_2})<v_j(X_{t_2})$.
    Agent $a_i$ and $a_j$ can then exchange parts of their allocations on $X_{t_1}$ and $X_{t_2}$ to improve the utility for both of them, which guarantees the resultant allocation is still proportional.
    For example, set $\varepsilon>0$ be a very small number.
    Agent $a_i$ can give a length of $\frac{\varepsilon}{v_i(X_{t_2})+v_j(X_{t_2})}$ from $X_{t_2}$ to agent $a_j$, in exchange of a length of $\frac{\varepsilon}{v_i(X_{t_1})+v_j(X_{t_1})}$ from $X_{t_1}$.
    This describes the allocation $(B_1,\ldots,B_n)$.
\end{proof}

We can now prove \Cref{thm:volatile-cake-degree-n-1}.
\begin{proof}[Proof of \Cref{thm:volatile-cake-degree-n-1}]
Consider an arbitrary agent $a_i$ with the true value density function $v_i$, and an arbitrary agent $a_j$ whose reported value density function is unknown to $a_i$.
Fix $n-2$ arbitrary value density function $\{v_k\}_{k\notin\{i,j\}}$ for the remaining $n-2$ agents.
Consider an arbitrary manipulation $v_i'\neq v_i$.

Choose a uniform segment $X_t$ with respect to $(v_1,\ldots,v_{j-1},v_{j+1},\ldots,v_n)$,
satisfying $v_i(X_t)>0$.
Choose a very small interval $E\subseteq X_t$, such that the value density function
$$\bar{v}_j=\left\{\begin{array}{ll}
    0 & \mbox{if }x\in E \\
    v_i(x) & \mbox{otherwise}
\end{array}\right.$$
is not a scaled version of some $v_k$ with $k\in[n]\setminus\{i,j\}$.
Apply the volatility of $\Gamma$ to find a value density function $v_j$ for agent $a_j$ that rescales $\bar{v}_j$ such that
\begin{enumerate}
    \item when agent $a_i$ reports $v_i$, agent $a_i$ is the first in the order output by $\Gamma$;
    \item when agent $a_i$ reports $v_i'$, agent $a_j$ is the first in the order output by $\Gamma$.
\end{enumerate}

Let $(A_1,\ldots,A_n)$ and $(A_1',\ldots,A_n')$ be the output allocation for the profiles $\{v_k\}_{k\notin\{i,j\}}\cup\{v_i\}\cup\{v_j\}$ and $\{v_k\}_{k\notin\{i,j\}}\cup\{v_i'\}\cup\{v_j\}$ respectively.
Since $\bar{v}_j$ is not a scaled version of some $v_k$, its rescaled version $v_j$ is also different.
By Proposition~\ref{prop:strictlymorethanproportional}, $V_j(A_j')>\frac1nV_j([0,1])$,
as $a_j$ is the highest-priority agent when $a_i$ reports $v'_i$.
Let $D$ be some subset of $A_j'$ with $V_j(D)>0$ and $V_j(A_j'\setminus D)\geq\frac1nV_j([0,1])$, and consider the allocation $(A_1^+,\ldots,A_n^+)$ in which $D$ is moved from $a_j$ to $a_i$, that is,
\begin{itemize}
    \item for $k\notin\{i,j\}$, $A_k^+=A_k'$;
    \item $A_i^+=A_i'\cup D$;
    \item $A_j^+=A_j'\setminus D$.
\end{itemize}
It is clear by our construction that the new allocation is still proportional with respect to $\{v_k\}_{k\notin\{i,j\}}\cup\{v_i'\}\cup\{v_j\}$.
In addition, by the relation between $\bar{v}_j$ and $v_i$ (and thus the relation between $v_j$ and $v_i$), we have $V_i(D)>0$ based on agent $a_i$'s true value density function $v_i$.
Therefore, under agent $a_i$'s true valuation, $V_i(A_i^+)>V_i(A_i')$.

If the allocation $(A_1^+,\ldots,A_n^+)$ is not proportional under the profile $\{v_k\}_{k\notin\{i,j\}}\cup\{v_i\}\cup\{v_j\}$ (where $v_i'$ is changed to $v_i$), then the only agent for whom proportionality is violated must be agent $i$, that is,$V_i(A_i^+)<\frac1nV_i([0,1])$.
It then implies $V_i(A_i')<\frac1nV_i([0,1])$.
On the other hand, agent $a_i$ receives at least her proportional share when reporting truthfully her value density function $v_i$.
This already implies the manipulation is not safe.

If the allocation $(A_1^+,\ldots,A_n^+)$ is proportional under the profile $\{v_k\}_{k\notin\{i,j\}}\cup\{v_i\}\cup\{v_j\}$, then it is in $\propallocations$.
Since agent $a_i$ is the first agent in the order when reporting $v_i$ truthfully, we have $V_i(A_i)\geq V_i(A_i^+)$, which further implies $V_i(A_i)>V_i(A_i')$.
Again, the manipulation is not safe.
\end{proof}

\end{toappendix}
Finally, we analyze the run-time of our mechanism.

\begin{theoremrep}
The Volatile Priority Cake Cutting mechanism can be computed in polynomial time.
\end{theoremrep}

\begin{proofsketch}
    Once $\pi$ is determined, the allocation is computed by solving a sequence of linear programs. Let $x_{it}$ be the length of the $t$-th uniform segment allocated to agent $a_i$.
    Then an agent $a_i$'s utility is a linear expression $\sum_{t=1}^mv_{i}(X_t)x_{it}$, and requiring an agent's utility is at least some value (e.g., her proportional share) is a linear constraint.
\end{proofsketch}

\begin{proof}
We first note that $\Gamma$ can be computed in polynomial time.
Finding $v^\ast$ and reading the information of $i,t,a$, and $b$ can be performed in linear time, as it mostly only requires reading the input of the instance.
In particular, the lengths of $a$ and $b$ are both less than the input length, so $as+b$ is of at most linear length and can also be computed in linear time.
Finally, the length of $n!$ is $\Theta(n\log n)$, so $as+b \bmod (n!)$ can be computed in polynomial time.
We conclude that $\Gamma$ can be computed in polynomial time.

We next show that $\Psi$ can be computed by solving linear programs.
Let $x_{it}$ be the length of the $t$-th uniform segment allocated to agent $a_i$.
Then an agent $a_i$'s utility is a linear expression $\sum_{t=1}^mv_{i}(X_t)x_{it}$, and requiring an agent's utility is at least some value (e.g., her proportional share) is a linear constraint.
We can use a linear program to find the maximum possible utility $u_{\pi_1}^\ast$ for agent $\pi_1$ among all proportional allocations.
In the second iteration, we write the constraint $\sum_{t=1}^mv_{\pi_1}(X_t)x_{it}\geq u_{\pi_1}^\ast$ for agent $\pi_1$, the proportionality constraints for the $n-2$ agents $[n]\setminus\{\pi_1,\pi_2\}$, and maximize agent $\pi_2$'s utility.
This can be done by another linear program and gives us the maximum possible utility $u_{\pi_2}^\ast$ for agent $\pi_2$.
We can iteratively do this to figure out all of $u_{\pi_1}^\ast,u_{\pi_2}^\ast,\ldots,u_{\pi_n}^\ast$ by linear programs. 
\end{proof}
\section{Single-Winner Ranked Voting}\label{sec:single-winner-voting}
We consider $n$ voters (the agents) who need to elect one winner from a set $C$ of $m$ \emph{candidates}.
The agents' preferences are given by strict linear orderings $\succ_i$ over the candidates.

When there are only two candidates, the majority rules and its variants (weighted majority rules)  are truthful.
With three or more candidates, the   Gibbard--Satterthwaite Theorem 
\cite{gibbard1973manipulation,satterthwaite1975strategy}
implies that the only truthful rules are dictatorships. 
Our goal is to find non-dictatorial rules with a high RAT-degree.

\newcommand{\scorevector}{\mathbf{s}}
\newcommand{\score}{\operatorname{score}}

We focus on \emph{positional voting rules}, which parameterized by a vector of scores, $\scorevector=(s_1,\ldots,s_m)$, where $s_1\leq \cdots \leq  s_m$ and $s_1 < s_m$.
Each voter reports his entire ranking of the $m$ candidates. Each such ranking is translated to an assignment of a score to each candidate: the lowest-ranked candidate is given a score of $s_1$, the second-lowest candidate is given $s_2$, etc., and the highest-ranked candidate is given a score of $s_m$. 
The total score of each candidate is the sum of scores he received from the rankings of all $n$ voters. The winner is the candidate with the highest total score. 

Formally, for any subset $N'\subseteq N$ and any candidate $c\in C$, we denote by $\score_{N'}(c)$ the total score that $c$ receives from the votes of the agents in $N'$. Then the winner is $\arg\max_{c\in C}\score_N(c)$. If there are several agents with the same maximum score, then the outcome is considered a tie.

Common special cases of positional voting are \emph{plurality voting}, in which $\scorevector = (0,0,0,\ldots,0,1)$, and 
\emph{anti-plurality voting}, in which $\scorevector = (0,1,1,\ldots,1,1)$.
By the Gibbard--Satterthwaite theorem, all positional voting rules are manipulable, so their RAT-degree is smaller than $n$.
But, as we will show next, some positional rules have a higher RAT-degree than others.

\paragraph{Results.} 
\rmark{\citet{conitzer2011dominating} showed that all positional voting rules are immune to dominating manipulations when $n\geq 6(m-2)$ --- in our terms, that they have positive RAT-degree. We extend their results for the case where $n\geq 2m$ by proving} that all positional voting rules have an RAT-degree between $\approx n/m$ and $\approx n/2$. These bounds are almost tight: the 
upper bound is attained by plurality and the lower bound is attained by anti-plurality (up to small additive constants).

These results raise the question of whether some other, non-positional voting rules have RAT-degrees substantially higher than $n/2$.
Using our volatile priority approach (see \Cref{sec:indivisible-good-aloc}), 
we could choose a ``dictator'', and take her first choice.
This deterministic mechanism has RAT-degree $n-1$, any manipulation risks losing the chance to be the dictator.
However, besides the fact that this is an unnatural mechanism, it suffers from other problems such as the \emph{no-show paradox} (a participating voter might affect the selection rule in a way that will make another agent a dictator, which might be worse than not participating at all).

Our main open problem is therefore to devise natural voting rules with a high RAT-degree.

\begin{open}\label{open:voting-1}
Does there exist a non-dictatorial voting rule that satisfies the participation criterion (i.e. does not suffer from the no-show paradox),  with RAT-degree larger than $\ceil{n/2}+1$? 
\end{open}

\rmark{\citet{conitzer2011dominating} further proved, in our terms, that any Condorcet-consistent voting rule has a positive RAT-degree. This raises the following question

\begin{open}\label{open:voting-2}
    What are the exact RAT-degrees of known Condorcet-consistent rules?
\end{open}

\citet{Moulin1982vetoPower} suggested the notion of \emph{veto power} for voting. We provide further details on the relationship between the two notions in \Cref{sec:voting-RAT-vs-veto-power}.

}

\begin{table}[h]
    \centering
    \begin{tabular}{|ll|c|c|}
    \hline 
    \multicolumn{2}{|c|}{\multirow{2}{*}{Mechanism}} & \multicolumn{2}{c|}{RAT-Degree} \\
    \cline{3-4}
    \multicolumn{2}{|c|}{}& Lower Bound & Upper Bound\\
    \hhline{|==|=|=|}
    Positional Voting Rules & (assuming $n\geq 2m$) & $\floor{(n+1)/m}-1$ & $ \ceil{n/2} +1$\\
    \hline 
    Plurality &(assuming $n \geq 5$) & $\floor{n/2} +1$ & $ \ceil{n/2} +1$ \\
    \hline 
    Anti-Plurality &(assuming $n \geq m^2$)& $\floor{(n+1)/m}-1$ & $\floor{n/m} +1$ \\
    \hline 
    \end{tabular}
    \caption{\centering Single-Winner Ranked Voting with $m \geq 3$: Summary of Results. }
    \label{tab:single-winner-sum}
\end{table}


Throughout the analysis, we consider a specific agent Alice, who looks for a safe profitable manipulation. Her true ranking is $c_m \succ_A \cdots \succ_A c_1$.
We assume that, for any $j>i$, Alice strictly prefers a victory of $c_j$ to a tie between $c_j$ and $c_i$, and strictly prefers this tie to a victory of $c_i$.%
\footnote{We could also assume that ties are broken at random, but this would require us to define preferences on lotteries, which we prefer to avoid in this paper.}

\subsection{Positional Voting Rules: General Bounds}\label{sec:pos-voting-general-bounds}

\begin{toappendix}
    \subsection{Positional Voting Rules: General Bounds}
\end{toappendix}

In the upcoming lemmas, we identify the manipulations that are safe and profitable for Alice under various conditions on the score vector $\scorevector$. We assume throughout that there are $m\geq 3$ candidates, and that $n$ is sufficiently large.
We allow an agent to abstain, which means that his vote gives the same score to all candidates.%
\footnote{
We need the option to abstain in order to avoid having different constructions for even $n$ and odd $n$; see the proofs for details.
}

The first lemma implies a lower bound of $\approx n/m$ on the RAT-degree of positional voting rules.
\begin{lemmarep}
\label{lem:lower-positional}
In any positional voting rule for $m\geq 3$ candidates,
if the number of known agents is at most $(n+1)/m - 2$,
then Alice has no safe profitable manipulation.
\end{lemmarep}

\begin{proofsketch}
For any combination of rankings of the known agents, it is possible that the unknown agents vote in a way that balances out the votes of the known agents, such that all agents have almost the same score; if Alice is truthful, there is a tie between two candidates, and if she manipulates, the worse of these candidates win.
\end{proofsketch}

\begin{proof}
If a manipulation does not change any score, then it is clearly not profitable. So suppose the manipulation changes the score of some candidates. Let $c_i$ be a candidate whose score increases by the largest amount by the manipulation (note that it cannot be Alice's top candidate). If there are several such candidates, choose the one that is ranked highest by Alice. Let $c_j$ be some candidate that Alice prefers to $c_i$ \rmark{ whose score does not increase by the same amount.}

To show that the manipulation is not safe, suppose the unknown agents vote as follows.
\begin{itemize}
\item For every known agent $a_i$, some $m-1$ unknown agents vote with ``rotated'' variants of $a_i$'s ranking (e.g. if $a_i$ ranks $c_1\succ c_2\succ c_3 \succ c_4$, then three unknown agents rank $c_2\succ c_3 \succ c_4 \succ c_1$,
$c_3 \succ c_4 \succ c_1\succ c_2$ 
and 
$c_4 \succ c_1\succ c_2\succ c_3$).
\item Additional $m-1$ unknown agents vote with rotated variants of Alice's true ranking.

\item Additional $\ceil{m/2-1}$ unknown agents rank 
rank $c_i$ first and $c_j$ second, 
and additional $\ceil{m/2-1}$ unknown agents rank $c_j$ first and $c_i$ second. These $2\ceil{m/2-1}$ unknown agents rank the remaining $m-2$ candidates such that each candidate appears last at least once (note that $2\ceil{m/2-1}\geq m-2$ so this is possible).
\item The other unknown agents, if any, abstain.
\end{itemize}
When Alice is truthful, the scores are as follows:
\begin{align*}
\score_N(c_i) = \score_N(c_j) &= 
(k+1)\cdot \sum_{j=1}^m s_j + 
(\ceil{m/2-1})(s_m+s_{m-1})
\\
\score(c_{\ell}) &= 
(k+1)\cdot \sum_{j=1}^n s_j + 
S_{\ell}
&& \forall \ell\neq i,j
\end{align*}
where $S_{\ell}$ is the sum of some $2\ceil{m/2-1}$ scores; all these scores are at most $s_{m-1}$, and some of them are equal to $s_1$. As $s_1<s_m$ for any score vector, 
$S_{\ell} < (\ceil{m/2-1})(s_m+s_{m-1})$ for all $\ell\neq i,j$. Hence, the two candidates $c_i$ and $c_j$ both have a score strictly higher than every other candidate.

When Alice manipulates, the score of $c_i$ increases by the largest amount, so $c_i$ wins. Since $c_j \succ_A c_i$, this outcome is worse for Alice than the tie.

The number of agents required is 
\begin{align*}
k + (m-1)k + (m-1) + 1 + 2\ceil{m/2-1}
\leq 
m (k+2) - 1.
\end{align*}
The condition on $k$ in the lemma ensures that $n$ is at least as large.
\end{proof}

We now prove an upper bound of $\approx n/2$ on the RAT-degree. We need several lemmas.

\begin{lemmarep}
\label{lem:s2>s1}
Let $m\geq 3$ and $n\geq 2m$.
If $s_2 > s_1$
and there are $k\geq \ceil{n/2}+1$ known agents,
then
switching the bottom two candidates ($c_2$ and $c_1$) may be a safe profitable manipulation for Alice.
\end{lemmarep}
\begin{proofsketch}
For some votes by the known agents,
$c_1$ has no chance to win, so the worst candidate for Alice that could win is $c_2$. Therefore, switching $c_1$ and $c_2$ cannot harm, but may help a better candidate win over $c_2$.
\end{proofsketch}


\begin{proof}
Suppose there is a subset $K$ of $\ceil{n/2}+1$ known agents, who vote as follows:
\begin{itemize}
\item $\floor{n/2}-1$ known agents rank $c_2 \succ c_m \succ \cdots \succ c_1$.
\item Two known agents rank $c_m \succ c_2 \succ
\cdots \succ c_1$. 
\item In case $n$ is odd, the remaining known agent abstains.
\end{itemize}
We first show that $c_1$ cannot win. To this end, we show that the difference in scores between $c_2$ and $c_1$ is always strictly positive.
\begin{itemize}
\item The difference in scores given by the known agents is 
\begin{align*}
\score_K(c_2)-\score_K(c_1) =
&
(\floor{n/2}-1)(s_m-s_1) 
+ 2(s_{m-1}-s_1).
\end{align*}
\item There are
$\floor{n/2}-1$ agents not in $K$ (including Alice).
These agents can reduce the score-difference by at most 
$(\floor{n/2}-1)(s_m-s_1)$.
Therefore, 
\begin{align*}
\score_N(c_2)-\score_N(c_1) \geq 2(s_{m-1}-s_1),
\end{align*}
which is positive 
by the assumption $s_2>s_1$.
So $c_1$ has no chance to win or even tie.
\end{itemize}
Therefore, switching $c_2$ and $c_1$ can never harm Alice --- the manipulation is safe.

Next, we show that the manipulation can help $c_m$ win, when the agents not in $K$ vote as follows:
\begin{itemize}
\item $\floor{n/2}-3$ unknown agents rank $c_m\succ c_2\succ \cdots $,
where each candidate except $c_1,c_2,c_m$ is ranked last by at least one voter (here we use the assumption $n\geq 2m$).
\item One unknown agent ranks 
$c_2\succ \cdots \succ c_m \succ c_1$;
\item Alice votes truthfully $c_m\succ  \cdots \succ c_2 \succ c_1$.  
\end{itemize}
Then,
\begin{align*}
\score_N(c_2) - \score_N(c_m)
=&
(\floor{n/2}-1)(s_{m}-s_{m-1}) 
+ 2(s_{m-1}-s_{m})
\\
&+
(\floor{n/2}-3)(s_{m-1}-s_{m}) 
+ (s_m-s_2)
+ (s_2-s_m)
\\
=&
0.
\end{align*}
Moreover, for any $j\not\in\{1,2,m\}$, the score of $c_j$ is even lower (here we use the assumption that $c_j$ is ranked last by at least one unknown agent):
\begin{align*}
\score_N(c_2) - \score_N(c_j)
\geq &
(\floor{n/2}-1)(s_{m}-s_{m-2}) 
+ 2(s_{m-1}-s_{m-2})
\\
&+
(\floor{n/2}-4)(s_{m-1}-s_{m-2}) 
+ (s_{m-1}-s_1)
+ (s_m-s_{m-1})
+ (s_2-s_{m-1})
\\
\geq & (s_{m-1}-s_{1})
+ (s_2-s_{m-1})
\\
= & s_2 - s_1,
\end{align*}
which is positive by the lemma assumption.
Therefore, when Alice is truthful, the outcome is a tie between $c_m$ and $c_2$.

If Alice switches $c_1$ and $c_2$, then the score of $c_2$ decreases by $s_2-s_1$, which is positive by the lemma assumption, and the scores of all other candidates except $c_1$ do not change. So $c_m$ wins, which is better for Alice than a tie.
Therefore, the manipulation is profitable.
\end{proof}

\Cref{lem:s2>s1} can be generalized as follows. 

\begin{lemmarep}
\label{lem:st1>st}
Let $m\geq 3$ and $n\geq 2m$.
For every integer $t \in \{1,\ldots, m-2\}$,
if $s_{t+1} > s_t = \cdots = s_1$,
and there are $k\geq \ceil{n/2}+1$ known agents,
then switching $c_{t+1}$ and $c_t$ may be a safe profitable manipulation.
\end{lemmarep}

\begin{proofsketch}
For some votes by the known agents,
all candidates $c_1,\ldots,c_t$ have no chance to win, 
so the worst candidate for Alice that could win is $c_{t+1}$.
Therefore, switching $c_t$ and $c_{t+1}$ cannot harm, but can help better candidates win over $c_{t+1}$.
\end{proofsketch}
\begin{proof}
Suppose there is a subset $K$ of $\ceil{n/2}+1$ known agents, who vote as follows:
\begin{itemize}
\item $\floor{n/2}-1$ known agents rank $c_{t+1} \succ c_m$ first and rank $c_t \succ \cdots \succ c_1$ last.
\item Two known agents rank $c_m \succ c_{t+1}$ first and rank $c_t \succ \cdots \succ c_1$ last.
\item In case $n$ is odd, the remaining known agent abstains.
\end{itemize}
We first show that the $t$ worst candidates for Alice ($c_1,\ldots, c_t$) cannot win. 
Note that, by the lemma assumption $s_t = \cdots = s_1$, all these candidates receive exactly the same score by all known agents. We show that the difference in scores between $c_{t+1}$ and $c_t$ (and hence all $t$ worst candidates) is always strictly positive.
\begin{itemize}
\item The difference in scores given by the known agents is 
\begin{align*}
\score_K(c_{t+1})-\score_K(c_t) =
&
(\floor{n/2}-1)(s_m-s_1) 
+ (s_{m-1}-s_1).
\end{align*}
\item There are
$\floor{n/2}-1$ agents not in $K$ (including Alice).
These agents can reduce the score-difference by at most 
$(\floor{n/2}-1)(s_m-s_1)$.
Therefore, 
\begin{align*}
\score_N(c_{t+1})-\score_N(c_t) \geq (s_{m-1}-s_1),
\end{align*}
which is positive 
by the assumption $m-2 \geq t$ and $s_{t+1}>s_t$.
So no candidate in $c_1,\ldots,c_t$ has a chance to win or even tie.
\end{itemize}
Therefore, switching $c_{t+1}$ and $c_t$ can never harm Alice --- the manipulation is safe.

Next, we show that the manipulation can help $c_m$ win. We compute the score-difference between $c_m$ and the other candidates with and without the manipulation. 

Suppose that the agents not in $K$ vote as follows:
\begin{itemize}
\item $\floor{n/2}-3$ unknown agents rank $c_m\succ c_{t+1}\succ \cdots $,
where each candidate in $c_{t+2},\ldots,c_{m-1}$ is ranked last by at least one voter (here we use the assumption $n\geq 2m$).
\item One unknown agent ranks 
$c_{t+1}\succ \cdots \succ c_m \succ c_t \succ \cdots \succ c_1$;
\item Alice votes truthfully $c_m\succ  \cdots \succ c_{t+1} \succ  c_t \succ \cdots \succ c_1$.  
\end{itemize}
Then,
\begin{align*}
\score_N(c_{t+1}) - \score_N(c_m)
=&
(\floor{n/2}-1)(s_{m}-s_{m-1}) 
+ 2(s_{m-1}-s_{m})
\\
&+
(\floor{n/2}-3)(s_{m-1}-s_{m}) 
+ (s_m-s_{t+1})
+ (s_{t+1}-s_m)
\\
=&
0.
\end{align*}
Moreover, for any $j\in\{t+2,\ldots,m-1\}$, the score of $c_j$ is even lower (here we use the assumption that $c_j$ is ranked last by at least one unknown agent):
\begin{align*}
\score_N(c_{t+1}) - \score_N(c_j)
\geq &
(\floor{n/2}-1)(s_{m}-s_{m-2}) 
+ 2(s_{m-1}-s_{m-2})
\\
&+
(\floor{n/2}-4)(s_{m-1}-s_{m-2}) 
+ (s_{m-1}-s_1)
+ (s_m-s_{m-1})
+ (s_{t+1}-s_{m-1})
\\
\geq & (s_{m-1}-s_{1})
+ (s_{t+1}-s_{m-1})
\\
= & s_{t+1} - s_1,
\end{align*}
which is positive by the lemma assumption.
Therefore, when Alice is truthful, the outcome is a tie between $c_m$ and $c_{t+1}$.

If Alice switches $c_t$ and $c_{t+1}$, then the score of $c_{t+1}$ decreases by $s_{t+1}-s_t$, which is positive by the lemma assumption, and the scores of all other candidates except $c_t$ do not change. As $c_t$ cannot win, $c_m$ wins, which is better for Alice than a tie.
Therefore, the manipulation is profitable.
\end{proof}

\begin{lemmarep}
\label{lem:sm>sm1}
Let $m\geq 3$ and $n\geq 4$.
If $s_m > s_{m-1}$
and there are $k\geq \ceil{n/2}+1$ known agents,
then
switching the top two candidates ($c_m$ and $c_{m-1}$) may be a safe profitable manipulation for Alice.
\end{lemmarep}
\begin{proofsketch}
For some votes by the known agents, the manipulation is safe since $c_m$ has no chance to win, and it is profitable as it may help $c_{m-1}$ win over worse candidates.
\end{proofsketch}
\begin{proof}
Suppose there is a subset $K$ of $\ceil{n/2}+1$ known agents, who vote as follows:
\begin{itemize}
\item $\floor{n/2}-1$ known agents rank $c_{m-2} \succ c_{m-1} \succ \cdots \succ c_m$.
\item One known agent ranks $c_{m-2} \succ c_m \succ
c_{m-1} \succ \cdots $. 
\item One known agent ranks $c_{m-1} \succ c_{m-2} \succ \cdots \succ c_m$. 
\item In case $n$ is odd, the remaining known agent abstains.
\end{itemize}
We first show that $c_m$ cannot win. To this end, we show that the difference in scores between $c_{m-2}$ and $c_m$ is always strictly positive.
\begin{itemize}
\item The difference in scores given by the known agents is 
\begin{align*}
\score_K(c_{m-2})-\score_K(c_m) =
&
(\floor{n/2}-1)(s_m-s_1) 
+ (s_m-s_{m-1})
+ (s_{m-1}-s_1).
\\
=&
(\floor{n/2})(s_m-s_1) 
\end{align*}
\item There are
$\floor{n/2}-1$ agents not in $K$ (including Alice).
These agents can reduce the score-difference by at most 
$(\floor{n/2}-1)(s_m-s_1)$.
Therefore, 
\begin{align*}
\score_N(c_{m-2})-\score_N(c_m) \geq (s_m-s_1),
\end{align*}
which is positive for any score vector.
So $c_m$ has no chance to win or even tie.
\end{itemize}
Therefore, switching $c_m$ and $c_{m-1}$ can never harm Alice --- the manipulation is safe.

Next, we show that the manipulation can help $c_{m-1}$ win. We compute the score-difference between $c_{m-1}$ and the other candidates with and without the manipulation. 

Suppose that the agents not in $K$ vote as follows:
\begin{itemize}
\item the $\floor{n/2}-2$ unknown agents%
\footnote{Here we use the assumption $n\geq 4$.}
rank $c_{m-1}\succ c_{m-2}\succ \cdots $.
\item Alice votes truthfully $c_m\succ c_{m-1}\succ c_{m-2} \cdots \succ c_1$.
\end{itemize}
Then,
\begin{align*}
\score_N(c_{m-2}) - \score_N(c_{m-1})
=&
(\floor{n/2}-1)(s_{m}-s_{m-1}) 
+ (s_m-s_{m-2})
+ (s_{m-1}-s_{m})
\\
&+
(\floor{n/2}-2)(s_{m-1}-s_{m}) 
+ (s_{m-2}-s_{m-1})
\\
=&
(s_{m}-s_{m-1}),
\end{align*}
which is positive by the assumption $s_m>s_{m-1}$.
The candidates $c_{j<m-2}$ are ranked even lower than $c_{m-1}$ by all agents. Therefore the winner is $c_{m-2}$.

If Alice switches $c_{m-1}$ and $c_m$, then the score of $c_{m-1}$ increases by $s_m-s_{m-1}$ and the scores of all other candidates except $c_m$ do not change. Therefore, 
$\score_N(c_{m-2}) - \score_N(c_{m-1})$ becomes $0$, and there is a tie between $c_{m-2}$ and $c_{m-1}$, which is better for Alice by assumption.
Therefore, the manipulation is profitable.

\end{proof}

Combining the lemmas leads to the following bounds on the RAT-degree:
\begin{theorem}
\label{thm:positional-bounds}
For any positional voting rule with $m\geq 3$ candidates:

(a) The RAT-degree is at least $\floor{(n+1)/m}-1$;

(b) When $n\geq 2m$, the RAT-degree is at most $\ceil{n/2}+1$.
\end{theorem}
\begin{proof}
The lower bound follows immediately from \Cref{lem:lower-positional}.

For the upper bound, consider a positional voting rule with score-vector $\scorevector$. Let $t \in \{1,\ldots,m-1\}$ be the smallest index for which $s_{t+1} > s_t$ (there must be such an index by definition of a score-vector).

If $t\leq m+2$, then \Cref{lem:st1>st} implies that, for some votes by the $\ceil{n/2}+1$ known agents, switching $c_{t+1}$ and $c_t$ may be a safe and profitable manipulation for Alice.

Otherwise, $t=m-1$, and \Cref{lem:sm>sm1} implies the same.

In all cases, Alice has a safe profitable manipulation.
\end{proof}

\subsection{Plurality and Anti-plurality}
\begin{toappendix}
    \subsection{Plurality and Anti-plurality}
\end{toappendix}

We now show that the bounds of \Cref{thm:positional-bounds} are tight up to small additive constants.

We first show that the upper bound of $\approx n/2$ is attained by the \emph{plurality voting rule}, which is the positional voting rule with score-vector $(0,0,0,\ldots,0,1)$.

\begin{lemmarep}
\label{lem:lower-plurality}
\er{In the plurality voting rule with $n\geq 5$ agents,
if the number of known agents is at most $n/2$, then Alice has no safe profitable manipulation.
}
\end{lemmarep}
\begin{proofsketch}
When there are at most $n/2$ known agents, there are at least $n/2-1$ unknown agents. For some votes of these unknown agents, the outcome when Alice votes truthfully is a tie between Alice's top candidate and another candidate. But when Alice manipulates, the other candidate wins.
\end{proofsketch}

\begin{proof}

\newcommand{\aFav}{c_m}
\newcommand{\aAlt}{c^A_{alt}}
\newcommand{\kAlt}{c^K_{alt}}

If a manipulation does not involve Alice's top candidate $c_m$, then it does not affect the outcome and cannot be profitable. So let us consider a manipulation in which Alice ranks another candidate $\aAlt\neq c_m$ at the top. We show that the manipulation is not safe.

Note that there are $n-k-1$ unknown agents; the lemma condition implies $n-k-1\geq n-n/2-1 = n/2-1 \geq k-1$.

Let $\displaystyle \kAlt = \argmax_{j \in [m-1]} \score_K(c_j)$ denote the candidate with the highest number of votes among the known agents, except Alice's top candidate ($c_m$).
Consider the following two cases.

\paragraph{\underline{Case 1:} $\score_K(\kAlt) = 0$.} Since $\kAlt$ is a candidate who got the maximum number of votes from $K$ except $\aFav$, this implies that all $k$ known agents either vote for $\aFav$ or abstain.

Suppose that some $k-1$ unknown agents vote for $\aAlt$ or abstain, such that the score-difference $\score(c_m)-\score(\aAlt) = 1$ (if there are additional agents, they abstain).
Then, when Alice is truthful, her favorite candidate, $\aFav$ wins, as $\score_N(c_m) -\score_N(\aAlt) = 2$ and the scores of all other candidates are $0$. 
But when Alice manipulates and votes for $\aAlt$, the outcome is a tie between $\aFav$ and $\aAlt$, which is worse for Alice.

\paragraph{\underline{Case 2:} $\score_K(\kAlt) \geq 1$.}
Then again the manipulations not safe, as it is possible that the unknown agents vote as follows: 
\begin{itemize}
\item Some $\score_K(\aFav)$ agents vote for $\kAlt$;
\item Some 
$\score_K(\kAlt) -1$ agents vote for $c_m$. 

This is possible as both values are non-negative and $\score_K(\aFav) + \score_K(\kAlt) \leq \sum_{j =1}^m \score_K(c_j) \leq  k$, so $\score_K(\aFav) + \left(\score_K(\kAlt)-1\right) \leq k-1\leq $ the number of unknown agents.
\item 
The remaining unknown agents (if any) are split evenly between $c_m$ and $\kAlt$; if the number of remaining unknown agents is odd, then the extra agent abstains. 
\end{itemize}
We now prove that the manipulation is harmful for Alice.

Denote $N' := N\setminus \{Alice\} = $ all agents except Alice. Then
\begin{align*}
\score_{N'}(\kAlt) &= \score_K(\aFav) + \score_K(\kAlt);
\\
\score_{N'}(\aFav) &= \score_K(\aFav) + \score_K(\kAlt) - 1.
\end{align*}
so the score-difference is exactly $1$.

Also, as $\kAlt$ has the largest score among the known agents, this still holds with the unknown agents, as all of them vote for either $\kAlt$ or $\aFav$.

We claim that $\score_{N'}(\kAlt)$ is strictly higher than that of all other candidates. Indeed:
\begin{itemize}
\item If $|R|\geq 2$, then $\kAlt$ receives at least one vote from an unknown agent, whereas all other candidates except $\aFav$ receive none.

\item Otherwise, $|R|\leq 1$, which means that $\score_K(\aFav) + \score_K(\kAlt) - 1 \geq n-k-2$, so
$\score_{N'}(\kAlt) \geq n-k-1\geq n/2-1$, which is larger than $1$ since $n\geq 5$. 
On the other hand, 
$\score_K(\aFav) + \score_K(\kAlt) - 1 \geq k-1$, which  implies that all other candidates together received at most one vote from all known agents.
\end{itemize}
Now, if Alice is truthful, the outcome is a tie between $\kAlt$ and $\aFav$, but when she manipulates and removes her vote from $c_m$, the outcome is a victory for $\kAlt$, which is worse for her.

Thus, in all cases, Alice does not have a safe profitable manipulation.
\end{proof}

Combining \Cref{lem:sm>sm1} and 
\Cref{lem:lower-plurality} gives an almost exact RAT-degree of plurality voting.
\begin{theorem}
With $m\geq 3$ candidates and $n\geq 5$ agents, the RAT-degree of plurality voting is $n/2+1$ when $n$ is even;
it is between $\floor{n/2}+1$ and  $\ceil{n/2}+1$ when $n$ is odd.
\end{theorem}

Next, we show that the lower bound of $\approx n/m$ is attained by the \emph{anti-plurality voting rule}. 

\begin{lemmarep}
    With $m \geq 3$ and $n\geq m^2$, the RAT-degree of anti-plurality is at most $\lfloor n/m \rfloor +1$.
\end{lemmarep}

\begin{proofsketch}
    We show that switching the two least-favorite candidates ($c_1$ and $c_2$) constitutes a safe and profitable manipulation for Alice, as the number of known agents is sufficient to ensure that her least-preferred candidate cannot win.
\end{proofsketch}

\begin{proof}
    We shall prove that anti-plurality is $(\lfloor n/m \rfloor +1)$-known agents manipulable.

    Consider the following case. Alice's preference order is $c_m \succ_A c_{m-1} \succ_A \cdots \succ_A c_1$, and she knows that the $\lfloor n/m \rfloor -1$ known-agents veto Alice's least-preferred candidate, $c_1$.

    We prove that switching the two least-favorite candidates ($c_1$ and $c_2$) is a safe and profitable manipulation for Alice.

    We first show that $c_1$ has no chance of winning, which makes the manipulation safe. Consider the highest possible score of $c_1$ when Alice manipulates (among all possible profiles for the unknown agents):
\begin{align*}
    \score_N(c_1) &\leq \underbrace{k \cdot 0}_{\text{knowns}} 
    + \underbrace{1}_{\text{Alice}} 
    + \underbrace{(n-k-1) \cdot 1}_{\text{best unknowns}} 
    = n-k.
\end{align*}

Next, we prove that,for any possible profile for the unknown agents, the highest score among the other candidates is higher than the highest possible score of $c_1$.
Each voter gives a score of $1$ to $m-1$ candidates and $0$ to one candidate.
There are $m-1$ candidates other than $c_1$, and $n-k$ voters we need to consider --- the unknown agents and Alice (each known agent gives score of $1$ to any candidate other than $c_1$).
Each of these voters gives a score of $0$ to at most $1$ of these $m-1$ candidates (since each one of the unknown agent might also give its $0$ to $c_1$).
Hence, overall they contribute at most $n-k$ zeros in total, distributed among the $m-1$ candidates.
By the reverse pigeonhole principle, there exists at least one candidate $c_j \neq c_1$ that receives at most $\lfloor \frac{n-k}{m-1} \rfloor$
zeros from these voters.
Thus, its overall score is at least $n- \lfloor \frac{n-k}{m-1} \rfloor $. 
We shall now see that as $k = \lfloor n/m \rfloor +1$, this score is higher than the best possible score of $c_1$:
\begin{align*}
    & k = \lceil  \frac{n}{m} \rfloor +1 \quad \Rightarrow k > \frac{n}{m} \quad \Rightarrow k \cdot m -k > n -k \quad \Rightarrow k > \frac{n-k}{m-1}  \quad \Rightarrow k > \lfloor \frac{n-k}{m-1} \rfloor\\
    &\Rightarrow n - k < n- \lfloor \frac{n-k}{m-1} \rfloor -1
\end{align*}

Hence, $c_1$ never wins, which shows that the manipulation is safe.

 Next, to prove that the manipulation is profitable, consider the following profile for the unknown agents:
    \begin{itemize}
        \item Each $\lfloor n/m \rfloor$ of them veto one candidate among $c_3, \dots, c_{m-1}$.
        \item Any remaining unknown agents (if $n$ is not divisible by $m$) can veto any candidate among $c_3, \dots, c_{m-1}$; this does not affect the argument.
    \end{itemize}
    That is, no unknown voter veto candidates $c_1, c_2$ and $c_m$, but some unknown voters veto the others.

    When Alice is truthful, the scores of $c_m$ and $c_2$ are $n$ as no voter veto them. The score of any other candidate is smaller --- which means that the outcome is a tie between $c_m$ and $c_2$. 

    However, when Alive manipulate, the score of $c_2$ decreases by one, which makes her favorite candidate $c_m$ the winner. By out assumption regarding tie breaking, as Alice prefers $c_m$ over $c_2$, this outcome is better for her than the tie.
\end{proof}

To this end, we also prove several upper bounds on the RAT-degree for more general score-vectors.

The following lemma strengthens 
\Cref{lem:s2>s1}.

\newcommand{\topLscores}{s_{\mathrm{top:}\ell}}
\newcommand{\botLscores}{s_{\mathrm{bot:}\ell}}
\begin{lemmarep}
\label{lem:z:s2>s1}
Let $\ell \in \{2,\ldots, m-1\}$ be an integer.
Consider a positional voting setting with $m\geq 3$ candidates and $n\geq (\ell+1)m$ agents.
Denote $\topLscores := \sum_{j=m-\ell+1}^m s_j = $  the sum of the $\ell$ highest scores and $\botLscores := \sum_{j=1}^{\ell}s_j = $ the sum of the $\ell$ lowest scores.

If $s_2 > s_1$ and there are $k$ known agents,
where 
\begin{align*}
k > \frac{\ell s_m - \botLscores}{\ell s_m + \topLscores - \botLscores - \ell s_1} n,
\end{align*}
then switching the bottom two candidates ($c_2$ and $c_1$) may be a safe profitable manipulation for Alice.
\end{lemmarep}
\begin{proofsketch}
The proof has a similar structure to that of \Cref{lem:s2>s1}.
Note that the expression at the right-hand side can be as small as $\displaystyle \frac{1}{\ell+1}n$ (for the anti-plurality rule), which is much smaller than the $\ceil{n/2}+1$ known agents required in \Cref{lem:s2>s1}.
Still, we can prove that, for some reports of the known agents, the score of $c_1$ is necessarily lower than the \emph{arithmetic mean} of the scores of the $\ell$ candidates $\{c_m, c_2, \cdots, c_{\ell}\}$. Hence, it is lower than at least one of these scores. Therefore ,$c_1$ still cannot win, so switching $c_1$ and $c_2$ is safe.
\end{proofsketch}
\begin{proof}
Suppose there is a subset $K$ of $k$ known agents, who vote as follows:
\begin{itemize}
\item $k-2$ known agents rank $c_2 \succ c_m$, then all candidates $\{c_3, \cdots , c_{\ell}\}$ in an arbitrary order, then the rest of the candidates in an arbitrary order, and lastly $c_1$.
\item Two known agents rank $c_m \succ c_2$, then all candidates $\{c_3 , \cdots , c_{\ell}\}$ in an arbitrary order, then the rest of the candidates in an arbitrary order, and lastly $c_1$.
\end{itemize}
We first show that $c_1$ cannot win. 
Denote $L := \{c_m, c_2, c_3, \ldots, c_{\ell}\}$.
We show that the difference in scores between some of the $\ell$ candidates in $L$ and $c_1$ is always strictly positive.
\begin{itemize}
\item The known agents rank all candidates in $L$ at the top $\ell$ positions. Therefore, each agent gives all of them together a total score of $\topLscores$. So
\begin{align*}
\sum_{c\in L} (\score_K(c)-\score_K(c_1)) =
&
k(\topLscores - \ell s_1).
\end{align*}
\item There are $n-k$ agents not in $K$ (including Alice). 
Each of these agents gives all candidates in $L$ together at least $\botLscores$, and gives $c_1$ at most $s_m$ points. Therefore, we can bound the sum of score differences as follows:
\begin{align*}
\sum_{c\in L} (\score_N(c)-\score_N(c_1)) \geq
&
k (\topLscores - \ell s_1)
+ (n-k) (\botLscores - \ell s_m)
\\
=&
k (\ell s_m + \topLscores - \botLscores - \ell s_1)
+ n(\botLscores - \ell s_m).
\end{align*}
The assumption on $k$ implies that this expression is positive. Therefore, for at least one $c\in L$, $\score_N(c)-\score_N(c_1) > 0$.
So $c_1$ has no chance to win or even tie.
Therefore, switching $c_2$ and $c_1$ is a safe manipulation.
\end{itemize}

Next, we show that the manipulation can help $c_m$ win, when the agents not in $K$ vote as follows:
\begin{itemize}
\item $k-4$ unknown agents rank $c_m\succ c_2\succ \cdots $,
where each candidate except $c_1,c_2,c_m$ is ranked last by at least one voter (here we use the assumption $n\geq (\ell+1)m$: the condition on $k$ implies $k>n/(\ell+1)\geq m$, so $k\geq m+1$ and $k-4\geq m-3$).
\item One unknown agent ranks 
$c_2\succ \cdots \succ c_m \succ c_1$;
\item Alice votes truthfully $c_m\succ  \cdots \succ c_2 \succ c_1$.  
\item If there are remaining unknown agents, then they are split evenly between 
$c_m\succ c_2\succ \cdots $ and 
$c_2\succ c_m\succ \cdots $ (if the number of remaining agents is odd, then the last one abstains).
\end{itemize}
Then,
\begin{align*}
\score_N(c_2) - \score_N(c_m)
=&
(k-2)(s_{m}-s_{m-1}) 
+ 2(s_{m-1}-s_{m})
\\
&+
(k-4)(s_{m-1}-s_{m}) 
+ (s_m-s_2)
+ (s_2-s_m)
\\
=&
0.
\end{align*}
Moreover, for any $j\not\in\{1,2,m\}$, the score of $c_j$ is even lower (here we use the assumption that $c_j$ is ranked last by at least one unknown agent):
\begin{align*}
\score_N(c_2) - \score_N(c_j)
\geq &
(k-2)(s_{m}-s_{m-2}) 
+ 2(s_{m-1}-s_{m-2})
\\
&+
(k-5)(s_{m-1}-s_{m-2}) 
+ (s_{m-1}-s_1)
+ (s_m-s_{m-1})
+ (s_2-s_{m-1})
\\
\geq & (s_{m-1}-s_{1})
+ (s_2-s_{m-1})
\\
= & s_2 - s_1,
\end{align*}
which is positive by the lemma assumption.
Therefore, when Alice is truthful, the outcome is a tie between $c_m$ and $c_2$.

If Alice switches $c_1$ and $c_2$, then the score of $c_2$ decreases by $s_2-s_1$, which is positive by the lemma assumption, and the scores of all other candidates except $c_1$ do not change. So $c_m$ wins, which is better for Alice than a tie.
Therefore, the manipulation is profitable.
\end{proof}

In particular, for the anti-plurality rule the condition in \Cref{lem:z:s2>s1} for $\ell=m-1$ is $k>n/m$, which implies a lower bound of $\floor{n/m}+1$. 
Combined with the general upper bound of \Cref{thm:positional-bounds}, we get:

\begin{theorem}
With $m\geq 3$ candidates and $n\geq m^2$ agents,
The RAT-degree of anti-plurality voting is 
at least $\floor{(n+1)/m}-1$ and
at most $\floor{n/m}+1$.
\end{theorem}
Intuitively, the reason that anti-plurality fares worse than plurality is that, even with a small number of known agents, it is possible to deduce that some candidate has no chance to win, and therefore there is a safe manipulation.

While we do not yet have a complete characterization of the RAT-degree of positional voting rules, our current results already show the strategic importance of the choice of scores.

\eden{\clearpage=============== NEW ===============}
\subsection{RAT-Degree and Veto Power}\label{sec:voting-RAT-vs-veto-power}

In this section, we examine the relation between the RAT-degree and the notion of \emph{veto power} introduced by \citet{Moulin1982vetoPower}.
Formally, a coalition of voters of size~$t$ is said to have veto power~$v(t)$ if it can prevent any set of~$v(t)$ candidates from winning, regardless of the preferences of the remaining voters.

Initially, we thought that the veto power can  provide an upper bound on the RAT-degree, as a manipulator who knows that some candidate has no chance of winning, can use this knowledge to obtain a safe and profitable manipulation.
Therefore, we expected that the minimal size of a coalition with positive veto power would be an upper bound on the RAT-degree.
However, this claim appears to fail in general: 

\begin{lemma}
    There exist voting rules for which, even when the manipulator knows that one candidate has no chance of winning, this is not sufficient to ensure the existence of a safe and profitable manipulation.
\end{lemma}

\begin{proof}
    Without loss of generality, assume $m$ is even and consider the $m/2$–approval rule.
    Under $m/2$–approval, each voter gives score $1$ to his top $m/2$ candidates and $0$ to the rest; thus any unilateral change by $i$ can change every candidate's total score by at most $\pm 1$. 

    Let Alice be the manipulator with true preferences $c_m \succ m_{m-1} \succ \cdots \succ c_1$.
    We show that if Alice only knows that one candidate has no chance of winning, any profitable manipulation is not safe. 
    Consider the following cases:
    
    \begin{itemize}
        \item The eliminated candidate is one of Alice's top $m/2$ candidates:

        In this case, the  manipulations that are potentially profitable are to move this eliminated candidate from the approved set (score~$1$) to the unapproved set (score~$0$), and to promote instead one of the unapproved candidates to the approved set. 
        Any other manipulation contributes exactly the same total scores as truthful reporting and therefore cannot be profitable. 
        However, this swap may cause a worse outcome for Alice: it is possible that when she tell the truth the result is a tie between the her favorite candidate and the less-preferred candidate that she promoted and so when she manipulates the less-preferred candidate wins.
        Therefore, such a manipulation is not safe.

        \item The eliminated candidate is among Alice’s bottom $m/2$ candidates:

        Here, the only potentially profitable manipulation is to promote this eliminated candidate into the approved set (score~$1$) and to demote the least-preferred among the currently approved candidates. 
        Any other manipulation yields the same contribution as truth-telling. 
        While this move might seem harmless --- since the promoted candidate cannot win --- it may reduce the score of one of the approved candidates enough to let a less-preferred candidate win, 
        which is therefore not safe.
    \end{itemize}
\end{proof}

\rmark{This lemma shows that knowing that a single candidate has no chance of winning is generally insufficient to guarantee a safe and profitable manipulation.
However, we show that such knowledge can still be helpful in some cases. In particular, under certain voting rules, knowing that several candidates  are ruled out might help, where the required number of candidates depends on the structure of the specific rule. We illustrate this through two partial results that hold for any \emph{positional voting rule}. 
We prove that knowing that some candidates are ruled out may enable a manipulation that is safe; however, it is unclear whether it is necessarily profitable. 
We leave this as an open question and provide full details below.} 
Specifically,
\Cref{lemma:RAT-and-veto-power-positional-index} establishes a connection based on the first position where the score strictly increases, while \Cref{lemma:RAT-and-veto-power-positional-length} relates the RAT-degree to the length of the largest sequence of identical scores. 

\begin{lemma}\label{lemma:RAT-and-veto-power-positional-index}
    Consider any position voting rule with $s_1 <s_m$.
    Let~$\ell \in \{1,\ldots, m-1\}$ be the smallest position such that $s_{\ell} < s_{\ell +1}$.
    \quad
    Let $t$ be the smallest coalition size with veto power at least $\ell$--- i.e., $v(t)\ge \ell$.
    \rmark{Then, if the number of known-agents is at least $t$, then the manipulator has safe (but not necessarily profitable) manipulation.}
\end{lemma}

\begin{proof}
    Let Alice be the manipulator with true preferences $c_m \succ m_{m-1} \succ \cdots \succ c_1$. Let $K$ be a set of $t$ voters with veto power $\ell$. By definition, they can rule out any set of $\ell$  candidates.  
    Consider the case where they rule out Alice's $\ell$ least-preferred candidates.
    
    Notice that $c_{\ell}$ and $c_{\ell+1}$ are the candidate at position $\ell$ and $\ell+1$ in Alice's truthful preferences, respectively.
    Consider the following manipulation for Alice: switching between $c_{\ell}$ and $c_{\ell+1}$.
    

    First, as $s_{\ell} < s_{\ell+1}$, swapping $c_{\ell}$ and $c_{\ell+1}$ strictly increases $c_{\ell}$'s score and strictly decreases $c_{\ell+1}$'s score.
    This means that the manipulation is safe: promoting $c_{\ell}$ cannot make this candidate win, since it is one of the candidates ruled out by $K$.

    \rmark{The manipulation might be profitable as the competition is only between $c_{\ell+1}$ and candidates that are more preferred by Alice, so decreasing the score of $c_{\ell+1}$ can only benefit her. 
    However, for a manipulation to be profitable, we must show that there exists some profile, consistent with the agent’s knowledge, under which the outcome is strictly improved compared to truthful reporting. In particular, if we could guarantee, for example, the existence of a profile in which there is a tie between $c_{\ell+1}$ and a more preferred candidate, then the manipulation would break the tie in favor of the latter, yielding a strictly better outcome. However, since we only know that certain candidates are ruled out by the agent’s knowledge, it is unclear whether such a profile must necessarily exist\footnote{Assume, for example, that the only profile for the known-agents that rules out the desired candidates (Alice's $\ell$ least-favorite candidates) also rules out candidate $c_{\ell+1}$. In this case, $c_{\ell+1}$ has no chance of winning and therefore the manipulation is not profitable.}.
    }
\end{proof}

\begin{lemma}\label{lemma:RAT-and-veto-power-positional-length}
    Consider any position voting rule with $s_1 <s_m$.
    Let $q \{1,\ldots, m-1\}$ be the length of the largest sequence of identical scores:
    $$
        q = \max_{1 \le j \le m-1} \{\,x \mid s_j = s_{j+1} = \cdots = s_{j+x-1}\,\}.
    $$
    Let $t$ be the smallest coalition size with veto power at least $m-q$ --- i.e., $v(t)\ge m-q$.
    
    Then, if the number of known-agents is at least $t$, then the manipulator has safe (but not necessarily profitable) manipulation.
\end{lemma}

\begin{proof}
Let Alice be the manipulator with true preferences $c_m \succ c_{m-1} \succ \cdots \succ c_1$. 
Let $K$ be a set of $t$ voters with veto power $m-q$. By definition, they can rule out any set of $m-q$ candidates.  
Consider the case where they rule out all but a set of $q$ candidates corresponding to the largest equal-score sequence (if there are several such sets, pick one). 
Formally, let $j$ be the smallest index of that sequence and $C_q := \{c_j, \ldots, c_{j+q-1}\}$; then the known agents rule out all candidates that are not in $C_q$.

Since $s_1 < s_m$, at least one of the two cases must hold:
\begin{itemize}
    \item $j > 1$: 
    
    Consider the following manipulation for Alice: switching between $c_1$ and $c_j$. 
    Notice that $s_{j} > s_{1}$ (by the minimality of the length of the equal-score sequence), thus the manipulation strictly increases $c_1$'s score and strictly decreases $c_j$'s score.

    As the competition is only between the candidates in $C_q$, increasing $c_1$'s score cannot harm Alice (as $c_1$ has no chance of winning), while decreasing $c_j$, which is the least-preferred candidate in $C_q$, can only benefit her. 
    Therefore, the manipulation is safe.

    \rmark{It might also be profitable: if under Alice's truthful report the outcome is a tie between $c_{j}$ and a more-preferred candidate, then this manipulation breaks the tie in favor of her more preferred one.
    It is unclear if such a profile must necessarily exist.}

    \item $j + q - 1 < m$: 
    
    In this case, switching between $c_{j+q-1}$ and $c_m$ is safe and profitable.
    Here, we can conclude that $s_{j+q-1} < s_{m}$. which means that the manipulation strictly increases $c_{j+q-1}$'s score and strictly decreases $c_m$'s score.

    As the competition is only between the candidates in $C_q$, decreasing $c_m$'s score cannot harm Alice (as it has no chance of winning).
    However, increasing $c_{j+q-1}$, which is her most-preferred candidate in $C_q$, can only benefit her. 
    Thus, the manipulation is safe. 
    
    \rmark{The manipulation might be profitable: if under her truthful report the outcome is a tie between $c_{j+q-1}$ and a least-preferred candidate, then this manipulation breaks the tie in favor of her more preferred one.
    Unfortunately, it is unclear whether such a profile must necessarily exist.}
\end{itemize}

\end{proof}


    
    
    




\rmark{This raises the following question: 

\begin{open}\label{open:veto-power}
    Is there a direct relation between Moulin’s veto power and the RAT degree, and if so, how can it be characterized?
\end{open}}


\newcommand{\men}{M}
\newcommand{\women}{W}
\newcommand{\man}{m}
\newcommand{\woman}{w}

\section{Two-sided Matching}\label{sec:matching} 



In this section, we consider mechanisms for two-sided matching. 
Here, the $n$ agents are divided into two disjoint subsets, $\men$ and $\women$, that need to be matched to each other. The most common examples are men and women or students and universities. 
Each agent has a strict preference order over the agents in the other set and being unmatched -- for each $\man \in \men$, an order $\succ_{\man} $ over $\women\cup \{\phi\}$; and for each $\woman \in \women$ an order, $\succ_{\woman}$, over $\men\cup \{\phi\}$. 

A \emph{matching} between $\men$ to $\women$ is a mapping $\mu$ from $\men \cup \women$ to $\men \cup \women \cup \{\phi\}$ such that (1) $\mu(\man) \in \women \cup \{\phi\}$ for each $\man \in \men$, (2) $\mu(\woman) \in \men \cup \{\phi\}$ for each $\woman \in \women$, and (3) $\mu(\man) = \woman$ if and only if $\mu(\woman) = \man$ for any $(\man, \woman) \in \men \times \women$. 
When $\mu(a) = \phi$ it means that agent $a$ is unmatched under $\mu$. 

A mechanism in this context gets the preference orders of all agents and returns a matching.
See \citet{gonczarowski2024structural} for a recent description of the structure of matching mechanisms.

\erel{Consider citing other matching-related papers, such as  \citet{gonczarowski2014manipulation}, and other recent papers. }

\paragraph{Results.}
Our results for this problem are preliminary, so we provide only a brief overview here, with full descriptions and proofs in the appendix. We believe, however, that this is an important problem and that our new definition opens the door to many interesting questions.

We start with the well-known \emph{deferred acceptance} mechanism of \citet{gale1962college}, which always returns a \emph{stable} matching—no agent prefers being unmatched over their assigned match, and there is no pair of agents who would both prefer to be matched to each other over their current assignments.
It is known that \emph{no} stable matching mechanism is truthful for all agents \cite{roth1982economics}.
Indeed, deferred acceptance is known to be truthful only for one side of the market. 
But what happens on the other side? 
Our analysis reveals that the RAT-degree of deferred acceptance is very low: it is at least $1$ and at most $3$. 
The proof of the upper bound relies on \emph{truncation}, where an agent falsely reports preferring to remain unmatched over certain options. We further show that even when agents are required to report complete preferences—thus ruling out this type of manipulation—the RAT-degree is at most $5$.

This raises the following important and interesting question: 
\begin{open}\label{open-matching-1}
    Is there a stable matching mechanism with RAT-degree in $\Omega(n)$?
\end{open}

We also examine the \emph{Boston mechanism} \cite{abdulkadirouglu2003school}, which is a widely used in practice for assigning students to schools. We establish an upper bound of $2$ on its RAT-degree.


\begin{table}[h]
    \centering
    \begin{tabular}{|l|c|c|c|}
    \hline 
    \multirow{2}{*}{Mechanism} & \multicolumn{2}{c|}{RAT-Degree}  & \multirow{2}{*}{Properties}\\
    \cline{2-3}
     & Lower Bound & Upper Bound & \\
    \hhline{|=|=|=|=|}
    Deferred Acceptance (DA) & $ 1$ &$3$ & \multirow{2}{*}{Stable, Truthful for $\men$}\\
    \cline{1-3}
    DA under Complete Preferences & $ 1$ & $5$ &\\
    \hline 
    Boston  & 0 & 2& \\
    \hline
    \end{tabular}
    \caption{\centering Matchings: Summary of Results.
    }
    \label{tab:matchings-sum}
\end{table}

\eden{I think the rest should be in the appendix}

\subsection{Deferred Acceptance (Gale-Shapley)}\label{sec:deferred-acceptance}

The \emph{deferred acceptance} algorithm \cite{gale1962college} is one of the most well-known mechanisms for computing a stable matching. 
In this algorithm, one side of the market --- here, $\men$ --- proposes, while the other side --- $\women$ --- accepts or rejects offers iteratively. For completeness, we provide its description in the appendix.
\begin{toappendix}
\subsection{Deferred Acceptance (Gale-Shapley): descriptions and proofs}
The algorithm proceeds as follows:
\begin{enumerate}
    \item Each $\man \in \men$ proposes to his most preferred alternative according to $\succ_{\man}$ that has not reject him yet and that he prefers over being matched. 

    \item Each $\woman \in W$ tentatively accepts her most preferred proposal according to $\succ_{\woman}$ that she prefers over being matched, and rejects the rest.

    \item The rejected agents propose to their next most preferred choice as in step 1.

    \item The process repeats until no one of the rejected agents wishes to make a new proposal.

    \item The final matching is determined by the last set of accepted proposal.
\end{enumerate}
\end{toappendix}

It is well known that the mechanism is truthful for the proposing side ($\men$) but untruthful for the other side ($\women$).
That is, the agents in $\women$ may have an incentive to misreport their preferences to obtain a better match.
Moreover, there is provably no mechanism for two-sided matching that is truthful for both sides \cite{roth1982economics}.

\er{This section provides a more nuanced analysis of the amount of knowledge required by $\women$ agents to manipulate safely. We focus on a specific agent $w_1\in W$, 
with ranking $m_1 \succ_{w_1} m_2 \succ_{w_1} \cdots $.
There are three kinds of potential manipulations for any $w_1\in W$:}
\begin{enumerate}
\item 
Demoting some $m_j\in M$ from above $\phi$ to below $\phi$ (i.e., claiming that an acceptable partner is unacceptable for her). This is equivalent to reporting only a prefix of the ranking, sometimes called \emph{truncation}  \citep{roth1999truncation,ehlers2008truncation,coles2014optimal}.
\item 
Promoting some $m_j\in M$ from below $\phi$ to above $\phi$
(i.e., claiming that an unacceptable partner is acceptable for her). 
\item 
Reordering some $m_i,m_j\in M$ (i.e., claiming that she prefers $m_i$ to $m_j$ where in fact she prefers $m_j$ to $m_i$).
\end{enumerate}

\begin{lemmarep}
\label{lem:da-truncation}
\er{
In Deferred Acceptance with $k\geq 3$ known agents, there may be a safe and profitable truncation manipulation for $w_1$.
}
\end{lemmarep}
 
\begin{proof}
Suppose the $3$ known agents are as follows:
\begin{itemize}
\item Let $\woman_2 \in \women$ be an agent whose preferences are $\man_2 \succ_{\woman_2} \man_1 \succ_{\woman_2} \cdots $.

\item The preferences of $\man_1$ are $\woman_2 \succ_{\man_1} \woman_1 \succ_{\man_1} \cdots$.

\item The preferences of $\man_2$ are $\woman_1 \succ_{\man_2} \woman_2 \succ_{\man_2} \cdots$.
\end{itemize}
When $\woman_1$ is truthful, the resulting matching includes the pairs $(\man_1, \woman_2)$ and $(\man_2, \woman_1)$, since in this case it proceeds as follows:
\begin{itemize}
\item In the first step, all the agents in $\men$ propose to their most preferred option: $\man_1$ proposes to $\woman_2$ and $\man_2$ proposes to $\woman_1$.

Then, the agents in $\women$ tentatively accept their most preferred proposal among those received, as long as she prefers it to remaining unmatched: 
$\woman_1$ tentatively accepts $\man_2$ since she prefers him over being unmatched, and since $\man_2$ must be her most preferred option among the proposers as $\man_1$ (her top choice) did not propose to her.
Similarly, $\woman_2$ tentatively accepts $\man_1$.

\item In the following steps, more rejected agents in $\men$ might propose to $\woman_1$ and $\woman_2$, but they will not switch their choices, as they prefer $\man_2$ and $\man_1$, respectively. 

Thus, when the algorithm terminates $\man_1$ is matched to $\woman_2$ and $\man_2$ is matched to $\woman_1$, which means that $\woman_1$ is matched to her second-best option.
\end{itemize}

We shall now see that $\woman_1$ can increase her utility by truncating her preference order to $\man_1 \succ'_{\woman_1} \phi \succ'_{\woman_1} \man_2 \succ'_{\woman_1} \cdots $.
The following shows that in this case, the resulting matching includes the pairs $(\man_1, \woman_1)$ and $(\man_2, \woman_2)$, meaning that $\woman_1$ is matched to her most preferred option (instead of her second-best).
\begin{itemize}
\item In the first step, as before, $\man_1$ proposes to $\woman_2$ and $\man_2$ proposes to $\woman_1$.

However, here,  $\woman_1$ rejects $\man_2$ because, according to her false report, she prefers being unmatched over being matched to $\man_2$.
As before, $\woman_2$ tentatively accepts $\man_1$.

\item In the second step, $\man_2$, having been rejected by $\woman_1$, proposes to his second-best choice $\woman_2$.

Since $\woman_2$ prefers $\man_2$ over $\man_1$,  she rejects $\man_1$ and tentatively accepts $\man_2$.

\item In the third step, $\man_1$, having been rejected by  $\woman_2$, proposes to his second-best choice $\woman_1$.

$\woman_1$ now tentatively accepts $\man_1$ since according to her false report, she prefers him over being unmatched.

\item In the following steps, more rejected agents in $\men$ might propose to $\woman_1$ and $\woman_2$, but they will not switch their choices, as they prefer $\man_2$ and $\man_1$, respectively. 

Thus, when the algorithm terminates $\man_1$ will be matched to $\woman_1$ and and $\man_2$ will be matched to $\woman_2$. 
\end{itemize}

Thus, regardless of the reports of the remaining $(n-4)$ remaining (unknown) agents, $\woman_1$ strictly prefers to manipulate her preferences.
\end{proof}

In some settings, it is reasonable to assume that agents always prefer being matched if possible. In such cases, the mechanism is designed to accept only preferences over agents from the opposite set (or equivalently, orders where being unmatched is always the least preferred option). Clearly, under this restriction, truncation is not a possible manipulation.
We prove that even when truncation is not possible, the RAT-degree is bounded.

\begin{lemmarep}
\label{lem:da-switch}
\er{
In Deferred Acceptance without truncation, with $k\geq 5$ known agents, there may be a safe and profitable reorder manipulation for $w_1$.
}
\end{lemmarep}

\begin{proof}
Let us assume that the manipulator is $w_1$, with preference list $m_1 \succ m_2 \succ m_3 \succ \cdots$.

Suppose the $5$ known agents are as follows:
\begin{itemize}
\item Let $\woman_2 \in \women$ be an agent whose preferences are $\man_1 \succ_{\woman_2} \man_2 \succ_{\woman_2} \man_3 \succ_{\woman_2} \cdots $.

\item Let $\woman_3 \in \women$ be an agent whose preferences are $\man_2 \succ_{\woman_3} \man_1 \succ_{\woman_3} \man_3 \succ_{\woman_2} \cdots $.

\item The preferences of $\man_1$ are $\woman_3 \succ_{\man_1} \woman_1 \succ_{\man_1} \woman_2 \succ_{\man_1} \cdots$.

\item The preferences of $\man_2$ are $\woman_1 \succ_{\man_2} \woman_3 \succ_{\man_2} \woman_2 \succ_{\man_2} \cdots$.

\item The preferences of $\man_3$ are $\woman_1 \succ_{\man_2} \woman_3 \succ_{\man_2} \woman_2 \succ_{\man_2} \cdots$.
\end{itemize}

When $\woman_1$ is truthful, the resulting matching includes the pairs $(\man_1, \woman_3)$, $(\man_2, \woman_1)$ and $(\man_3, \woman_2)$, since in this case it proceeds as follows:
\begin{itemize}
\item In the first step, all the agents in $\men$ propose to their most preferred option: $\man_1$ proposes to $\woman_3$, while $\man_2$ and $\man_3$ proposes to $\woman_1$.

Then, the agents in $\women$ tentatively accept their most preferred proposal among those received.
$\woman_1$ tentatively accepts $\man_2$ since he must be her most preferred option among the proposers -- as he is her second-best and her top choice, $\man_1$, did not propose to her; and rejects $\man_3$.
Similarly, $\woman_3$ tentatively accepts $\man_1$.
$\woman_2$ did not get any proposes.  

\item In the second step, $\man_3$, having been rejected by $\woman_1$, proposes to his second-best choice $\woman_3$.

Since $\woman_3$ prefers her current match $\man_1$ over $\man_3$,  she rejects $\man_3$.

\item In the third step, $\man_3$, having been rejected by $\woman_3$, proposes to his third-best choice $\woman_2$.

Since $\woman_2$ does not have a match, she tentatively accepts $\man_3$.

\item In the following steps, more rejected agents in $\men$ - that are not $\man_1, \man_2$ and $\man_3$, might propose to $\woman_1, \woman_2$ and $\woman_3$, but they will not switch their choices, as they can only be least preferred than their current match. 

Thus, when the algorithm terminates $\man_1$ is matched to $\woman_3$, $\man_2$ is matched to $\woman_1$, and $\man_3$ is matched to $\woman_2$,
which means that $\woman_1$ is matched to her second-best option.
\end{itemize}

But if $\woman_1$ swaps $m_2$ and $m_3$ and reports $\man_1 \succ'_{\woman_1} \man_3 \succ'_{\woman_1} \man_2 \succ'_{\woman_1} \cdots$, then the resulting matching includes the pairs $(\man_1, \woman_1)$, $(\man_2, \woman_3)$ and $(\man_3, \woman_2)$, meaning that $\woman_1$ is matched to her most preferred option (instead of her second-best).
\begin{itemize}
\item In the first step, as before, $\man_1$ proposes to $\woman_3$, while $\man_2$ and $\man_3$ proposes to $\woman_1$.

However, here,  $\woman_1$ tentatively accepts $\man_3$ and rejects $\man_2$.
As before, $\woman_3$ tentatively accepts $\man_1$ and $\woman_2$ did not get any proposes.

\item In the second step, $\man_2$, having been rejected by $\woman_1$, proposes to his second-best choice $\woman_3$.

Since $\woman_3$ prefers $\man_2$ over $\man_1$,  she rejects $\man_1$ and tentatively accepts $\man_2$.

\item In the third step, $\man_1$, having been rejected by  $\woman_3$, proposes to his second-best choice $\woman_1$.

$\woman_1$ tentatively accepts $\man_1$ since according to her false report, she prefers him over $\man_3$.

\item In the fourth step, $\man_3$, having been rejected by  $\woman_1$, proposes to his second-best choice $\woman_3$.

$\woman_3$ prefers her current match $\man_2$, and thus rejects $\man_3$.

\item In the fifth step, $\man_3$, having been rejected by  $\woman_3$, proposes to his third-best choice $\woman_2$.

As $\woman_2$ does not have a match, she tentatively accepts $\man_3$.

\item In the following steps, more rejected agents in $\men$ - that are not $\man_1, \man_2$ and $\man_3$, might propose to $\woman_1, \woman_2$ and $\woman_3$, but they will not switch their choices, as they can only be least preferred than their current match. 

Thus, when the algorithm terminates $\man_1$ is matched to $\woman_1$, $\man_2$ is matched to $\woman_3$, and $\man_3$ is matched to $\woman_2$. 
\end{itemize}

Thus, regardless of the reports of the remaining $(n-6)$ remaining (unknown) agents report, $\woman_1$ strictly prefers to manipulate her preferences.
\end{proof}

\begin{remark}
    In the women-proposing variant, $w_1$ gets $m_1$ both when she is truthful and when she manipulates. 

    Hence, A "mixture" variant of DA, which mixes between men-proposing and women-proposing, will not have a higher RAT degree.
\end{remark}

\er{We conjecture that the numbers $3$ and $5$ are tight, but currently have only a weaker lower bound of $1$, which follows from the following lemma:}
\begin{lemmarep}
\er{
In Deferred Acceptance, no agent has a safe manipulation.
}
\end{lemmarep}
\begin{proof}
We consider each of the three possible kinds of manipulations, and show that each of them is not safe for $w_1$.

\paragraph{Demoting some $m_j\in M$ to below $\phi$.}
For some unknown agents' rankings, $m_j$ is the only one who proposes to $w_1$.
Therefore, when $w_1$ demotes $m_j$ she rejects him and remains unmatched, but when she is truthful she is matched to him, which is better for her.

\paragraph{Promoting some $m_j\in M$ to above $\phi$.}
An analogous argument works in this case too.

\paragraph{Reordering some $m_i,m_j\in M$.}
For some unknown agents' rankings, $m_i$ and $m_j$ are the only ones who propose to $w_1$, and they do so simultaneously. Therefore, when $w_1$ is truthful she is matched to the one she prefers, but when she manipulates she rejects him in favor of the less-preferred one.

In all three cases, the manipulation is not safe.
\end{proof}

Combining the above lemmas gives:

\begin{theorem}
\label{prop-def-acc-trunc} 
\label{prop-def-acc-no-trunc}
The RAT-degree of Deferred Acceptance is at least $1$.

It is at most $3$ when truncation is allowed, 
and at most $5$ when truncation is not allowed.
\end{theorem}

\subsection{Boston Mechanism}
The \emph{Boston} mechanism \cite{abdulkadirouglu2003school} is a widely used mechanism for assigning students or schools. 
It is not truthful for both sides. Moreover, we show that it can be safely manipulated with little information.
\begin{toappendix}
\subsection{Boston Mechanism: descriptions and proofs}
The mechanism proceeds in rounds as follows:
\begin{enumerate}
    \item Each $\man \in \men$ proposes to his most preferred alternative according $\succ_{\man}$ that has not yet rejected him and is still available.

    \item Each $\woman \in \women$ (permanently) accepts her most preferred proposal according to $\succ_{\woman}$ and rejects the rest. Those who accept a proposal become unavailable. 

    \item The rejected agents propose to their next most preferred choice as in step $1$.

    \item The process repeats until all agents are either assigned or have exhausted their preference lists.
\end{enumerate}

\end{toappendix}

\begin{lemmarep}
\label{lem:prop-boston}
\er{
In the Boston Mechanism with $k\geq 2$ known agents, there may be a safe and profitable reorder manipulation for some $m_1\in M$.
}
\end{lemmarep}
\begin{proof}
Let $\man_1 \in \men$ be an agent with preferences $\woman_1 \succ_{\man_1} \woman_2 \succ_{\man_1} \cdots$.
Suppose the two known agents are as follows:
\begin{itemize}
\item Let $\man_2 \in \men$ be an agent whose preferences are similar to $m_1$, $\woman_1 \succ_{\man_2} \woman_2 \succ_{\man_2} \cdots$.
\item The preferences of $\woman_1$ are $\man_2 \succ_{\woman_1} \man_1 \succ_{\woman_1} \cdots$.
\end{itemize}

When $\man_1$ reports truthfully, the mechanism proceeds as follows: In the first round, both $\man_1$ and $\man_2$ proposes to $\woman_1$. Since $\woman_1$ prefers $\man_2$, she rejects $\man_1$ and becomes unavailable. Thus, in the second round, $\man_1$ proposes to $\woman_2$.

We prove that $\man_1$ has a safe-and-profitable manipulation: $\woman_2 \succ'_{\man_1} \woman_3 \succ'_{\man_1} \cdots \succ'_{\man_1} \woman_1$.
The manipulation is safe since $\man_1$ never had a chance to be matched with $\woman_1$ (regardless of his report). 
The manipulation is profitable as there exists a case where the manipulation improves $\man_1$’s outcome. Consider the case where the top choice of $\woman_2$ is $\man_1$ and there is another agent, $\man_3$, whose top choice is $\woman_2$. 
Notice that $\woman_2$ prefers $\man_1$ over $\man_3$. 
When $\man_1$ reports truthfully, then in the first round, $\man_3$ proposes to $\woman_2$ and gets accepted, making her unavailable by the time $\man_1$ reaches her in the second round.
However, if $\man_1$ manipulates, he proposes to $\woman_2$ in the first round, and she will accept him (as she prefers him over $\man_3$). This guarantees that $\man_1$ is matched to $\woman_2$, improving his outcome compared to truthful reporting.
\end{proof}

The lemma implies:
\begin{theorem}
    \label{prop-boston}
    The RAT-degree of the Boston mechanism is at most $2$.
\end{theorem}


\section{Discussion and Future Work}\label{sec:discussion}

Our main goal in this paper is to encourage a more quantitative approach to truthfulness that can be applied to various problems. When truthfulness is incompatible with other desirable properties, we aim to find mechanisms that are ``as hard to manipulate as possible'', where hardness is measured by the amount of knowledge required for a safe manipulation. 



\onlyComsoc{\eden{this command is very important! to same somewhere}
\setlength{\parskip}{2pt}}

\paragraph{Avoiding Risk.} Our definition applies only to agents who entirely \emph{avoid} even a small risk.
This definition can be justified based on the well-known behavioral phenomenon called "zero-risk bias"\cite{crosby2021laws,wells2017monetizing}. For example, experiments reported in the paper 'Prospect Theory: An Analysis of Decision under Risk', by \citet{kahneman2013prospect}, indicates that people underweight outcomes that are merely probable, in comparison with outcomes that are obtained with certainty.

Another justification stems from our distinction between two types of risk: the risk due to the randomization of the mechanism, and the risk due to uncertainty about other agents' preferences.
Our risk-avoidance assumption concerns only the second type of risk, for which there is usually no known probability distribution.
Without such knowledge, agents are facing a non-quantifiable unknown risk (often referred to as \emph{ambiguity}).\footnote{This is related to another behavioral phenomenon known as "ambiguity aversion" or "uncertainty aversion".} We believe that it is more reasonable to assume agents avoid this unknown risk altogether as they cannot even reliably distinguish between high and low risks.
%
To compute expected utility with respect to this risk, one would need to adopt a Bayesian framework, which often requires strong and somewhat artificial assumptions about the strategies of other agents. We leave this direction for future work.

\paragraph{Randomized Mechanisms.} This distinction between the two types of risk naturally leads to the topic of randomized mechanisms.
In some settings, RAT-degree is more interesting for deterministic mechanisms, as many impossibility results apply only to deterministic mechanisms.
However, in contexts such as voting, where there are impossibility results for randomized mechanisms too, the RAT-degree is applicable and potentially useful.
In such settings, one can assume that agents avoid the non-quantifiable unknown risk that comes from the uncertainty about other agents' preferences, while treating the known risk induced by the mechanism  construction using standard risk models—such as risk-neutral or risk-averse—aimed at maximizing expected utility.

\paragraph{Utility: Zero vs. Positive.}
Our model assumes that even a small positive gain can have a significant impact on behavior.
For instance, the RAT-degree of the first-price auction is $0$, but increases to $1$ when we introduce a fixed positive discount, even when the discount is arbitrarily small (see \Cref{sec:single-item-auction}).
This assumption is supported by findings in behavioral economics.
For example, the paper 'Zero as a Special Price: The True Value of Free Products', by \citet{shampanier2007zero}, shows that behavior changes significantly when something is free versus when it carries even a tiny cost.
%
A related real-world example is the shift in behavior following the introduction of a small charge on plastic bags in many countries. When bags were free, most people used them without thinking; even a tiny positive cost was sufficient for `nudging' people to bring reusable bags.

\paragraph{Tradeoffs.} 
A particularly interesting future direction is to explore the tradeoffs between a high RAT-degree and other desirable properties, such as fairness. 
For example, in the case of two-sided matching (\Cref{sec:matching}), we know that stability can be achieved with a low RAT-degree of at most $3$, but is impossible to achieve with a RAT-degree of $n$. Where exactly does the boundary lie? Can we characterize the RAT-degree in relation to different fairness properties?

\paragraph{Beyond Worst-Case.} We defined the RAT-degree as a ``worst case'' concept: to prove an upper bound, we find a single example of a safe manipulation. This is similar to the situation with the classic truthfulness notion, where to prove non-truthfulness, it is sufficient to find a single example of a manipulation.
To go beyond the worst case, one could follow relaxations of truthfulness, such as truthful-in-expectation \cite{lavi2011truthful} or strategyproofness-in-the-large \citep{azevedo2019strategy}, and define similarly ``RAT-degree in expectation'' or ``RAT-degree in the large''.

\paragraph{Information About the Known-Agents.}
Our definition assumes that whenever an agent is known, we know their exact preference—that is, the precise preferences $P_i$ they will report from their domain $\domain_i$. 
In practice, however, we often have only partial information on the known-agents; for example, we might know that their preferences belong to a smaller subset of their domain.
Consider \Cref{tab:safe-manip-i} in \Cref{sec:intuitive}. 
Such information still allows us to consider only a strict subset of the columns, leading to a weaker requirement than standard truthfulness.
However, our results show that even under the assumption that we know the exact preferences, identifying the RAT-degree can already be highly nontrivial. For this reason, we only highlight this direction as a promising avenue for future research.

\paragraph{Changing Quantifiers.}
One could argue for a stronger definition requiring that a safe manipulation exists for every possible set of $k$ known-agents, rather than for some set, or similarly for every possible preference profile for the known agents rather than just in some profile.
However, we believe such definitions would be less informative, as in many cases, a manipulation that is possible for some set of $k$ known-agents, is not possible for \emph{any} such set. For example, in the first-price auction with discount (see \Cref{sec:action-first-price-w-discount}), the RAT-degree is $1$ under our definition. But if we required the the knowledge on any agent’s bid would allow manipulation rather than just one, the degree would automatically jump to $n-1$, making the measure far less meaningful.

\paragraph{Combining the "known agents" concept with other notions.} 
We believe that the ``known agents'' approach can be used to quantify the degree to which a mechanism is robust to other types of manipulations (besides safe manipulations),
such as ``always-profitable'' manipulations or ``obvious'' manipulation.
Accordingly, one can define the ``max-min-strategyproofness degree'' or the ``NOM degree'' (see the \extendedVer).



\paragraph{Alternative Information Measurements.}
Another avenue for future work is to study other ways to quantify truthfulness. For example, instead of counting the number of known \emph{agents}, one could count the number of \emph{bits} that an agent should know about other agents' preferences in order to have a safe manipulation. The disadvantage of this approach is that different domains have different input formats, and therefore it would be hard to compare numbers of bits in different domains (see the \extendedVer~ for more details).

\paragraph{Other applications.} The RAT-degree can potentially be useful in any social-choice setting in which truthful mechanisms are not known, or lack other desirable properties. Examples include combinatorial auctions, multiwinner voting, budget aggregation and facility location.

\section{Acknowledgement}
This research is partly supported by the Israel Science Foundation grants 712/20, 3007/24, 2697/22 and 1092/24.
The research of Biaoshuai Tao is supported by the National Natural Science Foundation of China (No. 62472271) and the Key Laboratory of Interdisciplinary Research of Computation and Economics (Shanghai University of Finance and Economics), Ministry of Education.

We sincerely appreciate Alexandros Psomas for his valuable contribution to this work.
We are grateful to Sophie Bade, Avinatan Hassidim, Assaf Romm, Yonatan Aumann, Vincent Conitzer, Jerome Lang, and Reshef Meir for their insights and helpful answers, to Paritosh Verma for his contribution; and to Michael Greinecker for his helpful responses.\footnote{https://economics.stackexchange.com/q/59899/385 and https://economics.stackexchange.com/a/24369/385}

Lastly, we would like to thank the reviewers in EC 2025 and COMSOC 2025 for their helpful comments.

\bibliographystyle{ACM-Reference-Format}
\bibliography{main}

@article{bartholdi1989computational,
  title={The computational difficulty of manipulating an election},
  author={Bartholdi, John J and Tovey, Craig A and Trick, Michael A},
  journal={Social choice and welfare},
  volume={6},
  pages={227--241},
  year={1989},
  publisher={Springer}
}

@article{bartholdi1991single,
  title={Single transferable vote resists strategic voting},
  author={Bartholdi III, John J and Orlin, James B},
  journal={Social Choice and Welfare},
  volume={8},
  number={4},
  pages={341--354},
  year={1991},
  publisher={Springer}
}

@article{walsh2011hard,
  title={Where are the hard manipulation problems?},
  author={Walsh, Toby},
  journal={Journal of Artificial Intelligence Research},
  volume={42},
  pages={1--29},
  year={2011}
}

@article{faliszewski2010ai,
  title={AI’s war on manipulation: Are we winning?},
  author={Faliszewski, Piotr and Procaccia, Ariel D},
  journal={AI Magazine},
  volume={31},
  number={4},
  pages={53--64},
  year={2010}
}

@article{veselova2016computational,
  title={Computational complexity of manipulation: a survey},
  author={Veselova, Yu A},
  journal={Automation and Remote Control},
  volume={77},
  pages={369--388},
  year={2016},
  publisher={Springer}
}

@article{waxman2021manipulation,
  title={Manipulation of k-coalitional games on social networks},
  author={Waxman, Naftali and Hazon, Noam and Kraus, Sarit},
  journal={arXiv preprint arXiv:2105.09852},
  year={2021}
}

@inproceedings{chevaleyre2009compiling,
	title={Compiling the votes of a subelectorate},
	author={Chevaleyre, Yann and Lang, J{\'e}r{\^o}me and Maudet, Nicolas and Ravilly-Abadie, Guillaume},
	booktitle={Proceedings of the 21st International Joint Conference on Artificial Intelligence ({IJCAI-2009})},
	pages={97--102},
	year={2009}
}

@inproceedings{xia2010compilation,
	title={Compilation complexity of common voting rules},
	author={Xia, Lirong and Conitzer, Vincent},
	booktitle={Proceedings of the 24th AAAI Conference on Artificial Intelligence  ({AAAI-2010})},
	pages={915--920},
	year={2010}
}

@inproceedings{karia2021compilation,
	title={Compilation Complexity of Multi-Winner Voting Rules (Student Abstract)},
	author={Karia, Neel and Lang, J{\'e}r{\^o}me},
	booktitle={Proceedings of the 35th AAAI Conference on Artificial Intelligence ({AAAI-2021})},
	pages={15809--15810},
	year={2021}
}

@inproceedings{barrot2017manipulation,
	title={Manipulation of Hamming-based approval voting for multiple referenda and committee elections},
	author={Barrot, Nathana{\"e}l and Lang, J{\'e}r{\^o}me and Yokoo, Makoto},
	booktitle={Proceedings of the 16th Conference on Autonomous Agents and MultiAgent Systems ({AAMAS-2017})},
	pages={597--605},
	year={2017}
}

@inproceedings{lackner2018approval,
	title={Approval-based multi-winner rules and strategic voting},
	author={Lackner, Martin and Skowron, Piotr},
	booktitle={Proceedings of the 27th International Joint Conference on Artificial Intelligence ({IJCAI-2018})},
	pages={340--346},
	year={2018}
}

@inproceedings{lackner2023free,
	title={Free-Riding in Multi-Issue Decisions},
	author={Lackner, Martin and Maly, Jan and Nardi, Oliviero},
	booktitle={Proceedings of the 22nd International Conference on Autonomous Agents and Multiagent Systems ({AAMAS-2023})},
	pages={2040--2048},
	year={2023}
}

@inproceedings{lipton2004approximately,
  title={On approximately fair allocations of indivisible goods},
  author={Lipton, Richard J and Markakis, Evangelos and Mossel, Elchanan and Saberi, Amin},
  booktitle={Proceedings of the 5th ACM Conference on Electronic Commerce},
  pages={125--131},
  year={2004}
}

@article{barman2019fair,
  title={Fair division of indivisible goods among strategic agents},
  author={Barman, Siddharth and Ghalme, Ganesh and Jain, Shweta and Kulkarni, Pooja and Narang, Shivika},
  journal={arXiv preprint arXiv:1901.09427},
  year={2019}
}

@article{Steinhaus48,
	author = {Steinhaus, Hugo},
	title = {The Problem of Fair Division},
	journal = {Econometrica},
	volume = {16},
	number = {1},
	year = {1948},
	pages = {101--104},
}

@article{Steinhaus49,
	author = {Steinhaus, Hugo},
	title = {Sur La Division Pragmatique},
	journal = {Econometrica},
	volume = {17},
	year = {1949},
	pages = {315--319},
}

@inproceedings{tao2022existence,
  title={On existence of truthful fair cake cutting mechanisms},
  author={Tao, Biaoshuai},
  booktitle={Proceedings of the 23rd ACM Conference on Economics and Computation},
  pages={404--434},
  year={2022}
}

@article{dubins1961cut,
  title={How to cut a cake fairly},
  author={Dubins, Lester E and Spanier, Edwin H},
  journal={The American Mathematical Monthly},
  volume={68},
  number={1P1},
  pages={1--17},
  year={1961},
  publisher={Taylor \& Francis}
}

@article{woodall1986note,
  title={A note on the cake-division problem.},
  author={Woodall, Douglas R},
  journal={J. Comb. Theory, Ser. A},
  volume={42},
  number={2},
  pages={300--301},
  year={1986}
}

@article{aleskerov1999degree,
  title={Degree of manipulability of social choice procedures},
  author={Aleskerov, Fuad and Kurbanov, Eldeniz},
  journal={Current trends in Economics: theory and applications},
  pages={13--27},
  year={1999},
  publisher={Springer}
}

@article{andersson2014budget,
  title={Budget balance, fairness, and minimal manipulability},
  author={Andersson, Tommy and Ehlers, Lars and Svensson, Lars-Gunnar},
  journal={Theoretical Economics},
  volume={9},
  number={3},
  pages={753--777},
  year={2014},
  publisher={Wiley Online Library}
}

@article{andersson2014least,
  title={Least manipulable envy-free rules in economies with indivisibilities},
  author={Andersson, Tommy and Ehlers, Lars and Svensson, Lars-Gunnar},
  journal={Mathematical Social Sciences},
  volume={69},
  pages={43--49},
  year={2014},
  publisher={Elsevier}
}

@article{troyan2020obvious,
  title={Obvious manipulations},
  author={Troyan, Peter and Morrill, Thayer},
  journal={Journal of Economic Theory},
  volume={185},
  pages={104970},
  year={2020},
  publisher={Elsevier}
}

@article{ortega2022obvious,
  title={Obvious manipulations in cake-cutting},
  author={Ortega, Josu{\'e} and Segal-Halevi, Erel},
  journal={Social Choice and Welfare},
  volume={59},
  number={4},
  pages={969--988},
  year={2022},
  publisher={Springer}
}

@article{BU2023Rat,
title = {On existence of truthful fair cake cutting mechanisms},
journal = {Artificial Intelligence},
volume = {319},
pages = {103904},
year = {2023},
issn = {0004-3702},
doi = {https://doi.org/10.1016/j.artint.2023.103904},
url = {https://www.sciencedirect.com/science/article/pii/S0004370223000504},
author = {Xiaolin Bu and Jiaxin Song and Biaoshuai Tao}
}

@article{lavi2011truthful,
  title={Truthful and near-optimal mechanism design via linear programming},
  author={Lavi, Ron and Swamy, Chaitanya},
  journal={Journal of the ACM (JACM)},
  volume={58},
  number={6},
  pages={1--24},
  year={2011},
  publisher={ACM New York, NY, USA}
}

@article{azevedo2019strategy,
  title={Strategy-proofness in the large},
  author={Azevedo, Eduardo M and Budish, Eric},
  journal={The Review of Economic Studies},
  volume={86},
  number={1},
  pages={81--116},
  year={2019},
  publisher={Oxford University Press}
}

@article{brams2006better,
  title={Better ways to cut a cake},
  author={Brams, Steven J and Jones, Michael A and Klamler, Christian and others},
  journal={Notices of the AMS},
  volume={53},
  number={11},
  pages={1314--1321},
  year={2006}
}

@inproceedings{slinko2008nondictatorial,
  title={Nondictatorial social choice rules are safely manipulable},
  author={Slinko, Arkadii and White, Shaun},
  booktitle={Proceedings of the Second International Workshop on Computational Social Choice (COMSOC-2008)},
  pages={403--414},
  year={2008}
}

@article{slinko2014ever,
  title={Is it ever safe to vote strategically?},
  author={Slinko, Arkadii and White, Shaun},
  journal={Social Choice and Welfare},
  volume={43},
  pages={403--427},
  year={2014},
  publisher={Springer}
}

@inproceedings{hazon2010complexity,
  title={Complexity of safe strategic voting},
  author={Hazon, Noam and Elkind, Edith},
  booktitle={Algorithmic Game Theory: Third International Symposium, SAGT 2010, Athens, Greece, October 18-20, 2010. Proceedings 3},
  pages={210--221},
  year={2010},
  organization={Springer}
}

@article{coles2014optimal,
  title={Optimal truncation in matching markets},
  author={Coles, Peter and Shorrer, Ran},
  journal={Games and Economic Behavior},
  volume={87},
  pages={591--615},
  year={2014},
  publisher={Elsevier}
}

@article{roth1999truncation,
  title={Truncation strategies in matching markets—in search of advice for participants},
  author={Roth, Alvin E and Rothblum, Uriel G},
  journal={Econometrica},
  volume={67},
  number={1},
  pages={21--43},
  year={1999},
  publisher={Wiley Online Library}
}

@article{ehlers2008truncation,
  title={Truncation strategies in matching markets},
  author={Ehlers, Lars},
  journal={Mathematics of Operations Research},
  volume={33},
  number={2},
  pages={327--335},
  year={2008},
  publisher={INFORMS}
}

@article{gale1962college,
  title={College admissions and the stability of marriage},
  author={Gale, David and Shapley, Lloyd S},
  journal={The American Mathematical Monthly},
  volume={69},
  number={1},
  pages={9--15},
  year={1962},
  publisher={Taylor \& Francis}
}

@inproceedings{gonczarowski2014manipulation,
  title={Manipulation of stable matchings using minimal blacklists},
  author={Gonczarowski, Yannai A},
  booktitle={Proceedings of the fifteenth ACM conference on Economics and computation},
  pages={449--449},
  year={2014}
}

@inproceedings{gonczarowski2024structural,
  title={Structural Complexities of Matching Mechanisms},
  author={Gonczarowski, Yannai A and Thomas, Clayton},
  booktitle={Proceedings of the 56th Annual ACM Symposium on Theory of Computing},
  pages={455--466},
  year={2024}
}

@article{ausubel2006lovely,
  title={The lovely but lonely Vickrey auction},
  author={Ausubel, Lawrence M and Milgrom, Paul},
  journal={Combinatorial auctions},
  volume={17},
  number={3},
  pages={22--26},
  year={2006},
  publisher={Citeseer}
}

@book{nisan2007algorithmic,
  title={Algorithmic Game Theory},
  author={Nisan, Noam and Roughgarden, Tim and Tardos, Eva and Vazirani, Vijay},
  year={2007},
  pages={238},
  publisher={Cambridge University Press}
}

@book{krishna2009auction,
  title={Auction theory},
  author={Krishna, Vijay},
  year={2009},
  publisher={Academic press}
}

@inproceedings{braverman2016interpolating,
  title={Interpolating between truthful and non-truthful mechanisms for combinatorial auctions},
  author={Braverman, Mark and Mao, Jieming and Weinberg, S Matthew},
  booktitle={Proceedings of the Twenty-Seventh Annual ACM-SIAM Symposium on Discrete Algorithms},
  pages={1444--1457},
  year={2016},
  organization={SIAM}
}

@TechReport{regret2018Fernandez,
  author={Fernandez, Marcelo Ariel},
  title={{Deferred acceptance and regret-free truth-telling}},
  year=2018,
  month=Jun,
  institution={The Johns Hopkins University,Department of Economics},
  type={Economics Working Paper Archive},
  url={https://ideas.repec.org/p/jhu/papers/65832.html},
  number={65832},
  doi={},
}

@article{gibbard1973manipulation,
  title={Manipulation of voting schemes: a general result},
  author={Gibbard, Allan},
  journal={Econometrica: journal of the Econometric Society},
  pages={587--601},
  year={1973},
  publisher={JSTOR}
}

@article{satterthwaite1975strategy,
  title={Strategy-proofness and Arrow's conditions: Existence and correspondence theorems for voting procedures and social welfare functions},
  author={Satterthwaite, Mark Allen},
  journal={Journal of economic theory},
  volume={10},
  number={2},
  pages={187--217},
  year={1975},
  publisher={Elsevier}
}

@inproceedings{chen2011profitable,
  title={How profitable are strategic behaviors in a market?},
  author={Chen, Ning and Deng, Xiaotie and Zhang, Jie},
  booktitle={Algorithms--ESA 2011: 19th Annual European Symposium, Saarbr{\"u}cken, Germany, September 5-9, 2011. Proceedings 19},
  pages={106--118},
  year={2011},
  organization={Springer}
}

@article{chen2022incentive,
  title={Incentive ratio: A game theoretical analysis of market equilibria},
  author={Chen, Ning and Deng, Xiaotie and Tang, Bo and Zhang, Hongyang and Zhang, Jie},
  journal={Information and Computation},
  volume={285},
  pages={104875},
  year={2022},
  publisher={Elsevier}
}

@inproceedings{li2024bounding,
  title={Bounding the Incentive Ratio of the Probabilistic Serial Rule},
  author={Li, Bo and Sun, Ankang and Xing, Shiji},
  booktitle={Proceedings of the 23rd International Conference on Autonomous Agents and Multiagent Systems},
  pages={1128--1136},
  year={2024}
}

@article{bei2025incentive,
  title={The incentive guarantees behind nash welfare in divisible resources allocation},
  author={Bei, Xiaohui and Tao, Biaoshuai and Wu, Jiajun and Yang, Mingwei},
  journal={Artificial Intelligence},
  pages={104335},
  year={2025},
  publisher={Elsevier}
}

@inproceedings{tao2024fair,
  title={Fair and Almost Truthful Mechanisms for Additive Valuations and Beyond},
  author={Tao, Biaoshuai and Yang, Mingwei},
  booktitle={International Conference on Web and Internet Economics},
  year={2024},
  organization={Springer}
}

@inproceedings{cheng2022tight,
  title={Tight incentive analysis on sybil attacks to market equilibrium of resource exchange over general networks},
  author={Cheng, Yukun and Deng, Xiaotie and Li, Yuhao and Yan, Xiang},
  booktitle={Proceedings of the 23rd ACM Conference on Economics and Computation},
  pages={792--793},
  year={2022}
}

@article{cheng2019improved,
  title={An improved incentive ratio of the resource sharing on cycles},
  author={Cheng, Yu-Kun and Zhou, Zi-Xin},
  journal={Journal of the Operations Research Society of China},
  volume={7},
  pages={409--427},
  year={2019},
  publisher={Springer}
}

@article{chen2024regret,
  title={Regret-free truth-telling in school choice with consent},
  author={Chen, Yiqiu and M{\"o}ller, Markus},
  journal={Theoretical Economics},
  volume={19},
  number={2},
  pages={635--666},
  year={2024},
  publisher={Wiley Online Library}
}

@article{abdulkadirouglu2003school,
  title={School choice: A mechanism design approach},
  author={Abdulkadiro{\u{g}}lu, Atila and S{\"o}nmez, Tayfun},
  journal={American economic review},
  volume={93},
  number={3},
  pages={729--747},
  year={2003},
  publisher={American Economic Association}
}

@article{nisan2002communication,
  title={The communication complexity of efficient allocation problems},
  author={Nisan, Noam and Segal, Ilya},
  journal={Draft. Second version March 5th},
  pages={173--182},
  year={2002}
}

@article{grigorieva2006communication,
  title={The communication complexity of private value single-item auctions},
  author={Grigorieva, Elena and Herings, P Jean-Jacques and M{\"u}ller, Rudolf and Vermeulen, Dries},
  journal={Operations Research Letters},
  volume={34},
  number={5},
  pages={491--498},
  year={2006},
  publisher={Elsevier}
}

@inproceedings{Communication2019Branzei,
author = {Br\^{a}nzei, Simina and Nisan, Noam},
title = {Communication Complexity of Cake Cutting},
year = {2019},
isbn = {9781450367929},
publisher = {Association for Computing Machinery},
address = {New York, NY, USA},
url = {https://doi.org/10.1145/3328526.3329644},
doi = {10.1145/3328526.3329644},
booktitle = {Proceedings of the 2019 ACM Conference on Economics and Computation},
pages = {525},
numpages = {1},
location = {Phoenix, AZ, USA},
series = {EC '19}
}

@inproceedings{Babichenko2019communication,
author = {Babichenko, Yakov and Dobzinski, Shahar and Nisan, Noam},
title = {The communication complexity of local search},
year = {2019},
isbn = {9781450367059},
publisher = {Association for Computing Machinery},
address = {New York, NY, USA},
url = {https://doi.org/10.1145/3313276.3316354},
doi = {10.1145/3313276.3316354},
booktitle = {Proceedings of the 51st Annual ACM SIGACT Symposium on Theory of Computing},
pages = {650–661},
numpages = {12},
location = {Phoenix, AZ, USA},
series = {STOC 2019}
}

@article{roth1982economics,
  title={The economics of matching: Stability and incentives},
  author={Roth, Alvin E},
  journal={Mathematics of operations research},
  volume={7},
  number={4},
  pages={617--628},
  year={1982},
  publisher={INFORMS}
}

@article{peters2022robust,
  title={Robust rent division},
  author={Peters, Dominik and Procaccia, Ariel D and Zhu, David},
  journal={Advances in Neural Information Processing Systems},
  volume={35},
  pages={13864--13876},
  year={2022}
}

@inproceedings{amanatidis2017truthful,
  title={Truthful allocation mechanisms without payments: Characterization and implications on fairness},
  author={Amanatidis, Georgios and Birmpas, Georgios and Christodoulou, George and Markakis, Evangelos},
  booktitle={Proceedings of the 2017 ACM Conference on Economics and Computation},
  pages={545--562},
  year={2017}
}

@article{shampanier2007zero,
  title={Zero as a special price: The true value of free products},
  author={Shampanier, Kristina and Mazar, Nina and Ariely, Dan},
  journal={Marketing science},
  volume={26},
  number={6},
  pages={742--757},
  year={2007},
  publisher={INFORMS}
}

@book{crosby2021laws,
  title={The Laws of Wealth},
  author={Crosby, Daniel},
  year={2021},
  publisher={Jaico Publishing House}
}

@book{wells2017monetizing,
  title={Monetizing your data: A guide to turning data into profit-driving strategies and solutions},
  author={Wells, Andrew Roman and Chiang, Kathy Williams},
  year={2017},
  publisher={John Wiley \& Sons}
}

@incollection{kahneman2013prospect,
  title={Prospect theory: An analysis of decision under risk},
  author={Kahneman, Daniel and Tversky, Amos},
  booktitle={Handbook of the fundamentals of financial decision making: Part I},
  pages={99--127},
  year={2013},
  publisher={World Scientific}
}

@article{even1984note,
  title={A note on cake cutting},
  author={Even, Shimon and Paz, Azaria},
  journal={Discrete Applied Mathematics},
  volume={7},
  number={3},
  pages={285--296},
  year={1984},
  publisher={Elsevier}
}

@inproceedings{conitzer2011dominating,
  title={Dominating manipulations in voting with partial information},
  author={Conitzer, Vincent and Walsh, Toby and Xia, Lirong},
  booktitle={Proceedings of the AAAI conference on artificial intelligence},
  volume={25},
  number={1},
  pages={638--643},
  year={2011}
}

@inproceedings{Reijngoud2012,
author = {Reijngoud, Annemieke and Endriss, Ulle},
title = {Voter response to iterated poll information},
year = {2012},
isbn = {0981738125},
publisher = {International Foundation for Autonomous Agents and Multiagent Systems},
address = {Richland, SC},
booktitle = {Proceedings of the 11th International Conference on Autonomous Agents and Multiagent Systems - Volume 2},
pages = {635–644},
numpages = {10},
location = {Valencia, Spain},
series = {AAMAS '12}
}

@inproceedings{endriss2016strategic,
  author = {Endriss, Ulle and Obraztsova, Svetlana and Polukarov, Maria and Rosenschein, Jeffrey S.},
    title = {Strategic voting with incomplete information},
    year = {2016},
    isbn = {9781577357704},
    publisher = {AAAI Press},
    booktitle = {Proceedings of the Twenty-Fifth International Joint Conference on Artificial Intelligence},
    pages = {236–242},
    numpages = {7},
    location = {New York, New York, USA},
    series = {IJCAI'16}
}

@article{Moulin1982vetoPower,
 ISSN = {00129682, 14680262},
 URL = {http://www.jstor.org/stable/1912535},
 author = {H. Moulin},
 journal = {Econometrica},
 number = {1},
 pages = {145--162},
 publisher = {[Wiley, Econometric Society]},
 title = {Voting with Proportional Veto Power},
 urldate = {2025-11-06},
 volume = {50},
 year = {1982}
}

@article{caragiannis2019unreasonable,
  title={The unreasonable fairness of maximum Nash welfare},
  author={Caragiannis, Ioannis and Kurokawa, David and Moulin, Herv{\'e} and Procaccia, Ariel D and Shah, Nisarg and Wang, Junxing},
  journal={ACM Transactions on Economics and Computation (TEAC)},
  volume={7},
  number={3},
  pages={1--32},
  year={2019},
  publisher={ACM New York, NY, USA}
}

\appendix
\newpage

\section{Open Questions: Summary}

\begin{itemize}
    \item  Indivisible Goods Allocations

    \begin{itemize}
        \item \Cref{open-alloc-1}: 
        Is there a polynomial time EF1 goods allocation rule with RAT-degree $n-1$?

        \item \Cref{open:envy-elimination}:
        What is the RAT-degree of The Envy-Cycle Elimination Mechanism by \citet{lipton2004approximately}?
    
        \item \Cref{open-alloc-2}: What is the RAT-degree of Volatile Priority Round Robin?

        \item \Cref{open:allocation-mnw}: Are there any variants of MNW with a RAT degree larger than 1?
    \end{itemize}

    \item Cake Cutting 
    \begin{itemize}
        \item \Cref{open-cake-1}:Is there a proportional connected cake-cutting rule with RAT-degree at least $2$?

        \item \Cref{open-cake-2}:
    What is the RAT-degree of the Even-Paz algorithm?

        \item  \Cref{open-cake-3}:
    What is the RAT-degree of the maximum Nash welfare mechanism? 
    \end{itemize}

    \item Single-Winner Voting
    \begin{itemize}
        \item \Cref{open:voting-1}:
    Does there exist a non-dictatorial voting rule that satisfies the participation criterion (i.e. does not suffer from the no-show paradox),  with RAT-degree larger than $\ceil{n/2}+1$? 

    \item \Cref{open:voting-2}: What are the exact RAT-degrees of known Condorcet-consistent rules?

    \item \Cref{open:veto-power}: Is there a direct relation between Moulin’s veto power and the RAT degree, and if so, how can it be characterized?
    \end{itemize}

    \item Two-Sided Matching

    \begin{itemize}
        \item \Cref{open-matching-1}:
    Is there a stable matching mechanism with RAT-degree in $\Omega(n)$?
    \end{itemize}
\end{itemize}

\section{Related Work: Extended Discussion}\label{apx:related}

To maintain consistency with the proceedings version, we have kept the appendix, even though its content is now integrated into the main body of the paper. For the full discussion, please refer to the Related Work section in this version.

\end{document}